\documentclass[11pt]{article}
\usepackage[margin=1in]{geometry}
\usepackage{amsmath}
\usepackage{amsthm}
\usepackage{amssymb}
\usepackage{algorithm}
\usepackage{caption}
\usepackage{subcaption}
\usepackage{color}
\usepackage[english]{babel}
\usepackage{graphicx}
\usepackage{grffile}
\usepackage{wrapfig,epsfig}
\usepackage{epstopdf}
\usepackage{url}
\usepackage{color}
\usepackage{epstopdf}
\usepackage{algpseudocode}
\usepackage{hyperref}
\usepackage[T1]{fontenc}
\usepackage{bbm}
\usepackage{comment}
\usepackage{dsfont}
\usepackage{thm-restate}
\usepackage{bm}
\usepackage{tcolorbox}
\usepackage{makecell}
\usepackage{tikz}
\usetikzlibrary{shapes,arrows.meta,positioning}

\newtheorem{theorem}{Theorem}[section]
 
\newtheorem{lemma}[theorem]{Lemma}
\newtheorem{definition}[theorem]{Definition}

\newtheorem{conjecture}[theorem]{Conjecture}

\newtheorem{fact}[theorem]{Fact}
\newtheorem{remark}[theorem]{Remark}
\newtheorem{claim}[theorem]{Claim}

\newcommand{\bvi}{\mathbf{v}^{(i)}}

\newcommand{\wh}{\widehat}
\newcommand{\wt}{\widetilde}
\newcommand{\ov}{\overline}
\newcommand{\eps}{\epsilon}

\newcommand{\R}{\mathbb{R}}

\renewcommand{\varepsilon}{\epsilon}
\renewcommand{\tilde}{\wt}
\renewcommand{\hat}{\wh}

\renewcommand{\eps}{\epsilon}

\newcommand{\bx}{\mathbf{x}}
\newcommand{\ba}{\mathbf{a}}
\newcommand{\by}{\mathbf{y}}
\newcommand{\bz}{\mathbf{z}}
\newcommand{\bb}{\mathbf{b}}

\newcommand{\bu}{\mathbf{u}}

\newcommand{\bc}{\mathbf{c}}
\newcommand{\bv}{\mathbf{v}}

\newcommand{\boldm}{\mathbf{m}}

\newcommand{\ttop}{{(t)}}

\newcommand{\bX}{\mathbf{X}}
\newcommand{\bY}{\mathbf{Y}}

\newcommand{\bH}{\mathbf{H}}
\newcommand{\bM}{\mathbf{M}}
\newcommand{\bA}{\mathbf{A}}
\newcommand{\bN}{\mathbf{N}}
\newcommand{\bC}{\mathbf{C}}
\newcommand{\bB}{\mathbf{B}}
\newcommand{\bD}{\mathbf{D}}
\newcommand{\bG}{\mathbf{G}}
\newcommand{\bJ}{\mathbf{J}}
\newcommand{\bI}{\mathbf{I}}
\newcommand{\bU}{\mathbf{U}}
\newcommand{\bV}{\mathbf{V}}

\newcommand{\bW}{\mathbf{W}}
\newcommand{\bw}{\mathbf{w}}
\newcommand{\bdelta}{\boldsymbol\delta}
\newcommand{\omv}{\mathsf{OMv}}
\newcommand{\tauo}{\tau_{\mathsf{ols}}}

\DeclareMathOperator*{\E}{{\mathbb{E}}}
\DeclareMathOperator*{\Var}{{\bf {Var}}}

\DeclareMathOperator{\poly}{poly}

\DeclareMathOperator{\nnz}{nnz}

\DeclareMathOperator{\diag}{diag}

\DeclareMathOperator{\new}{new}

\DeclareMathOperator{\im}{Im}

\makeatletter
\newcommand*{\RN}[1]{\expandafter\@slowromancap\romannumeral #1@}
\makeatother

\title{The Complexity of Dynamic Least-Squares Regression}

\date{}
\author{Shunhua Jiang \\ 
\vspace{.2in}
Columbia University\\ \texttt{sj3005@columbia.edu}
\and Binghui Peng \\ 
\vspace{.2in}
Columbia University\\
\texttt{bp2601@columbia.edu} 
\and Omri Weinstein\\ The Hebrew University \\  and Columbia University\\ \texttt{omri@cs.columbia.edu}
}

\begin{document}

\maketitle
\begin{abstract}
    
We settle the complexity of dynamic least-squares regression (LSR),\footnote{The two main results of this manuscript (Theorems \ref{thm:main_UB_informal} and \ref{thm_low_acc_LB_informal}) are new compared to the preliminary  arXiv version: The lower bound for fully-dynamic $\ell_2$-regression holds against \emph{constant} relative accuracy algorithms, and our partially-dynamic upper bound handles \emph{adaptive} adversarial updates.}  where rows and labels $(\mathbf{A}^{(t)}, \mathbf{b}^{(t)})$ can be \emph{adaptively} inserted and/or deleted, and the goal is to efficiently maintain an $\epsilon$-approximate solution to $\min_{\mathbf{x}^{(t)}} \| \mathbf{A}^{(t)} \mathbf{x}^{(t)} - \mathbf{b}^{(t)} \|_2$ for all $t\in [T]$. We prove sharp  separations ($d^{2-o(1)}$ vs. $\sim d$) between the amortized update time of: (i) Fully vs. Partially dynamic $0.01$-LSR; (ii) High vs. low-accuracy LSR in the partially-dynamic (insertion-only) setting.  

Our lower bounds  follow from a  gap-amplification reduction---reminiscent of  \emph{iterative refinement}---from the exact version of the \emph{Online Matrix Vector} Conjecture (OMv) [HKNS15],  to \emph{constant} approximate OMv  over the reals, where the $i$-th online product $\mathbf{H}\mathbf{v}^{(i)}$ only needs to be computed to $0.1$-\emph{relative} error. All previous fine-grained reductions from OMv to its approximate versions only show hardness for inverse polynomial approximation $\epsilon = n^{-\omega(1)}$ (additive or multiplicative) . This result is of independent interest in fine-grained complexity and for the investigation of the OMv Conjecture, which is still widely open.    

\end{abstract}
\setcounter{page}{0}
\thispagestyle{empty}
\newpage

\section{Introduction}
\label{sec:intro}
The problem of least-squares regression (LSR)  dates back to Gauss in 1821 \cite{GaussLS}, and is the backbone 
of high-dimensional statistical inference \cite{HFT_book_01}, signal processing \cite{rg75},  
convex optimization \cite{b15}, control theory \cite{k90}, 
network routing \cite{ls14,m13} and machine learning \cite{cortes1995support}. 
Given an overdetermined ($n\gg d$) linear system $\bA \in \R^{n \times d}, \bb \in \R^n $, 
the goal is to find the solution vector $\bx$ that minimizes the mean squared error (MSE)
\begin{align} \label{eq_LS_reg} 
\min_{\bx\in \R^n} \| \bA\bx - \bb \|_2.
\end{align}
The exact closed-form solution is given by the well-known Normal equation
$\bx^\star = (\bA^\top \bA)^{-1}\bA^\top \bb$, 
which requires $O(nd^2)$ time to compute using naive matrix-multiplication, or $O(nd^{\omega-1}) \approx O(n d^{1.37})$ time using fast matrix-multiplication (FMM) \cite{s69} for the current FMM exponent of $\omega \approx 2.37$ \cite{le14, aw21}.

Despite the elegance and simplicity of this closed-form solution, in practice the latter runtime is often too slow, especially in modern data analysis applications where both the dimension of the 
feature space ($d$) and the size of datasets ($n$) are overwhelmingly large. 
A more modest objective in attempt to circumvent this computational overhead, is to seek an $\eps$-accurate solution that satisfies 
\begin{align} \label{eq_APX_LS_solution} 
\|\bA \bx- \bb\|_2 \leq (1+\eps) \min_{\bx \in \R^d} \|\bA \bx - \bb \|_2 .
\end{align} 
A long line of work on sketching \cite{cw17} and sampling \cite{clmmps15}, combined with iterative linear-system solvers (preconditioned gradient descent), culminated in \emph{high-precision} algorithms that run in close to input-sparsity time 
$\wt{O}(\nnz(A)\log(1/\eps)+ d^{\omega})$\footnote{We use $\wt{O}$ to hide $\poly \log(d)$ factors.} 
for the offline problem \cite{s06, cw17, nn13, cherapanamjeri2023optimal}.
This is essentially optimal in the realistic setting $d\ll n$. 

\vspace{+2mm}
{\bf \noindent  Dynamic Least Squares \ \ } 
Many of the aforementioned applications of LSR, both in theory and practice, involve data that is continually changing, either by nature or by design. In such  applications, it is desirable to avoid recomputing the LSR solution from scratch, and instead maintain the solution \eqref{eq_APX_LS_solution} dynamically, under insertion and/or deletions of rows and labels $(\ba^\ttop, \beta^\ttop)$. 
The most compelling and realistic dynamic model is that of \emph{adaptive} row updates, where the algorithm  is required to be correct  against an adaptive adversary that chooses the next update as a function of previous outputs of the algorithm. This is a much stronger notion of  dynamic algorithms than the traditional \emph{oblivious} model, where the sequence of updates is chosen \emph{in advance}. 
This distinction is perhaps best manifested in 
\emph{iterative randomized} algorithms, where the input to the next iteration depends on the output of the previous iteration and hence on the internal randomness of the algorithm, and traditional sketching algorithms generally break under the stronger adaptive setting. \cite{hw13, bjwy22, hkm+22,cohen2022robustness}. 

Another aspect of dynamic LSR  is whether data is partially or fully dynamic -- some applications inherently involve incremental updates (row-insertions), whereas others require both insertion and deletions. 

\bf Fully Dynamic LSR. \rm 
In the \emph{fully dynamic} setting, rows $(\ba^\ttop, \beta^\ttop)$ can be adaptively inserted or deleted, and the goal is to minimizes the \emph{amortized update time} required to maintain an $\eps$-approximate solution to \eqref{eq_LS_reg} in each iteration. 
Variants of dynamic LSR show up in many important applications and iterative optimization methods, from 
Quasi-Newton \cite{pilanci2017newton} and interior-point methods (IPM) \cite{cls21}, 
the (matrix) multiplicative-weight updates framework \cite{AHK06}, iteratively-reweighted least squares \cite{Law61} to mention a few.   
We stress that some of these variant involve application-specific restrictions on the LSR updates (e.g., small $\ell_2$-norm or sparsity), whereas we study here the problem in full generality. 

\bf Partially Dynamic LSR. \rm \; 
In the \emph{partially-dynamic} setting, rows $(\ba^\ttop, \beta^\ttop)$ can only be inserted, and the goal is again to minimize the amortized update time for inserting a row.
This incremental setting is more natural in 
 control theory and dynamical linear systems \cite{Pla50, k60}, and in modern deep-learning applciations, in particular \emph{continual learning} \cite{pkp+19,cpp22}, where the 
goal is to finetune a neural network over arrival of new training data, \emph{without} training from scratch. 

\

The textbook solution for dynamic LSR, which handles general adaptive row updates, and dates back to \emph{Kalman's filter} and the recursive least-squares framework \cite{k60},  
is to apply 
\emph{Woodbury's identity}, 
which can implement each row update in $O(d^2)$ worst-case time \cite{k60}, and maintain the Normal equation exactly.
Interestingly, in the \emph{oblivious, partially dynamic} (insertion-only) setting, one can do much better --  \cite{cmp20} gave a streaming algorithm, based on \emph{online row-sampling}, which  maintains a subset of only $\tilde{O}(d/\eps^2)$ rows of $\bA^{(t)}$, and provides an $\eps$-approximate LSR solution at any given time (the algorithm of \cite{cmp20} does \emph{not} yield an efficient data structure, but we will show a stronger result implying this).  Very recently, this algorithm was extended to handle \emph{adaptive} incremental updates \cite{bhm+21}, at the price of a substantial increase in the number of sampled rows $\tilde{O}(d^2\kappa^2/\eps^2)$  (and hence the update time also increases), where $\kappa$ is the condition number of the inputs, which could scale polynomially with the total number of rounds $T$.

Our first result is an efficient partially-dynamic low-accuracy LSR data structure, against adaptive row-insertions, whose update time depends only logarithmically on the condition number:   

\begin{theorem}[Faster Adaptive Row Sampling] \label{thm:main_UB_informal}
For any accuracy parameter $0 < \eps < 1/8$, there is a randomized dynamic data structure which, 
with probability at least $0.9$, maintains an $\eps$-approximate LSR solution under adaptive row-insertions, 
simultaneously \underline{for all} iterations $t\in[T]$, with total update time 
\[ O\left(\nnz(\bA^{(T)}) \log(T) + \epsilon^{-4} d^5 \log^4(\frac{\sigma_{\max}}{\sigma_{\min}}) \log^3(T)\right), \]
where $\sigma_{\max}$ ($\sigma_{\min}$) is the maximum (minimum) singular value over all input matrices $[\bA^\ttop, \bb^\ttop]$ for $t\in [T]$.
\end{theorem}
For constant approximations ($\epsilon = 0.1$), Theorem~\ref{thm:main_UB_informal} almost matches the fastest static sketching-based solution, up to polylogarithmic terms and the additive term. When $T \gg d$, this theorem shows that amortized update time of our algorithm is $\tilde{O}(d)$. A key sub-routine of our algorithm is an improved analysis of the online leverage score sampling that reduces the number of rows from $\tilde{O}(d^2\kappa^2/\eps^2)$ to  $\tilde{O}(d^2 \log(\kappa)/\eps^2)$, where $\kappa:=\sigma_{\max}/\sigma_{\min}$.

\

Our main result is that, by contrast, in the \emph{fully-dynamic} setting, Kalman's classic approach is essentially optimal, even for maintaining a \emph{constant} approximate LSR solution, assuming the \emph{Online Matrix-Vector} ($\omv$) Conjecture \cite{hkns}: 

\begin{theorem}[Lower Bound for Fully-Dynamic LSR, Informal] \label{thm_low_acc_LB_informal}
There is an adaptive sequence of $T = \poly(d)$ row insertions and deletions, such that any dynamic data structure that maintains an $0.01$-approximate LSR solution, has amortized update time at least $\Omega(d^{2-o(1)})$ per row, under the $\omv$ Conjecture. 
\end{theorem}

Recall that the $\omv$ Conjecture \cite{hkns} postulates that computing \bf \emph{exact} \rm Boolean matrix-vector products, of a fixed $n\times n$ Boolean matrix $\bH$ with an \emph{online} sequence of vectors $\mathbf{v}^{(1)}, \ldots, \mathbf{v}^{(n)}$, one-at-a-time ($\bH\bv^{(i)}$),  requires $n^{3-o(1)}$ time (in sharp contrast to the offline batch setting, where this can be done using FMM in $n^\omega\ll n^3$ time, see Section \ref{sec:fully} for more details). 

Theorem~\ref{thm_low_acc_LB_informal} follows from a gap-amplification reduction from exact $\omv$ to \emph{approximate $\omv$} {\em over the reals}, asserting that the $\omv$ problem remains hard even when the online matrix-vector products $\bH\bvi$ need only be approximated to within \emph{constant relative} accuracy (i.e., $\| \mathbf{y}^{(i)} - \bH\bvi \|_2 \leq 0.1 \|\bH\bvi\|_2$) against adaptive sequences, see Theorem \ref{thm:hard-omv-approx}. 
All previous fine-grained reductions from $\omv$ to its approximate versions only show hardness for inverse-polynomial error $\eps = n^{-\omega(1)}$, see further discussion in the related work section. 

We believe Theorem \ref{thm:hard-omv-approx} may be useful for proving or refuting the $\omv$ Conjecture itself, both because constant relative approximation brings the problem closer to the realm of dimensionality-reduction tools (which only work in the low-accuracy, oblivious regime), and on the other hand, since the adaptive nature of our construction is necessary for relating the $\omv$ conjecture to more established fine-grained conjectures (3SUM, SETH, OV): 
Indeed, 
one reason for the lack of progress the $\omv$ Conjecture is that almost all known reductions in the fine-grained complexity literature are non-adaptive \cite{W18survey}, meaning that they apply equally to online and offline queries, and hence are futile for atacking the $\omv$ Conjecture (one exception is the adaptive reduction of \cite{WW18} for triangle detection).
We remark that adaptivity of the vectors $\bvi$'s is crucial for the proof of Theorems \ref{thm_low_acc_LB_informal} and \ref{thm:hard-omv-approx}, but also natural:   Iterative optimization algorithms and linear-system solvers, in particular first-order (Krylov) methods for quadratic minimization, are based on \emph{iterative refinement} of the residual error, hence the new error vector is a function of previous iterates \cite{wilkinson1994rounding,hestenes1952methods,akps19}. 
In fact, the proof of Theorems \ref{thm_low_acc_LB_informal} and \ref{thm:hard-omv-approx} is inspired precisely by this idea, see the technical overview below.

\

Finally, we prove a similar $d^{2-o(1)}$ amortized lower bound for \emph{high-accuracy} data structures   in the insertion-only setting, which shows that the accuracy  of our data structure from Theorem~\ref{thm:main_UB_informal} cannot be drastically improved:  

\begin{theorem}[Hardness of High-Precision Partially-Dynamic LSR, Informal] \label{thm_exact_LB_informal}
Assuming the $\omv$ Conjecture, any dynamic data structure that maintains an $\eps=1/\poly(T, d)$-approximate solution for the partially dynamic LSR over $T = \poly(d)$ iterations, must have $\Omega(d^{2-o(1)})$ amortized update time per iteration.
\end{theorem}

The three above theorems provide a rather complete characterization of the complexity of dynamic least squares regression.

\subsection{Related work}

\paragraph{Fine-grained complexity} The $\omv$ conjecture \cite{hkns} has originally been proposed as a unified approach to prove conditional lower bound for dynamic problems. It has broad applications to dynamic algorithms \cite{d16,bks17,jns19,lr21,jx22} and it is still widely open \cite{lw17, ckl18,aggs22,hs22}.
A few prior works \cite{acss20, cs17,bis17,bcis18} have shown fine-grained hardness of related matrix problems (e.g. kernel-density estimation, empirical risk minimization), based on the strong Exponential Time Hypothesis (SETH, see  \cite{ip01} and references therein). 
In contrast to Theorems \ref{thm_low_acc_LB_informal}, \ref{thm:hard-omv-approx}, all these works only establish hardness for exact or polynomially-small precision (i.e., $\eps=d^{-\omega(1)}$).

\paragraph{Least-squares regression in other models of computation}
The problem of (static) least-squares regression has a long history in TCS \cite{ac06,cw17,nn13,clmmps15,acw17}. Using dimensionality-reduction techniques (sketching or sampling) to precondition the input matrix, and running (conjugate) gradient descent, one can obtain an $\eps$-approximation solution in input-sparsity $\tilde{O}(\nnz(\bA)\log(1/\eps) + d^{\omega})$, see \cite{w14,w21} for a comprehensive survey.

The LSR problem has also been studied in different models, we briefly review here the most relevant literature. 
In the streaming model, \cite{cw09} gives the (tight) {\em space} complexity of $\tilde{\Theta}(d^2/\eps)$ when entries of the input are subject to changes.
In the online model, the input data arrives in an online streaming and \cite{cmp20} proposes the online row sampling framework, which stores $\tilde{O}(d/\eps^2)$ rows and maintain an $\eps$-spectral approximation of the input.
\cite{bdm+20} generalizes the guarantee to the {\em sliding window} model (among other numerical linear algebra tasks), where data still comes an online stream and but only the most recent updates form the underlying data set.

The focus of all aforementioned model is on the space (or the number of rows), instead of computation time, and they only work against an oblivious adversary.
Initiated by \cite{bjwy22}, a recent line of work \cite{bjwy22, hkm+22,wz22} aims to make streaming algorithm works against an adaptive adversary. 
As noted by \cite{workshop2021}, most existing results are for scalar output and it is an open question when the output is a large vector.
\cite{bhm+21} is most relevant to us and studies the online row sampling framework \cite{cmp20} (among other importance sampling approaches) against adaptive adversary, and prove it maintains an $\eps$-spectral approximation when storing $\tilde{O}(d^2\kappa^2/\eps^2)$ rows, where $\kappa = \sigma_{\max}/\sigma_{\min}$ is the condition number of the input. 
A key part of our algorithm is to give an improved analysis and reduce the number to $\tilde{O}(d^2\log (\kappa)/\eps^2)$.

Finally, we note the problem of {\em online regression} has been studied in the online learning literature \cite{h19}, where the goal is to minimize the total {\em regret}. this is very different from ours in that the main bottleneck  is {\em information-theoretic}, whereas the challenge in our loss-minimization problem is purely computational.

\paragraph{Comparison to the inverse-maintenance data structure in IPM} Similar dynamic regression problems have been considered in the literature of interior point methods (IPM) for solving linear programs (LP) \cite{cls21,b20,lsz19,jswz21,ls14,blss20,bll+21}. There the problem is to maintain $(\bA^\top \bW \bA)^{-1} \bA^\top \bW \bb$ for a slowly-changing diagonal matrix $\bW$. The aforementioned papers use sampling and sketching techniques to accelerate the amortized cost per iteration. The inverse-maintenance data structures in the IPM literature are solving a similar but incomparable dynamic LSR problem -- the updates in the IPM setting are adaptive \emph{fully dynamic} (i.e., general low-rank updates), and cannot recover the linear $\tilde{O}(d)$ update time of 
our data structure for \emph{incremental} row-updates (Theorem \ref{thm:main_UB_informal}).


\section{Technical Overview}

In this section we provide a high-level overview of Theorems \ref{thm:main_UB_informal} and \ref{thm_low_acc_LB_informal}. 

\subsection{Lower bound for fully dynamic LSR}
We start from the lower bound in Theorem \ref{thm_low_acc_LB_informal} for fully dynamic $\eps$-LSR, where we prove that  a dynamic data structure  with 
truly sub-quadratic $d^{2-\Omega(1)}$ amortized update time, even for constant approximation $\eps = 0.01$, would break the $\omv$ Conjecture. The key challenge in this proof is that the $\omv$ Conjecture itself only asserts the hardness for \emph{exact} matrix-vector products (over the boolean semiring). Our reduction proceeds in a few steps, where a key intermediate step is introducing the \emph{online projection} problem:

\begin{restatable}[Online projection]{definition}{Onlineprojection}
\label{def:online-projection} In the online projection problem, the input is a fixed orthonormal matrix $\bU \in \R^{d\times d_1}$ ($d_1\in [d]$), and a sequence of vectors $\bz^{(1)}, \ldots, \bz^{(T)}$ that arrives in an online stream. The goal is to compute the projection of $\bz^\ttop$ onto the column space of $\bU$, i.e., $\bU \bU^\top \bz^{(t)}$, at each iteration $t\in [T]$, before $\bz^{(t+1)}$ is revealed. 
\end{restatable}

For notation convenience, we write $\bz = \bz_{\bU} + \bz_{\bU_{\perp}}$ where $\bz_{\bU}$ is the projection onto $\bU$ and $\bz_{\bU_{\perp}}$ is the projection onto the orthogonal space $\bU_{\perp} \in \R^{d\times (d-d_1) }$.

\subsubsection{Hardness of online projection}
We first prove $1/\poly(d)$-hardness of online projection via reduction from $\omv$. The $\omv$ conjecture asserts the hardness of matrix-vector multiplication $(\|\bH\bz^\ttop\|)$ over Boolean semi-ring, and it is easy to see that the lower bound continues to hold when (1) the matrix $\bH$ is positive semidefinite (PSD), (2) the computation is over real, and (3) one allows $1/d^2$ error, i.e., the output $\by^\ttop$ only needs to satisfy $\|\by^\ttop - \bH\bz^\ttop\|_2 \leq O(1/d^2)$ when one normalizes $\|\bH\|_2 = 1$ and $\|\bz^\ttop\|_2 = 1$.

The online projection problem is clearly easier than arbitrary (PSD) matrix-vector multiplications, and we prove the reverse direction is also true. That is, one can (approximately) simulate a matrix-vector product with $O(\log d)$ projection queries. Given a PSD matrix $\bH$, we first perform the eigenvalue decomposition $\bH = \bU \Sigma \bU^\top$ at the preprocessing step, where $\Sigma =\diag(\lambda_1, \ldots, \lambda_d)$ is a diagonal matrix. 
We perform a {\em binary division trick} over the spectral of $\bH$. Let $S_j \subseteq [d]$ include all column indices $i\in [d]$, such that the $j$-th significant bit of $\lambda_{i}$ is non-zero. Let $\bU(j) \in \R^{d\times |S_j|}$ take columns of $\bU$ from $S_j$, then $\bH\bz^{(t)} = \sum_{j=1}^{O(\log d)}\frac{1}{2^j} \cdot \bz_{\bU(j)}^\ttop \pm O(1/d^2)$, i.e., one can obtain an $O(1/d^2)$ approximation of $\bH\bz^{(t)}$ by querying $O(\log d)$ online projection instances, with precision $O(1/d^2)$.

\subsubsection{Hardness amplification} 
Our next step is to boost the hardness of approximation from $O(1/d^2)$ to some constant. In particular, we prove the online projection problem is hard even one only needs 
$
\|\by^\ttop - \bU\bU^\top \bz^\ttop\|_2 \leq \alpha \|\bU\bU^\top \bz^\ttop\|_2 + \beta,
$
where $\alpha = 1/3$ is the multiplicative error and $\beta = 1/d^3$ is a small additive error.

Given any vector $\bz$, to obtain an $O(1/d^2)$ approximation of $\bz_{\bU}$, we set up two online projection instances, $\mathbb{P}_\bU$ and $\mathbb{P}_{\bU_{\perp}}$, and we assume $\mathbb{P}_{\bU}$ (resp.~$\mathbb{P}_{\bU_\perp}$) returns an $(\alpha,\beta)$-approximation to the projection onto $\bU$ (resp.~$\bU_{\perp}$).
A natural idea is to query $\mathbb{P}_{\bU_{\perp}}$ and obtain 
\[
\bw = \mathbb{P}_{\bU_{\perp}}(\bz) = \bz_{\bU_\perp} + \bdelta \quad \text{where the error term} \quad \|\bdelta\|_2 \leq \alpha \|\bz_{\bU_{\perp}}\|_2  + \beta.
\]
Subtracting $\bw$ and considering $\bz' = \bz - \bw$, the orthogonal component decreases by a factor of $\alpha$ (i.e., $\|\bz_{\bU_\perp}'\|_2 \leq \alpha \|\bz_{\bU_\perp}\|_2 + \beta$) and one hopes to repeat it for $O(\log d)$ times to remove the orthogonal component (almost) completely. 
However, the projection component also gets contaminated, i.e., $\bz_{\bU}' = \bz_{\bU} - \bdelta_{\bU}$. Hence, we need to further ``purify'' $\bdelta$ and ensure $\bdelta_{\bU} \approx 0$. 
We obtain it by another $O(\log d)$ iterations of refinement.\footnote{This might sound circular at a first glance because our original goal is to remove $\bz_{\bU_\perp}$ and we reduce it to remove $\bdelta_{\bU}$. The difference is that it is fine to change $\bdelta_{\bU_\perp}$ by a small multiplicative factor when removing $\bdelta_{\bU}$.}

\paragraph{Final reduction} Our final reduction proceeds in $R = O(\log d)$ rounds and each round further contains $K = O(\log d)$ iterations.
\begin{itemize}
\item {\bf Outer loop.} For each round $r \in [R]$, we wish to find $\bw_{r}$ such that  
(1) $\bw_{r}$ has a negligible component in $\bU$, i.e., $\bw_{r, \bU} \approx 0$, and  
(2) $\bw_{r}$ is $\alpha'$-approximate to $\bz_{r}$ in the direction of $\bU_{\perp}$, i.e., $\|\bw_{r, \bU_\perp} - \bz_{r, \bU_\perp}\|_2 \leq \alpha' \|\bz_{r, \bU_\perp}\|_2$ for some constant $\alpha' < 1$. 
By taking $\bz_{r+1} = \bz_{r} - \bw_{r}$, one can prove that the $\bz_{r,\bU}$ component does not change and the orthogonal component $\bz_{r, \bU_\perp}$ decreases by a factor of $\alpha'$. Repeating for $R = O(\log n)$ would be sufficient.
\item {\bf Inner loop.} Within round $r$, recall we first invoke the projection $\mathbb{P}_{\bU_{\perp}}$ and obtain  
$\bw_{r, 0} = \mathbb{P}_{\bU_{\perp}}(\bz_r)$. 
In order to remove $\bw_{r, 0, \bU}$, we query the projection $\mathbb{P}_{\bU}$ and obtain 
$\by_{r, 1} = \mathbb{P}_{\bU}(\bw_{r,0})$, and $\bw_{r, 1} = \bw_{r, 0} - \by_{r,1}$. 
We have the guarantee that $\|\bw_{r, 1, \bU}\|_2 \leq \alpha \|\bw_{r, 0, \bU}\|_2$. Repeat the above step for $K = O(\log n)$ iterations, we have $\bw_{r, K, \bU}\approx 0$. We also need to control the component in $\bU_{\perp}$. We can show that $\bw_{r, K, \bU_{\perp}} = \bw_{r, 0,\bU_{\perp}} - \sum_{k=1}^{K-1}\by_{r, k, \bU_{\perp}}$, where the second term consists of a geometric decreasing sequence with rate $\alpha$, and one has $\bw_{r, K, \bU_{\perp}} = (1 \pm O(\alpha))\bz_{r, \bU_\perp}$.
\end{itemize}

We note the above reduction is adaptive in nature, because the query depends heavily on the algorithm's previous outputs.

\subsubsection{Reduction to fully dynamic LSR}
The final step is to reduce $(\alpha, \beta)$-online projection to fully dynamic $\eps$-LSR, with the following choice of parameters $\alpha = 1/3, \eps= 0.01$ and $\beta = 1/d^3$. A natural first attempt is to set the initial feature matrix $\bA^{(0)} = \bU_{\perp}^{\top} \in \R^{(d-d_1)\times d}$ and the labels $\mathbf{b}^{(0)} = \frac{1}{\sqrt{d}}\mathbf{1}_{d-d_1}$. This is an under-constrained linear system. Let 
$
\bx^{*} = (\bU_{\perp}\bU_{\perp}^\top)^{\dagger}\bU_{\perp}\mathbf{b}^{(0)} = \frac{1}{\sqrt{d}}\sum_{j=1}^{d-d_1}\bU_{\perp, j}
$
be the normal equation -- this is the solution with the least $\ell_2$ norm. Suppose we wish to project $\bz$ onto $\bU$, then one can insert a new row of $(\bz, 10)$ and the Normal equation becomes 
\[
\bx^{*}_{\new} = \bx^{*} +  \frac{10 - \langle \bz_{\bU_{\perp}}, \bx^{*} \rangle}{\|\bz_{\bU}\|_2^2} \bz_{\bU}.
\]
If the $\eps$-approximate solution $\bx'$ returned by the algorithm is indeed close to $\bx^{*}$, we can obtain a scaled version of $\bz_{\bU}$ by computing $\bx' - \bx^{*} \approx \bx^{*}_{\new} - \bx^{*} \propto \bz_{\bU}$. 

Unfortunately, there are infinitely many optimal solutions and an algorithm does not need to output the normal form solution.
For example, an algorithm could remember a random direction $\bv$ that is orthogonal to $\bU_{\perp}$ (at the preprocessing step) and run binary search on $\bx^{*} + \xi \cdot \bv$ to resolve the new constraint $\langle \bz , \bx \rangle = 10$. It only requires $O(d)$ time and returns an exact solution.

\paragraph{The importance of regularization} The above issue seems to be inherent of an under-constrained linear system.
To resolve it, we consider the ridge regression instead and add a small regularization term $\lambda \|\bx\|_2$ for $\lambda = 1/d^{40}$. 
Our reduction starts with $\|\bU_{\perp}^\top \bx - \mathbf{1}_{d-d_2}\|_2 + \lambda \|\bx\|_2$, and inserts/deletes the row $(\bz , 10)$ to compute the projection of $\bz$. For ridge regression, $\bx^{*}_{\new}$ is actually not the optimal solution (instead, it is very close to the unique optimal solution) but we would prove an $\eps$-approximate solution $\bx'$ needs to be very close to $\bx_{\new}^{*}$. Therefore, one can retrieve an $(\alpha, \beta)$-approximate projection from $\bx'$ and $\bx^{*}$. 

\paragraph{Missing technical consideration} We outline a few missing details of the above argument. First, the above argument (i.e., $\bx'$ is close to $\bx^{*}_{\new}$) goes through only if $\bz_{\bU}$ is not too small (e.g., $\|\bz_{\bU}\|_2 \geq 1/d^4$). We need to efficiently test the norm $\|\bz_{\bU}\|_2$ and output $\mathbf{0}$ when it is too small. 
Second, even if $\bx_{\new}^{*}$ and $\bx'$ are close, we can only obtain a scaled version of $\bz_{\bU}$. It is not oblivious to determine the right ``scale'' because of the (constant) approximation error. Instead, we run a line search and output the minimizer of $\arg\min_{\xi}\|\bz - \xi\cdot(\bx' - \bx^{*})\|$, we prove that it gives good approximation to $\bz_{\bU}$.

\subsection{Algorithm for partially dynamic LSR}
Next we provide an overview of our algorithm in Theorem \ref{thm:main_UB_informal} for partially dynamic LSR (with row insertions only). 
Let $\boldm^\ttop = (\ba^\ttop, \beta^{(t)})$ be the $t$-th row and $\bM^\ttop$ be the input matrix that concatenates these rows.
Our approach follows the online row sampling framework \cite{clmmps15, cmp20, bdm+20}: 
When a new row arrives, we sample and keep the new row with probability (approximately) proportional to the {\em online leverage score} $\tauo^\ttop:= (\boldm^\ttop)^\top ((\bM^\ttop)^{\top} \bM^\ttop)^{-1}\boldm^\ttop$. 
We output the closed-form solution on the sampled rows: It is an $\eps$-approximate solution of LSR as long as the sampled matrix $\wt{\bM}^\ttop$ is an $\eps$-spectral approximation to the input matrix $\bM^\ttop$.

\paragraph{Warm up: oblivious adversary}
If the algorithm faces an oblivious adversary, then \cite{cmp20} proves that keeping $O(d\log(\frac{\sigma_{\max}}{\min})/\eps^2)$ rows is enough for  $\eps$-spectral approximation.
It remains to bound the computation time.
Note a direct computation of the online leverage score takes $O(d^2)$ time per row-insertion, which gives no benefit over the classic Kalman's approach. In order to accelerate this computation, we use a JL-embedding trick to compress the matrix $(\bM^\ttop)^\top\bM^\ttop$ (note a similar trick has been used in \cite{ss11, blss20}) and it reduces the computation time from $O(d^2)$ to $O(d)$ per update.

\paragraph{Adversarial robustness of online leverage score sampling}  We need a counterpart of \cite{cmp20} for the more challenging adaptive adversary.
The recent work of \cite{bhm+21} made a first step toward adversarially robust row-sampling. However, comparing to the oblivious setting, their algorithm increases the number of sampled rows 
by a factor of $d (\frac{\sigma_{\max}}{\sigma_{\min}})^2$. Note that it has a polynomial dependence on the condition number $\frac{\sigma_{\max}}{\sigma_{\min}}$, which could be as large as $\poly(dT)$.\footnote{Indeed, if the input are drawn from isotropic Gaussian $\mathcal{N}(0, \mathbf{I}_d)$ but with one direction removed, then the condition number can be as large as $T$.}

\cite{bhm+21} considers an $\epsilon$-net over the unit vectors in $\R^d$, and for any $\bx$ in the $\epsilon$-net, they use Freedman's inequality to prove that the sampled matrix $\tilde{\bM}^{(t)}$ satisfies that $\|\tilde{\bM}^{(t)} \bx\|_2 \approx \|\bM^{(t)} \bx\|_2$ (this brings an $O(d)$ overhead using a union bound). In order to apply Freedman's inequality, they need an estimate of $\|\bM^{(t)} \bx\|_2$, which is unknown in advance since the rows of $\bM^{(t)}$ are chosen adaptively. \cite{bhm+21} directly bounds this norm by the singular values: $\sigma_{\min} \|\bx\|_2 \leq \|\bM^{(t)} \bx\|_2 \leq \sigma_{\max} \|\bx\|_2$, resulting in the additional $(\frac{\sigma_{\max}}{\sigma_{\min}})^2$ overhead.

We provide a new analysis that overcomes this limitation. Similar to \cite{bhm+21}, we also take a union bound over the $\epsilon$-net of unit vectors $\bx \in \R^d$ to reduce to the scalar case. The key difference is that when applying Freedman's inequality, we instead consider $O(\frac{\sigma_{\min}}{\sigma_{\max}})$ truncated martingales that each guesses the correct value of $\|\bM^{(t)} \bx\|_2$, and becomes 0 once the guess becomes inaccurate. We prove that each truncated martingale concentrates according to Freedman's inequality, and since one of the guesses must be close to the true value of $\|\bM^{(t)} \bx\|_2$, taking a union bound over these $O(\frac{\sigma_{\min}}{\sigma_{\max}})$ martingales results in only an $O(\log(\frac{\sigma_{\max}}{\sigma_{\min}}))$ overhead. We believe this can also be used to improve other importance sampling schemes in \cite{bhm+21}, which we leave for future work. 

\paragraph{Robustness of JL estimation} Finally, we also need to prove the JL trick is adversarially robust. To this end, we renew the JL sketch for each sampled row. 
The sampling probability computed using the JL estimate is always an overestimate of the true online leverage score. Consequently, whenever our algorithm omits a row, the ideal algorithm using the exact online leverage scores also omits that row. This means the randomness of the JL matrix is not leaked until a new row is sampled, at which point we refresh the JL matrix.

\section{Preliminary}
\label{sec:formulation}

\paragraph{Notations}
Let $[n]=\{1,2,\cdots,n\}$ and $[n_1: n_2] = \{n_1, \ldots, n_2\}$. 
For any $x, y \in \R$, we write $x = y \pm \eps$ if $|x - y| \leq \eps$, for two vectors $\bx, \by \in \R^d$, $\bx = \by \pm \eps$ means $\|\bx - \by\|_2 \leq \eps$.
For any matrix $\bA \in \R^{n\times d}$ ($n \geq d$), let $\sigma_1(\bA)\geq \cdots \geq \sigma_d(\bA) \geq 0$ be its singular value, and $\kappa(\bA) = \frac{\sigma_1(\bA)}{\sigma_d(\bA)}$ be the condition number. When $\bA$ is a symmetric matrix, we use $\lambda_1(\bA) \geq \cdots \geq \lambda_d(\bA)$ to denote its eigenvalue.
Let $\nnz(\bA)$ be the number of non-zero entries of a matrix $\bA$.
We use $\|\cdot\|$ to denote the spectral norm, i.e. $\|\bA\| = \max_{\bx\in \R^d, \|\bx\|_2 =1}\|\bA \bx\|_2$ and $\|\cdot\|_{\mathsf{F}}$ to denote the Frobenius norm. We use $\ker[\bA]$ and $\im[\bA]$ to denote the kernel space and the column space of $\bA$. We use $\bA_{i,*}$ and $\bA_{*,j}$ to denote the $i$-th row and the $j$-th column of $\bA$. For two sets $S \subseteq [n], R \subseteq [d]$, we use $\bA_{S,*}$ and $\bA_{*,R}$ to denote the submatrix of $\bA$ obtained by taking the rows in $S$ or taking the columns in $R$.
The $d$-dimensional identity matrix is denoted as $\bI_d$, and we use $\mathbf{1}_d$ (resp. $\mathbf{0}_d$) to denote the all one (resp. all zero) vectors. For a vector $\bv \in \R^d$, denote $\diag(\bv) \in \R^{d \times d}$ as a diagonal matrix whose diagonal entries are $\bv$.

\subsection{Model}\label{sec:model}
We formally define the problem of fully dynamic least-squares regression.
\begin{definition}[Fully dynamic least-squares regression]\label{def:fully_dynamic_l2_regression}
Let $d$ be an integer, $T = \poly(d)$ be the total number of updates and $\eps \in [0,0.5)$ be the precision.
In the fully dynamic least-squares regression problem:
\begin{itemize}
    \item The data structure is given a matrix $\bA^{(0)} \in \R^{(d+1) \times d}$ and a vector $\bb^{(0)} \in \R^{n}$ in the preprocessing phase.
    \item For each iteration $t \in [T]$, the algorithm receives one of the following two updates:
    \begin{itemize}
    \item Incremental update: The update is $(\ba, \beta) \in \R^d \times \R$, and the matrix and the vector are updated to be $\bA^\ttop := [(\bA^{(t-1)})^{\top}, \ba]^{\top}$ and $\bb^\ttop := [(\bb^{(t-1)})^{\top}, \beta]^{\top}$.
    \item Decremental update: The update is a row index $i$, and the matrix $\bA^\ttop$ is $\bA^{(t-1)}$ with its $i$-th row deleted, and the vector $\bb^\ttop$ is $\bb^{(t-1)}$ with its $i$-th entry deleted.
    \end{itemize}
\end{itemize}
We say an algorithm solves $\epsilon$-approximate fully dynamic least squares regression if it outputs an $\epsilon$-approximate solution $\bx^{(t)} \in \R^d$ at every iteration $t \in [T]$:
    \[
    \|\bA^{(t)} \bx^{(t)} - \bb^{(t)}\|_2 \leq (1+\eps) \min_{\bx \in \R^d} \|\bA^{(t)} \bx - \bb^{(t)} \|_2.
    \]
\end{definition}

We use $\bM^{(t)} := [\bA^{(t)}, \bb^{(t)}]$ to denote the concatenation of input feature and their labels. 

We can similarly define a partially dynamic least-squares regression problem that allows insertion or deletion updates only. In our paper, we only study the incremental model.\footnote{In the streaming literature, it is also called the online model. We call it incremental model, to emphasize that the computation cost is the major consideration (instead of space, or the number of rows been kept).}

\begin{definition}[Partially dynamic least-squares regression, incremental update]
A partially dynamic least-squares regression is formalized similarly as Definition \ref{def:fully_dynamic_l2_regression}, with incremental update only. Moreover, we assume the least singular value of $\bM^{(0)}$ is at least $\sigma_{\min} > 0$, and the largest singular value of $\bM^{(T)}$ is at most $\sigma_{\max}$. 
\end{definition}

\begin{remark}[Singular value, incremental model]
We assume the least (and largest) singular value of data matrix $\bM^\ttop$ is bounded, this is standard in the literature (see \cite{cmp20, bhm+21}), and one should think of $1/\poly(T,d) \leq \sigma_{\min} \leq \sigma_{\max} \leq \poly(T, d)$. 
In practice, we can always add a polynomially small regularization term to ensure the smallest singular value of $\bM^{(0)}$ is at least $1/\poly(T, d)$. The largest singular value of $\bM^{(T)}$ is bounded by $\poly(T, d)$ as long as each entry has polynomially bounded value. 
\end{remark}

\begin{remark}[Preprocessing]
In the definition of both fully and partially dynamic LSR, we assume the problem is initialized with a full rank data matrix $\bM^{(0)} \in \R^{(d+1)\times(d+1)}$. This is wlog if one allows polynomial preprocessing time.
\end{remark}

Next, we state the (standard) notion of the {\em oblivious} adversary and the {\em adaptive} adversary.
\begin{definition}[Adversary model]
Two adversary models are of consideration:
\begin{itemize} 
\item {\bf Oblivious adversary.} For an oblivious adversary, the insertion/deletion updates are independent of algorithm's output. 
\item {\bf Adaptive adversary.} An adaptive adversary could choose the new update base on the algorithm's previous output. That is, at the $t$-th update, the new row $(\ba, \beta)$ of an insertion or the index of a deletion update could be a function of $\bx^{(1)}, \ldots, \bx^{(t-1)}$.
\end{itemize}
\end{definition}

\subsection{Mathematical tools}
\label{sec:basic_numerical_linear_algebra}

A positive semidefinite (PSD) matrix $\bA\in \R^{d\times d}$ is symmetric and satisfies $\bx^{\top}\bA\bx\geq 0$ for all $\bx\in \R^{d}$.
We write $\bA \succeq 0$ to denote that $\bA$ is PSD, and we write $\bB \succeq \bA$ to denote that $\bB - \bA$ is PSD. 

\begin{definition}[Spectral approximation]
For two symmetric matrices $\bA, \wt{\bA} \in \R^{n \times n}$, we say that $\wt{\bA}$ and $\bA$ are $\eps$-spectral approximations of each other (denoted as $\wt{\bA} \approx_{\epsilon} \bA$) if
\begin{align*}
    (1 - \epsilon) \cdot \bA \preceq \wt{\bA} \preceq (1 + \epsilon) \cdot \bA.
\end{align*}
\end{definition}

\paragraph{Online leverage scores}  
In the incremental model, rows arrive in online fasion and the online leverage score of the $t$-th row $\boldm^\ttop$ equals
\begin{align*}
    \tauo^\ttop:= (\boldm^\ttop)^\top (((\bM)^{(t-1)})^\top \bM^{(t-1)})^{-1} \boldm^\ttop.
\end{align*}

Online leverage scores satify the following property. Its proof is delayed to Appendix~\ref{sec:pre-app}.
\begin{fact}[Property of online leverage score]\label{fact:online_leverage_score}
For any $t\in [T]$, we have
\[
 \boldm^\ttop (\boldm^\ttop)^{\top}  \preceq \tauo^\ttop \cdot  ((\bM)^{(t-1)})^\top \bM^{(t-1)}.
\]
\end{fact}

Least squares regressions can be solved approximately using spectral approximations. See Section~2.5 of \cite{w14} for details.
\begin{lemma}[Approximate least squares regression from spectral approximation]\label{lem:approx_l2_regression_from_spectral_approx}
Given $\bA \in \R^{n \times d}$ and $\bb \in \R^n$, define a matrix $\bM := [\bA, \bb] \in \R^{n \times (d+1)}$. Let $\bD \in \R^{n' \times n}$ be a subspace embedding that satisfies $\bM^{\top} \bD^{\top} \bD \bM \approx_{\epsilon} \bM^{\top} \bM$.
Define $\bx \in \R^d$ to be
\begin{align*}
    \bx := \arg \min_{\bx' \in \R^d} \|\bD \bA \bx' - \bD \bb\|_2.
\end{align*}
Then with probability at least $1 - \delta$, $\bx$ satisfies
\begin{align*}
\|\bA \bx - \bb\|_2 \leq (1 + \epsilon) \min_{\bx' \in \R^d} \|\bA \bx' - \bb\|_2.
\end{align*}
\end{lemma}

We note the standard least squares regression has a closed-form solution.
\begin{fact}[Closed-form formula for least squares regression]
\label{fact:closed-form}
Let $n \geq d$ be two integers. For any matrix $\bA \in \R^{n \times d}$ with rank $d$, and any vector $\bb \in \R^d$, the vector 
$\bx^* := \bA^{\dagger} \cdot \bb = (\bA^{\top} \bA)^{-1} \bA^{\top} \cdot \bb$ satisfies
\begin{align*}
    \|\bA \bx^* - \bb\|_2 = \min_{\bx \in \R^d} \|\bA \bx - \bb\|_2.
\end{align*}
\end{fact}

\section{Lower bound for fully dynamic LSR}
\label{sec:fully}
In this section we prove that the fully dynamic LSR requires $\Omega(d^{2-o(1)})$ time per update to solve to constant accuracy under the $\omv$ conjecture.

\begin{theorem}[Hardness for fully dynamic LSR, formal version of Theorem \ref{thm_low_acc_LB_informal}]
\label{thm:lower-full}
Let $\gamma > 0$ be any constant. Let $d$ be the input dimension and $T = \poly(d)$ be the total number of update. Let $\eps < 0.01$ be any constant, assuming the $\omv$ conjecture is true, then any algorithm that maintains an $\eps$-approximate solution for fully dynamic least-squares regression problem requires amortized running time at least $\Omega(d^{2 -\gamma})$.
\end{theorem}

The $\omv$ conjecture was originally proposed by \cite{hkns}. In this paper, we work on the standard Word RAM model with word size $O(\log n)$.
\begin{conjecture}[$\omv$ conjecture, \cite{hkns}]
\label{conj:omv}
Let $\gamma > 0$ be any constant. Let $d$ be an integer and $T \geq d$. 
Let $\bB\in \{0,1\}^{d \times d}$ be a Boolean matrix. A sequence of Boolean vectors $\bz^{(1)}, \ldots, \bz^{(T)} \in \{0,1\}^d$ are revealed one after another, and an algorithm solves the $\omv$ problem if it returns the Boolean matrix-vector product $\bB \bz^{(t)} \in \R^d$ after receiving $\bz^\ttop$ at the $t$-th step.
The conjectures states that there is no algorithm that solves the $\omv$ problem using $\poly(d)$ preprocessing time and $O(d^{2-\gamma})$ amortized query time, and has an error probability $\leq 1/3$.
\end{conjecture}

The $\omv$ conjecture asserts the hardness of solving online Boolean matrix-vector product \emph{exactly}.
In order to prove Theorem~\ref{thm:lower-full}, it would be convenient to work with real-valued matrix-vector products. 
We prove that the same lower bound holds for (well-conditioned) PSD matrices, while allowing polynomially small error.  
The following result is a standard and its proof can be found in Appendix \ref{sec:fully-app}.
\begin{lemma}[Hardness of approximate real-valued OMv]
\label{lem:omv-real}
Let $\gamma > 0$ be any constant. Let $d$ be a sufficiently large integer, $T = \poly(d)$. 
Let $\bH\in \R^{d \times d}$ be any symmetric matrix whose eigenvalues satisfy 
$
1/3\leq \lambda_d(\bH) \leq \cdots \leq \lambda_1(\bH) \leq 1,
$
and $\bz^{(1)}, \ldots, \bz^{(T)}$ be online queries.
Assuming the $\omv$ conjecture is true, then there is no algorithm with $\poly(d)$ preprocessing time and $O(d^{2-\gamma})$ amortized running time that can return an $O(1/d^{2})$-approximate answer to $\bH \bz^{(t)}$ for all $t \in [T]$, i.e., a vector $\by^{(t)} \in \R^{d}$ s.t. $\|\by^{(t)} - \bH \bz^{(t)}\|_2 \leq O(1/d^2)$. 
\end{lemma}

The remaining proof of Theorem~\ref{thm:lower-full} proceeds in a few steps.
We introduce the online projection problem in Section \ref{sec:online-projection} and prove that the $\omv$ conjecture implies that the online projection problem requires $\Omega(d^{2-\gamma})$ amortized time to solve to $O(1/d^2)$ accuracy.
We amplify the hardness to constant accuracy in Section \ref{sec:hard-amplification}, and we reduce the online projection to fully dynamic-LSR in Section \ref{sec:reduction}. 
See an illustration of these steps in Figure~\ref{fig:fully}.

\begin{figure}[!ht]
  \centering
  \begin{tikzpicture}[node distance=1.2cm, box/.style={rectangle, draw, text width=4.2cm}, >={Stealth[length=5pt]}]
    \node[box] (1) {$\omv$ conjecture \\ (Conjecture~\ref{conj:omv})};
    \node[box, right=of 1] (2) {$O(1/d^2)$-approximate real-valued $\omv$ (Lemma~\ref{lem:omv-real})};
    \node[box, below=of 1] (3) {$O(1/d^2)$-approximate online projection (Lemma~\ref{lem:online-projection})};
    \node[box, right=of 3] (4) {$(1/3, 1/d^3)$-approximate online projection (Lemma~\ref{lem:hard-amplification})};
    \node[box, right=of 4] (5) {$0.01$-approximate fully dynamic LSR \\ (Theorem~\ref{thm:lower-full})};

    \draw[->] (1) -- (2);
    \draw[->] (2) -- (3);
    \draw[->] (3) -- (4);
    \draw[->] (4) -- (5);
  \end{tikzpicture}
  \caption{An illustration of the chain of proofs in this section.}
  \label{fig:fully}
\end{figure}
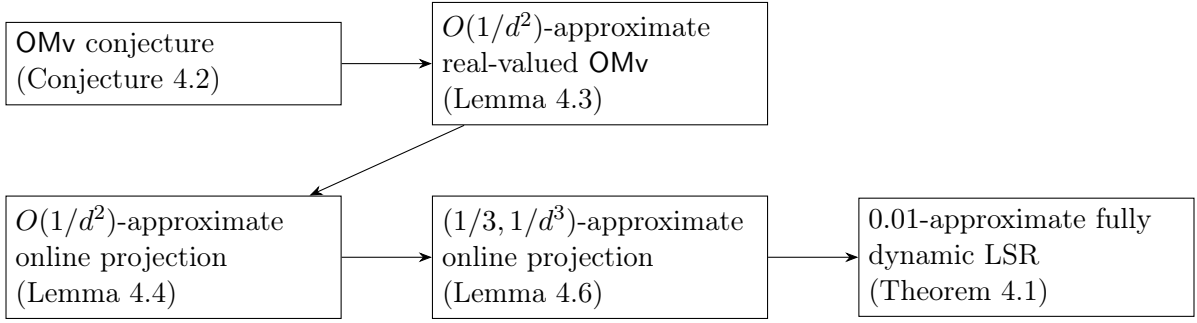

\subsection{Hardness of online projection}
\label{sec:online-projection}

Recall the definition of the online projection task, which asks to compute the projection of a sequence of online queries $\{\bz^\ttop\}_{t\in [T]}$ onto a fixed subspace $\bU$.

\Onlineprojection*

For any orthonormal $\bU \in \R^{d\times d_1}$, let $\bU_\perp \in \R^{d\times (d-d_1)}$ be the orthonormal matrix that spans the complementary of the column space of $\bU$, i.e., it satisfies $[\bU, \bU_\perp] \in \R^{d\times d}$ is a squared orthonormal matrix. For any vector $\bz \in \R^d$, define 
\begin{align}
\bz = \bz_{\bU} + \bz_{\bU_\perp} \quad \text{where} \quad \bz_\bU = \bU\bU^{\top}\bz \quad \text{and} \quad \bz_{\bU_\perp} = (\mathbf{I} - \bU\bU^{\top})\bz.
\end{align}
That is to say, $\bz_{\bU}$ is the projection of $\bz$ onto the subspace spanned by the columns of $\bU$, and $\bz_{\bU_{\perp}}$ is the projection of $\bz$ onto the subspace spanned by the complementary of $\bU$. In this section we also denote the projection of a vector $\bz_{k}$ as $\bz_{k, \bU}$.

We prove computing online projection requires $\Omega(d^{2-\gamma})$ amortized time assuming $\omv$.
\begin{lemma}[Hardness of online projection]
\label{lem:online-projection}
Let $\gamma > 0$ be any constant. Assuming the $\omv$ conjecture is true, then there is no algorithm with $\poly(d)$ preprocessing time and $O(d^{2-\gamma})$ amortized running time that can return an $O(1/d^2)$-approximate solution $\hat{\bz}^{\ttop}_{\bU} \in \R^d$ that satisfies $\|\hat{\bz}^{\ttop}_{\bU} - \bU\bU^{\top} \bz^{(t)}\|_2 \leq O(1/d^2)$ for the online projection problem.
\end{lemma}
\begin{proof}
By Lemma \ref{lem:omv-real}, it suffices to reduce from the $O(1/d^2)$-approximate $\omv$ problem of a real-valued PSD matrix $\bH$ with eigenvalues $1/3 \leq \lambda_d(\bH) \leq \cdots \leq \lambda_1(\bH) \leq 1$.
We apply a binary division trick and (approximately) decompose $\bH$ into $k = O(\log d)$ projection matrices $\bU(1), \ldots, \bU(k)$.
Formally, let $\bH = \bU\Sigma \bU^{\top}$ where $\bU \in \R^{d\times d}$ and $\Sigma = \diag(\lambda_1(\bH), \ldots, \lambda_n(\bH))$.
Let $\lambda_{i}(\bH) =0.\lambda_{i,1}\lambda_{i,2}\ldots$ be the binary representation of $\lambda_i$ ($i \in [n]$).
For each $j \in [k]$, let 
\[
S_j = \{i:  i\in [n], \lambda_{i, j} = 1\} \subseteq [n]
\]
be the subset of coordinates with non-zero binary value at the $j$-th bit. Let $\bU(j) = \bU_{*,S_j} \in \R^{d\times |S_j|}$ be an orthonormal matrix that takes columns from $S_j$.

We reduce an $\omv$ instance to $k$ online projection instances, and show that we can compute an $O(1/d^2)$-approximate solution to an $\omv$ query of $\bH$ by using $k$ online projection queries, one for each of the $k$ online projection instances. 
In the preprocessing step, we compute the orthonormal matrices $\bU(1), \ldots, \bU(k)$, and we let them be the initial matrices of the online projection instances. This step can be done in $O(d^3\log d)$ time. 
At the $t$-th step, given an online query $\bz^\ttop \in \R^d$ of $\omv$ with norm $\|\bz^\ttop\|_2 \leq 1$, we make one query for each instance of online projection and let $\bz_j^\ttop$ be the projection returned by the $j$-th instance.
It satisfies 
\begin{align}
\|\bz^{\ttop}_{j} - \bU(j)\bU(j)^{\top}\bz^\ttop\|_2 \leq O(1/d^2). \label{eq:online-proj-guarantee}
\end{align}
The vectors $\bz^\ttop_1, \ldots, \bz^\ttop_k$ can be computed in $k \mathcal{T} = O(\mathcal{T} \cdot \log d)$ (amortized) running time, where $\mathcal{T}$ is the runtime of the online projection algorithm. Finally, we output
\begin{align}
\by^\ttop = \sum_{j=1}^{k} \frac{1}{2^j} \cdot \bz^\ttop_j \label{eq:omv-output}
\end{align}
for the $\omv$ query.
Our goal is to prove $\by^\ttop$ is an $O(1/d^2)$-approximate solution to the $\omv$ query, i.e. 
\begin{align}\label{eq:omv_error}
\|\by^\ttop - \bH\bz^\ttop\|_2 \leq O(1/d^2).
\end{align}

To this end, we have
\begin{align}
\|\by^\ttop - \bH\bz^\ttop\|_2 = &~ \left\|\sum_{j=1}^{k} \frac{1}{2^j} \cdot \bz^\ttop_j - \bH\bz^\ttop\right\|_2 \notag \\
\leq &~  \left\|\sum_{j=1}^{k} \frac{1}{2^j} \cdot \bU(j)\bU(j)^{\top}\bz^\ttop - \bH\bz^\ttop\right\|_2 + \sum_{j=1}^{k}\frac{1}{2^j}\left\|\bU(j)\bU(j)^{\top}\bz^\ttop - \bz_j^\ttop\right\|_2 \notag \\
\leq &~ \left\|\sum_{j=1}^{k} \frac{1}{2^j} \cdot \bU(j)\bU(j)^{\top} \bz^\ttop - \bH\bz^\ttop\right\|_2 + O(1/d^2).\label{eq:online1}
\end{align}
Here the first step follows from the definition of $\by^\ttop$ in Eq.~\eqref{eq:omv-output}, the second step follows from the triangle inequality, and the last step follows from Eq.~\eqref{eq:online-proj-guarantee}.

It remains to bound the first term of RHS, and it suffices to prove $\sum_{j=1}^{k} \frac{1}{2^j} \cdot \bU(j)\bU(j)^{\top}$ is close to $\bH$. 
This holds (almost) by definition. 
Formally, let $\Sigma(j) = \frac{1}{2^j}\diag(\lambda_{1,j}, \ldots, \lambda_{d, j})$ for any $j \geq 1$. By the definition of $\bU(j)$, we have
\begin{align*}
\frac{1}{2^j} \cdot \bU(j)\bU(j)^{\top} =\bU\Sigma(j)\bU^{\top}, \quad \forall j \in [k],
\end{align*}
and therefore 
\begin{align*}
\left\|\bH - \sum_{j=1}^{k}\frac{1}{2^j} \cdot \bU(j)\bU(j)^{\top}\right\|_2 = &~ \left\| \bU \Sigma \bU^{\top} - \sum_{j=1}^{k}\bU\Sigma(j)\bU^{\top}\right\|_2 \\
= &~ \left\|\sum_{j> k} \bU\Sigma(j)\bU^{\top}\right\|_2 \leq \frac{1}{2^{k}} = O(1/d^2),
\end{align*}
where the third step follows from $\| \bU\Sigma(j)\bU^{\top} \|_2 = \frac{1}{2^j} \cdot \|\bU(j) \bU(j)^{\top}\|_2 = \frac{1}{2^j}$, and the last step follows from $k = O(\log d)$.

Plugging into Eq.~\eqref{eq:online1}, we have 
\begin{align*}
\|\by^\ttop - \bH\bz^\ttop\|_2 \leq &~\left\|\sum_{j=1}^{k} \frac{1}{2^j} \cdot \bU(j)\bU(j)^{\top}\bz^\ttop - \bH\bz^\ttop\right\|_2 + O(1/d^2)\\
\leq &~ \left\|\sum_{j=1}^{k} \frac{1}{2^j} \cdot \bU(j)\bU(j)^{\top} - \bH\right\|_2 \|\bz^\ttop\|_2  + O(1/d^2)\\
\leq &~ O(1/d^2)\cdot 1 + O(1/d^2) = O(1/d^2).
\end{align*}

In summary, the above reduction means there exist an $O(1/d^2)$-approximate $\omv$ algorithm for $\bH$ with $O(d^3 \log d)$ preprocessing time and $O(\mathcal{T} \cdot \log d)$ amortized query time. If $\mathcal{T} = O(d^{2-\gamma})$ for some constant $\gamma$, then we can solve the $O(1/d^2)$-approximate $\omv$ problem for $\bH$ in amortized $O(d^{2-\gamma} \cdot \log d)$ time, and by Lemma~\ref{lem:omv-real} this contradicts with the $\omv$ conjecture.
\end{proof}

\subsection{Hardness amplification}
\label{sec:hard-amplification}
So far we have proved an $\Omega(d^{2 - \gamma})$ lower bound of the online projection problem when it is required to output an $O(1/d^2)$-approximate answer per round. 
Our next step is to amplify this approximation precision to a constant.
We first formalize the notion of $(\alpha, \beta)$-approximate projection.

\begin{definition}[$(\alpha, \beta)$-approximate projection]
Given an orthonormal matrix $\bU \in \R^{d\times d_1}$ and a vector $\bz \in \R^{d}$, we say a vector $\by\in \R^{d}$ is an $(\alpha, \beta)$-approximate projection of $\bz$ onto $\bU$, if it satisfies 
\begin{align*}
\|\by - \bU\bU^{\top}\bz\|_2 \leq \alpha \|\bU\bU^\top \bz\|_2 + \beta\|\bz\|_2.
\end{align*}
\end{definition}
As we shall see soon, the interesting regime is $\alpha = \Theta(1)$ and $\beta = 1/\poly(d)$. The problem of online projection (Definition \ref{def:online-projection}) works well with the requirement of outputting an $(\alpha, \beta)$-approximate projection per round. 

\begin{lemma}[Hardness amplification]
\label{lem:hard-amplification}
Let $\gamma > 0$ be any constant. Let $\alpha = 1/3$ and $\beta = O(1/d^3)$.
Assuming the $\omv$ conjecture is true, then there is no algorithm with $\poly(d)$ preprocessing time and $O(d^{2-\gamma})$ amortized running time that can return an $(\alpha, \beta)$-approximate solution for the online projection problem.
\end{lemma}

The reduction is formally shown in Algorithm \ref{algo:hard-amplification}. Our goal is to show that we can answer online projection queries to $O(1/d^2)$ accuracy by using $(\alpha, \beta)$-approximate oracles. We use $\mathbb{P}_\bU, \mathbb{P}_{\bU_\perp}: \R^d \to \R^d$ to denote $(\alpha, \beta)$-approximate projection oracles whose outputs satisfy that for any vector $\bz \in \R^d$, 
(1) $\|\mathbb{P}_{\bU}(\bz) - \bz_\bU\|_2 \leq  \alpha \|\bz_\bU\|_2 + \beta \|\bz\|_2$, and 
(2) $\|\mathbb{P}_{\bU_\perp}(\bz) - \bz_{\bU_\perp}\|_2 \leq  \alpha \|\bz_{\bU_\perp}\|_2 + \beta \|\bz\|_2 $.

The reduction proceeds in $R = \Theta(\log d)$ rounds (i.e., the outer-loop on Line \ref{line:outer}), and we wish to show that each round (1) keeps the projected component $\bz_{r, \bU}$, and (2) reduces the orthogonal component $\bz_{r, \bU_\perp}$ (see Lemma \ref{lem:outer}).
In each round, the reduction first calls the approximate projection oracle onto $\bU_\perp$, which gives a good approximation to the orthogonal component $\bz_{r, \bU_\perp}$, but also has non-negligible component onto space $\bU$.
To resolve this, Algorithm \ref{algo:hard-amplification} proceeds in $K = O(\log d)$ iterations (i.e., inner loop on Line \ref{line:inner}), and in each iteration, it combines the previous output and sends it to $\mathbb{P}_{\bU}$. This process gradually purifies the component onto space $\bU$ (see Lemma \ref{lem:inner-induction}).

\begin{algorithm}[!htbp]
\caption{Hardness amplification}
\label{algo:hard-amplification}
\begin{algorithmic}[1]
\State {\bf Input:} Online query $\bz$, approximate projection oracles $\mathbb{P}_{\bU}$ and $\mathbb{P}_{\bU_\perp}$
\State $\bz_1 \leftarrow \bz$
\For{$r =1,2,\ldots, R$}\Comment{$R = O(\log d)$}\label{line:outer}
\State $\bw_{r, 0} \leftarrow \mathbb{P}_{\bU_\perp}(\bz_r)$
\For{$k=1,2,\ldots, K$}\Comment{$K = O(\log d)$} \label{line:inner}
\State $\by_{r, k} \leftarrow \mathbb{P}_{\bU}(\bw_{r, k-1})$
\State $\bw_{r, k} \leftarrow \bw_{r, k-1} - \by_{r, k} $ \label{line:w-def}
\EndFor
\State $\bz_{r+1} \leftarrow \bz_{r} - \bw_{r, K}$ \label{line:update}
\EndFor
\State \Return $\bz_{R+1}$
\end{algorithmic}
\end{algorithm}

We first state the guarantee of each outer loop. W.l.o.g., we assume $\|\bz\|_2 = 1$.
\begin{lemma}
\label{lem:outer}
For each round $r \in [R]$, we have
\begin{itemize}
\item $\|\bz_{r+1, \bU} - \bz_{r, \bU}\|_2 \leq 4(K+2)\beta$, 
\item $\|\bz_{r+1, \bU_\perp}\|_2 \leq 2\alpha \|\bz_{r, \bU_\perp}\|_2 + 4K^2\beta$, and
\item $\|\bz_{r+1}\|_2 \leq 1 + O(K^4r \beta)$.
\end{itemize}
\end{lemma}

It is useful to first understand the guarantee of the inner loops. 
At the beginning, we have the following lemma for $\bw_{r, 0}$.
\begin{lemma}
\label{lem:inner-begin}
For any round $r \in [R]$, assuming $\|\bz_r\|_2\leq 2$, then we can write 
\begin{align*}
\bw_{r, 0} = \bz_{r, \bU_\perp} - \bdelta_{r, 0} \quad \text{where} \quad \|\bdelta_{r, 0}\|_{2} \leq \alpha \|\bz_{r, \bU_{\perp}}\|_2 + 2\beta.
\end{align*}
\end{lemma}
\begin{proof}
The proof follows directly from the guarantee of $\mathbb{P}_{\bU_\perp}$. In particular, we have that 
\begin{align*}
\|\bw_{r, 0} - \bz_{r, \bU_{\perp}}\|_2 = \|\mathbb{P}_{\bU_\perp}(\bz_r) - \bz_{r, \bU_{\perp}}\|_2 \leq \alpha \|\bz_{r, \bU_{\perp}}\|_2 + \beta \|\bz_r\|_2 \leq \alpha \|\bz_{r, \bU_{\perp}}\|_2 + 2\beta,
\end{align*}
where the second step holds since $\mathbb{P}_{\bU_\perp}$ returns an $(\alpha, \beta)$ approximation over projection onto $\bU_{\perp}$ and the last step holds since $\|\bz_r\|_2 \leq 2$. 
We complete the proof here.
\end{proof}

For each iteration $k \in [K]$, we have the following lemma for the inner loop.
\begin{lemma}
\label{lem:inner-induction}
For any round $r\in [R]$ and iteration $k \in [K]$, assuming $\|\bz_r\|_2 \leq 2$, we can write
\begin{align*}
\bw_{r, k} = \bz_{r, \bU_{\perp}} - (\sum_{\tau = 0}^{k-1} \bdelta_{r, \tau, \bU_{\perp}}) - \bdelta_{r, k}, ~~~\text{and}~~
\by_{r, k} = -\bdelta_{r, k-1, \bU} + \bdelta_{r, k}, 
\end{align*}
and each $\bdelta_{r, k}$ satisfies
\begin{align*}
\|\bdelta_{r, k}\|_2 \leq \alpha \|\bdelta_{r, k-1, \bU}\|_2 + 4 \beta \quad \text{and} \quad \|\bdelta_{r, k}\|_2 \leq \alpha^{k+1}\|\bz_{r, \bU_{\perp}}\|_2 + 4(k+1)\beta.
\end{align*}
\end{lemma}
\begin{proof}
We prove the claim by induction. The base case that $\bw_{r, 0} = \bz_{r, \bU_\perp} - \bdelta_{r, 0}$ and $\|\bdelta_{r, 0}\|_{2} \leq \alpha \|\bz_{r, \bU_{\perp}}\|_2 + 2\beta$ is proved in Lemma~\ref{lem:inner-begin}. (Note that $\by_{r, 0}$ is not defined.)

Suppose the lemma statement holds for $k-1$, which means we have
\begin{align}
\bw_{r, k-1, \bU} = -\bdelta_{r, k - 1, \bU} \label{eq:inner2}
\end{align}
and
\begin{align}
\|\bw_{r, k-1}\|_2 = &~ \left\|\bz_{r, \bU_{\perp}} - (\sum_{\tau=0}^{k-2} \bdelta_{r, \tau, \bU_{\perp}}) - \bdelta_{r, k-1} \right\|_2 \notag \\
\leq &~ \|\bz_{r, \bU_{\perp}}\|_2 + (\sum_{\tau=0}^{k-2} \|\bdelta_{r, \tau, \bU_{\perp}}\|_2) + \|\bdelta_{r, k-1}\|_2 \notag \\
\leq &~ \|\bz_{r, \bU_{\perp}}\|_2 + \sum_{\tau=0}^{k-1} \|\bdelta_{r, \tau}\|_2  \notag\\
\leq &~ \sum_{\tau=0}^{k} \Big( \alpha^{\tau}\|\bz_{r, \bU_{\perp}}\|_2 + 4(\tau+1)\beta \Big) \notag \\
\leq &~ \frac{1-\alpha^{k+1}}{1-\alpha} \cdot \|\bz_{r, \bU_{\perp}}\|_2 + 4K^2\beta \leq 4, \label{eq:inner3}
\end{align}
where the second step holds from triangle inequality, the third step follows from 
\[
\|\bdelta_{r, \tau}\|_2^2 = \|\bdelta_{r, \tau, \bU} + \bdelta_{r, \tau, \bU_\perp}\|_2^2 = \|\bdelta_{r, \tau, \bU}\|_2^2 + \|\bdelta_{r, \tau, \bU_\perp}\|_2^2 \geq \|\bdelta_{r, \tau, \bU}\|_2^2,
\]
the fourth step holds from Lemma \ref{lem:inner-begin} and the induction hypothesis.
The last step follows the choice of $\alpha = 1/3, \beta = O(1/d^3)$, $K = O(\log d)$, and $\|\bz_{r, \bU_\perp}\|_2 \leq \|\bz_{r}\|_2 \leq 2$.

Now we are ready to prove that the induction hypothesis also holds for the $k$-th iteration.

{\bf Properties of $\bdelta_{r,k}$ and $\by_{r,k}$.}
We have
\begin{align}
\|\by_{r, k} + \bdelta_{r, k-1, \bU}\|_2 = &~ \|\mathbb{P}_{\bU}(\bw_{r, k-1}) + \bdelta_{r, k-1, \bU}\|_2 = \|\mathbb{P}_{\bU}(\bw_{r, k}) -  \bw_{r, k-1, \bU}\|_2\notag \\
\leq &~ \alpha\|\bw_{r, k-1, \bU}\|_2 + \beta \|\bw_{r, k-1}\|_2 = \alpha\|\bdelta_{r, k-1, \bU}\|_2 + 4\beta, \label{eq:inner4}
\end{align}
where the first step follows from the definition that $\by_{r, k} = \mathbb{P}_{\bU}(\bw_{r, k-1})$, the second step follows from Eq.~\eqref{eq:inner2}, the third step holds from the guarantee of $\mathbb{P}_{\bU}$ and the last step holds from Eq.~\eqref{eq:inner2}\eqref{eq:inner3}. 

Hence, define $\bdelta_{r, k} = \by_{r, k} + \bdelta_{r, k-1, \bU}$, from Eq.~\eqref{eq:inner4} we have that
\begin{align*}
\|\bdelta_{r, k}\|_2 \leq \alpha \|\bdelta_{r, k-1, \bU}\|_2 + 4\beta.
\end{align*}
Note that this definition of $\bdelta_{r, k}$ also gives us that
\begin{equation}\label{eq:inner_y}
\by_{r, k} = \bdelta_{r, k} - \bdelta_{r, k-1, \bU}.
\end{equation}
By induction hypothesis, we also have 
\[
\|\bdelta_{r, k}\|_2 \leq \alpha \|\bdelta_{r, k-1, \bU}\|_2 + 4\beta \leq \alpha\|\bdelta_{r, k-1}\|_2 + 4\beta \leq \alpha^{k + 1}\|\bz_{r, \bU_{\perp}}\|_2 + 4(k+1)\beta.
\]

{\bf Property of $\bw_{r,k}$.} We have
\begin{align*}
\bw_{r, k} = &~ \bw_{r, k-1} - \by_{r,k} \\
= &~ \Big( \bz_{r, \bU_{\perp}} - (\sum_{\tau=0}^{k-2} \bdelta_{r, \tau, \bU_{\perp}}) - \bdelta_{r, k-1}\Big) - \Big( \bdelta_{r, k} - \bdelta_{r, k-1, \bU} \Big)\\
= &~ \bz_{r, \bU_{\perp}} - (\sum_{\tau=0}^{k-1} \bdelta_{r, \tau, \bU_{\perp}}) - \bdelta_{r, k}.
\end{align*}
Here the first step follows from the definition of $\bw_{r, k}$ (Line \ref{line:w-def}), the second step follows from the induction hypothesis about $\bw_{r,k-1}$ and Eq.~\eqref{eq:inner_y} that we just proved. We conclude the proof here.
\end{proof}

Now we can go back to analyse the outer loops and prove Lemma \ref{lem:outer}.
\begin{proof}[Proof of Lemma \ref{lem:outer}]
Consider any round $r \in [R]$. We have
\begin{align}
\bz_{r+1} = &~ \bz_{r} - \bw_{r, K} \notag \\
= &~ \bz_{r} - \Big( \bz_{r, \bU_{\perp}} - (\sum_{\tau = 0}^{K-1} \bdelta_{r, \tau, \bU_{\perp}}) - \bdelta_{r, K} \Big) \notag \\
= &~ \bz_{r, \bU} + (\sum_{\tau = 0}^{K-1} \bdelta_{r, \tau, \bU_{\perp}}) + \bdelta_{r, K} \label{eq:outer1}
\end{align}
Here the first step follows from the update rule (Line \ref{line:update}), the second step follows from Lemma \ref{lem:inner-induction}. 

Hence, we have $\bz_{r+1, \bU} = \bz_{r, \bU} + \bdelta_{r, K, \bU}$, so for the first claim, we have
\begin{align*}
\|\bz_{r+1, \bU} - \bz_{r, \bU}\|_2 = &~ \|\bdelta_{r, K, \bU}\|_2 \leq \|\bdelta_{r, K}\|_2\\
\leq &~ \alpha^{K+1}\|\bz_{r, \bU_\perp}\|_2 + 4(K+1)\beta \leq 4(K+2)\beta.
\end{align*}
Here the third step follows from Lemma \ref{lem:inner-induction}, the last step follows from $\|\bz_{r, \bU_\perp}\|_2 \leq \|\bz_{r}\|_2 \leq 2$, and the choice of parameter that $K = O(\log d)$, $\alpha = 1/3$ and $\beta = O(1/d^3)$, so that $\alpha^K < \beta$.

For the second claim, the orthogonal component $\bz_{r+1, \bU_\perp}$ satisfies
\begin{align*}
\|\bz_{r+1, \bU_\perp}\|_2 = &~ \left\|\sum_{k=0}^{K} \bdelta_{r, k, \bU_\perp}\right\|_2 \leq \sum_{k=0}^{K}\|\bdelta_{r, k, \bU_\perp}\|_2 \leq \sum_{k=0}^{K}\|\bdelta_{r, k}\|_2 \\
\leq &~ \sum_{k=0}^{K} \Big( \alpha^{k+1}\|\bz_{r, \bU_\perp}\|_2 + 4(k+1) \beta \Big) \\
\leq &~ 2\alpha \|\bz_{r, \bU_\perp}\|_2 + 4K^2 \beta,
\end{align*}
where the first step follows from Eq.~\eqref{eq:outer1}, the second step follows from triangle inequality, the fourth step follows from Lemma \ref{lem:inner-induction}, and the last step follows from the choice of parameter that $\alpha = 1/3$.

Finally, we prove the third claim by induction on $r$. First note that in the base case where $r=0$, by definition we have $\bz_{1, \bU} = \bz$, so $\|\bz_{1, \bU}\|_2 = \|\bz\|_2 = 1$. Suppose the third claim continues to hold up to round $r-1$, for the $r$-th round, we have
\begin{align*}
\|\bz_{r+1}\|_2^2 = &~ \|\bz_{r+1, \bU_\perp}\|_2^2 + \|\bz_{r+1, \bU}\|_2^2 \\
\leq &~ \big(2\alpha \|\bz_{r, \bU_\perp}\|_2 + 4K^2\beta \big)^2 + \big(\|\bz_{r, \bU}\|_2 + 4(K+2)\beta \big)^2 \\
\leq &~ \big(\|\bz_{r, \bU_\perp}\|_2 + 4K^2\beta \big)^2 + \big(\|\bz_{r, \bU}\|_2 + 4 K^2 \beta \big)^2 \\
= &~ (\|\bz_{r, \bU_\perp}\|_2^2 + \|\bz_{r, \bU}\|_2^2) + 8 K^2 \beta \cdot (\|\bz_{r, \bU_\perp}\|_2 + \|\bz_{r, \bU}\|_2) + 32 K^4 \beta^2 \\
\leq &~ \|\bz_r\|_2^2 + 16 K^2 \beta \cdot \|\bz_r\|_2 + 32 K^4 \beta^2 \\
\leq &~ 1 + O(K^4 r \beta),
\end{align*}
where the second step follows from the first two claims that we just proved: $\|\bz_{r+1, \bU_\perp}\|_2 \leq 2\alpha \|\bz_{r, \bU_\perp}\|_2 + 4K^2\beta$, and $\|\bz_{r+1, \bU}\|_2 \leq \|\bz_{r, \bU}\|_2 + \|\bz_{r+1, \bU} - \bz_{r, \bU}\|_2 \leq \|\bz_{r, \bU}\|_2 + 4(K+2)\beta$, the third step follows from $2 \alpha < 1$ since $\alpha = 1/3$ and $K+2 < K^2$ since $K = O(\log d)$, the fifth step follows from $\|\bz_r\|_2^2 = \|\bz_{r, \bU_\perp}\|_2^2 + \|\bz_{r, \bU}\|_2^2$, and the last step follows from the induction hypothesis that $\|\bz_r\|_2 \leq 1 + O(K^4 (r-1) \beta)$, and that $K^2 \beta < K^4 \beta$ and $K^4 \beta^2 < K^4 \beta$.
\end{proof}

Now we can wrap up the reduction and prove Lemma \ref{lem:hard-amplification}.
\begin{proof}[Proof of Lemma \ref{lem:hard-amplification}]
We prove that if there is an algorithm that outputs $(\alpha, \beta)$-approximate solutions for the online projection problem in $O(d^{2-\gamma})$ amortized time, then we can use this algorithm to obtain $O(1/d^2)$-approximate solutions for the online projection problem in $O(d^{2-\gamma + o(1)})$ amortized time, and hence contradicts with Lemma~\ref{lem:online-projection}.

Given an orthonormal matrix $\bU$ and let $\bz^\ttop$ be the query at the $t$-th round of the online projection problem, then we perform the reduction shown in Algorithm \ref{algo:reduction} and its output $\bz^\ttop_{R+1}$ satisfies
\begin{align*}
\|\bz_{R+1, \bU}^\ttop - \bz_{\bU}^\ttop\|_2 = &~ \|\bz_{R+1, \bU}^\ttop - \bz_{1, \bU}^\ttop\|_2 \leq \sum_{r=1}^{R} \|\bz_{r+1, \bU}^\ttop - \bz_{r, \bU}^\ttop\|_2 \leq O(RK\beta).
\end{align*}
Here the first inequality follows from triangle inequality and the second one holds due to the first claim of Lemma \ref{lem:outer}.
Meanwhile, due to the second claim of Lemma \ref{lem:outer}, we have
\begin{align*}
\|\bz_{R+1, \bU_\perp}^\ttop\|_2 \leq (2\alpha)^{K}\|\bz_{1, \bU_\perp}^\ttop\|_2 + O(RK^2\beta) \leq O(RK^2\beta),
\end{align*}
where the second step follows from that $(2 \alpha)^K < 1/d^3$ since $K = O(\log d)$ and $\alpha = 1/3$.

Combining the above two inequalities, and since $R = O(\log d)$, $K = O(\log d)$, and $\beta = O(1/d^3)$, we obtain
\begin{align*}
\|\bz_{R+1}^\ttop - \bz_{\bU}^\ttop\|_2 \leq \|\bz_{R+1, \bU}^\ttop - \bz_{\bU}^\ttop\|_2 + \|\bz_{R+1, \bU_\perp}^\ttop\|_2 \leq O(RK\beta  + RK^2\beta) \leq O(1/d^2).
\end{align*}
That is to say, $\bz_{R+1}^\ttop$ is an $O(1/d^2)$-approximate projection of $\bz^\ttop$ onto $\bU$.

We still need to bound the runtime of the reduction. Let $\mathcal{T}$ denote the amortized query time of the $(\alpha, \beta)$-approximate oracles $\mathbb{P}_{\bU_\perp}$ and $\mathbb{P}_{\bU}$. The reduction involves $R = O(\log d)$ outer loops, with each outer loop requiring a single call to $\mathbb{P}_{\bU_\perp}$ and containing $K = O(\log d)$ inner loops. During each inner loop, a single call to $\mathbb{P}_{\bU}$ is made, and the construction of $\bw_{r,k}$ takes $O(d)$ time.

Therefore, we can conclude that the total runtime of the algorithm is bounded by $RK \cdot \mathcal{T} + O(RKd) = (\mathcal{T} + d) \cdot O(\log^2 d)$. If $\mathcal{T} = O(d^{2-\gamma})$, then we can solve the $O(1/d^2)$-approximate online projection problem in amortized $O(d^{2-\gamma} \cdot \log^2 d)$ time, and this contradicts with the $\omv$ conjecture by Lemma~\ref{lem:online-projection}. This completes the proof.
\end{proof}

As a corollary, we prove the hardness of constant approximate-$\omv$.
\begin{theorem}[Hardness of approximate-$\omv$]
\label{thm:hard-omv-approx}
Let $d$ be a sufficiently large integer and $T = \poly(d)$.
Let $\gamma > 0$ be any constant, $\alpha = 1/3$ and $\beta = O(1/d^3)$. 
Let $\bH\in \R^{d\times d}$ ($\|\bH\|_2 =1$), and $\bz^{(1)}, \ldots, \bz^{(T)}$ be online queries ($\|\bz^\ttop\|_2 = 1$).
Assuming the $\omv$ conjecture is true, then there is no algorithm with $\poly(d)$ preprocessing time and $O(d^{2-\gamma})$ amortized running time that can return an $(\alpha, \beta)$-approximate answer to $\bH\bz^\ttop$ for all $t \in [T]$, i.e., a vector $\by^\ttop$ s.t. $\|\by^\ttop - \bH\bz^\ttop\|_2 \leq \alpha \|\bH\bz^\ttop\|_2 + \beta$. This continues to hold when $\bH$ is a projection matrix.
\end{theorem}

\subsection{Reduction from online projection to fully dynamic LSR}
\label{sec:reduction}
Finally, we provide a reduction from $(\alpha, \beta)$-approximate online projection to $\epsilon$-approximate fully dynamic LSR, where $\alpha = 1/3$, $\beta = O(1/d^3)$, and $\eps = \frac{1}{100}$. 
Given an instance of online projection with orthonormal matrix $\bU \in \R^{d \times d_1}$, we first set up the LSR problem.

\vspace{+2mm}
{\bf \noindent Setup for reduction \ \ } Let 
\begin{align*}
\bA^{(0)} = 
\left[
\begin{matrix}
\sqrt{\lambda} \cdot \mathbf{I}_d \\
(\bU_\perp)^{\top}
\end{matrix}
\right] \in \R^{(2d-d_1) \times d}  \quad \text{and} \quad \bb^{(0)} = 
\left[
\begin{matrix}
\mathbf{0}_d\\
\frac{1}{\sqrt{d}} \cdot \mathbf{1}_{d-d_1}
\end{matrix} 
\right] 
\in \R^{2d-d_1}
\end{align*}
where $\lambda = 1/d^{40}$. For convenience, we have included a notation table in Table~\ref{tab:parameters}.

\begin{table}[ht]
\centering
\begin{tabular}{|c|c|c|}
\hline
Parameter & Value & Comment \\ \hline
$\eps$ & $< 1/100$ & approximation factor of fully dynamic LSR \\ \hline
$\lambda$ & $1/d^{40}$ & coefficient of the regularization term \\ \hline
$\alpha$ & $1/3$ & approximation factor of online projection \\ \hline
\end{tabular}
\caption{Parameters used in the reduction from online projection to fully dynamic LSR}
\label{tab:parameters}
\end{table}

It would be convenient to view the first $d$ rows as a regularization term, and the (squared) loss equals to
\begin{align*}
L(\bx) := \|\bA^{(0)} \bx - \bb^{(0)}\|_2^2 = \left\|(\bU_\perp)^{\top}\bx - \frac{1}{\sqrt{d}} \cdot \mathbf{1}_{d-d_1}\right\|_2^2 + \lambda \|\bx\|_2^2.
\end{align*}
In the processing step, we also compute
\begin{equation}\label{eq:def_x*}
\bx^{*} := \frac{1}{\sqrt{d}} \sum_{j=1}^{d-d_1}\bU_{\perp,j} \in \R^d,
\end{equation}
where with a slight abuse of notation we let $\bU_{\perp, j} \in \R^d$ denote the $j$-th column of matrix $\bU_\perp$ (Hence $\bx^{*}$ also lies in the column space of $\bU_\perp$).
Overall, the preprocessing step takes at most $O(d^\omega)$ time.

\vspace{+2mm}
{\bf \noindent Online projection query \ \ } Given an online projection query $\bz^\ttop \in \R^{d}$ of the $t$-th step, recall our goal is to find an $(\alpha, 1/d^3)$-approximate projection onto $\bU$.

The reduction is formally presented in Algorithm \ref{algo:reduction}. First, it inserts a new row of $(\frac{1}{10}\cdot \bz^\ttop, 1)$ to the matrix, and then it calls the dynamic LSR solver (Line \ref{line:regression}) to obtain an $\eps$-approximate solution $\bx^\ttop$. The final output is determined as follows: if the component $\bz_\bU^\ttop$ is already sufficiently small, then it is captured by the condition on Line \ref{line:termination}, and we can simply output $\bf{0}$. Otherwise, Algorithm \ref{algo:reduction} outputs a scaled version of $(\bx^\ttop - \bx^{*})$, where the scaling factor is determined by Eq.~\eqref{eq:interpolation}. Finally, the new row is deleted, and we return to the original setup.

\begin{algorithm}[!htbp]
\caption{Reduction: From $(\alpha, 1/d^3)$-approximate online projection to $\epsilon$-approximate dynamic LSR}
\label{algo:reduction}
\begin{algorithmic}[1]
\State Insert $(\frac{1}{10}\cdot\bz^\ttop, 1) \in \R^{d}\times \R$ \Comment{Insert a new row}
\State Call the regression solver and let $\bx^{\ttop}$ be an $\eps$-approximate solution of the square root of \label{line:regression}
\begin{align}
L^\ttop(\bx) := \left\|(\bU_{\perp})^\top\bx - \frac{1}{\sqrt{d}} \cdot \mathbf{1}_{d-d_1}\right\|_2^2 + \frac{1}{100}|\langle \bz^\ttop, \bx\rangle - 10|^2 + \lambda \|\bx\|_2^2   \label{eq:lsr-r}
\end{align}
\State $\by^\ttop \leftarrow \bx^{\ttop} - \bx^{*}$ \label{line:y-rt}
\If{$\|\by^\ttop\|_2 \geq d^3$ \textbf{or} $|10 - \langle \bz^\ttop, \bx^\ttop\rangle| \geq 200 d^4\sqrt{\lambda}$}  
\label{line:termination}
\State \Return $\hat{\bz}_\bU^\ttop \leftarrow \mathbf{0}$ 
\Else
\State \Return $\hat{\bz}_\bU^\ttop \leftarrow \xi^{*} \cdot \by^\ttop$ where \label{line:interpolation}
\begin{align}
\xi^{*} = \arg\min_{\xi} \|\bz^\ttop - \xi \cdot \by^\ttop\|_2 \label{eq:interpolation}
\end{align}
\EndIf
\State Delete the row $(\frac{1}{10}\cdot\bz^\ttop, 1)$ \Comment{Delete the new row}
\end{algorithmic}
\end{algorithm}

Intuitively, the first term $\|(\bU_{\perp})^\top\bx - \frac{1}{\sqrt{d}} \cdot \mathbf{1}_{d-d_1}\|_2^2$ of the loss $L^\ttop(\bx)$ enforces the approximate solution $\bx^{(t)}$ to satisfy that $\bx^{(t)}_{\bU_{\perp}} \approx \bx^*$ on the subspace $\bU_{\perp}$, since $\bx^*$ is the minimizer of the first term. The second term $\frac{1}{100}|\langle \bz^\ttop, \bx\rangle - 10|^2$ then enforces $\bx^{(t)}_{\bU}$ to be close to a scaled version of $\bz^{(t)}_{\bU}$. Thus, $(\bx^{(t)} - \bx^*)$ is approximately a scaled version of $\bz^{(t)}_\bU$.

Formally, our goal is to prove the following lemma.
\begin{lemma}
\label{lem:reduction-lsr}
For any $t \in [T]$, the output of Algorithm \ref{algo:reduction} satisfies
\begin{align*}
\|\hat{\bz}_{\bU}^{\ttop} - \bz_{\bU}^\ttop\|_2 \leq \alpha \|\bz_\bU^\ttop\|_2 + O(1/d^3).
\end{align*}
\end{lemma}
\begin{proof}
We will prove the lemma by considering three different cases. We first give a short summary.
\begin{itemize}
\item {\bf Case 1: $\|\bz_{\bU}^\ttop\|_2 \geq 1/d^4$.} We prove that in this case we always have $|10 - \langle \bz^\ttop, \bx^{(t)}\rangle| \geq 200 d^4\sqrt{\lambda}$, so the condition of Line~\ref{line:termination} reduces to test whether $\|\by^\ttop\|_2 \geq d^3$ or not.
\begin{itemize}
\item {\bf Case 1-1: $\|\by^\ttop\|_2 \geq d^3$.} Then the condition of Line \ref{line:termination} is satisfied and we prove $\|\bz_\bU^\ttop\|_2 \leq 1/d^3$, so it is fine to output $\hat{\bz}_\bU^\ttop = \bf{0}$.
\item {\bf Case 1-2: $\|\by^\ttop\|_2 < d^3$.} Then the condition is not satisfied, and we prove the output $\hat{\bz}_\bU^\ttop = \xi^{*} \cdot \by^\ttop$ is an $(\alpha, 1/d^3)$-approximate projection of $\bz^\ttop$. 
This is the main technical part of the proof.
\end{itemize}
\item {\bf Case 2: $\|\bz_{\bU}^\ttop\|_2 < 1/d^4$.} In this case we prove that the termination condition of Line \ref{line:termination} must be true, and therefore, the output $\bz^{\ttop} = \mathbf{0}$ is an $O(1/d^3)$-approximation of $\bz_\bU^{\ttop}$.
\end{itemize}
Before going into details of the three cases, we first define a vector
\begin{align}\label{eq:x_rt*}
\bx^{*}_{t} = \bx^{*} + \frac{10 - \langle \bz_{\bU_\perp}^\ttop,\bx^{*}\rangle}{\|\bz_{\bU}^\ttop\|_2^2} \bz_{\bU}^\ttop \in \R^d.
\end{align}
We note $\bx^{*}_{t}$ is not the optimal solution of Eq.~\eqref{eq:lsr-r}, but it gives a good upper bound of the loss. 

To understand the role of $\bx_t^{*}$, note that if $\bx^{*}_{t}$ is a good approximation of the optimal solution of Eq.~\eqref{eq:lsr-r}, then $\bx^{\ttop}$ will be close to $\bx^{*}_{t}$, so $\by^\ttop = \bx^{\ttop} - \bx^{*} \approx \bx^{*}_{t} - \bx^{*}$. As a result, $\|\by^\ttop\|_2 \approx \|\bx^{*}_{t} - \bx^{*}\|_2 = O(1/\|\bz_\bU^\ttop\|_2)$. If $\|\by^\ttop\|_2 \geq d^3$ (the first part of the termination condition on Line~\ref{line:termination}) then we have $\|\bz_{\bU}^\ttop\|_2$ is small, so $\mathbf{0}$ is a good approximation of $\bz^{\ttop}_{\bU}$. 
On the other hand, if $\bx^{*}_{t}$ is not a good approximation of the optimal solution of Eq.~\eqref{eq:lsr-r}, then this means $\|\bz_{\bU}^\ttop\|_2$ is way too small, and we can capture this by the second part of the termination condition, i.e., the second term in the objective will be large.

One can verify that $\bx^{*}_{t}$ obtains zero loss except for the third regularization term. That is, it satisfies 
\begin{align*}
\langle \bx^{*}_{t}, \bU_{\perp, j}\rangle = &~ \langle \bx^{*}, \bU_{\perp,j}\rangle = \frac{1}{\sqrt{d}}, \quad \forall j \in [d-d_1], \\
\text{and, } \langle \bx^{*}_{t}, \bz^\ttop\rangle = &~\langle \bx^{*}, \bz_{\bU_\perp}^\ttop \rangle + \frac{10 - \langle \bz_{\bU_\perp}^\ttop,\bx^{*}\rangle}{\|\bz_{\bU}^\ttop\|_2^2}\cdot \langle \bz_{\bU}^\ttop, \bz_{\bU}^\ttop\rangle = 10,
\end{align*}
where the first step of the second equation follows from $\bx^*$ is in the subspace $\bU_\perp$ so it's orthogonal to $\bz_{\bU}^\ttop$.
Consequently, we have $\|(\bU_\perp)^{\top}\bx_{t}^{*} - \frac{1}{\sqrt{d}} \cdot \mathbf{1}_{d-d_1}\|_2 = 0$ and $|\langle \bz^\ttop, \bx_{t}^{*}\rangle - 10| = 0$, so $L^\ttop(\bx_{t}^{*})$ defined in Eq.~\eqref{eq:lsr-r} satisfies
\begin{align}
L^\ttop(\bx_{t}^{*}) = \lambda\|\bx_{t}^{*}\|_2^2 
\leq \lambda \cdot \frac{(10 - \langle \bz_{\bU_\perp}^\ttop,\bx^{*}\rangle)^2}{\|\bz_{\bU}^\ttop\|_2^2} + \lambda := \lambda (\Delta_{t}^2 + 1),
\label{eq:loss}
\end{align}
Here the second step follows from the definition of $\bx_{t}^{*}$ in Eq.~\eqref{eq:x_rt*}, and that $\bx^*$ has norm $\|\bx^*\|_2 \leq 1$ and it's orthogonal to $\bz_{\bU}^\ttop$, for notational convenience in the third step we have defined
\begin{equation}\label{eq:def_Delta_rt}
\Delta_{t} := \frac{10 - \langle \bz_{\bU_\perp}^\ttop, \bx^{*}\rangle}{\|\bz_{\bU}^\ttop\|_2}.
\end{equation}

\vspace{+2mm}
{\bf \noindent Case 1 \ \ } Suppose $\|\bz_{\bU}^\ttop\|_2 \geq 1/d^4$. Since $\|\bx^{*}\|_2 \leq 1$ and $\|\bz_{\bU}^\ttop\|_2 \leq 1$, we have 
\begin{align}\label{eq:Delta_bound}
\Delta_{t} = \frac{10 - \langle \bz_{\bU_\perp}^\ttop, \bx^{*}\rangle}{\|\bz_{\bU}^\ttop\|_2} \in (9,  11d^4].
\end{align}
Therefore by Eq.~\eqref{eq:loss} we have
\begin{align*}
L^\ttop(\bx_{t}^{*}) \leq \lambda (\Delta_{t}^2 + 1) \leq \lambda \cdot (1 + 121d^8).
\end{align*}
The solution $\bx^\ttop$ is $\eps$-approximately optimal where $\eps \leq 1/100$, so
\begin{align}
L^\ttop(\bx^\ttop) \leq (1+\eps)^2 \cdot L^\ttop(\bx_{t}^*) \leq 200\lambda d^8. \label{eq:loss2}
\end{align} 
This implies that
\[
\frac{1}{100} |10 - \langle \bz^\ttop, \bx^{(t)}\rangle|^2 \leq L^\ttop(\bx^{(t)}) \leq 200 \lambda d^8.
\]
So in this case we always have $|10 - \langle \bz^\ttop, \bx^{(t)}\rangle| < 200 d^4\sqrt{\lambda}$, and this means the termination condition on Line~\ref{line:termination} is equivalent to whether $\|\by^\ttop\|_2 \geq d^3$.

We make the following claim about $\bx^\ttop$, and we defer the proof of this claim to Appendix \ref{sec:fully-app}, as it involves some detailed calculations.
\begin{claim}
\label{claim:decomposition}
Let $\bV_{t}$ be the orthonormal matrix that concatenates $\bU_\perp$ and $\bz_{\bU}^\ttop$, i.e, $\bV_{t} := [\bU_\perp, \frac{\bz_{\bU}^\ttop}{\|\bz_{\bU}^\ttop\|_2}]$. Then we have
\begin{align*}
\bU_\perp(\bU_\perp)^{\top} \bx^\ttop = &~  \bx^{*} \pm 20d^4 \sqrt{\lambda}, \\
\langle \bx^\ttop, \bz_{\bU}^\ttop\rangle = &~  10 - \langle \bz_{\bU_\perp}^\ttop,\bx^{*}\rangle \pm 200d^4 \sqrt{\lambda}, \\
\|(\mathbf{I} - \bV_{t}\bV_{t}^{\top}) \cdot \bx^\ttop\|_2 \leq &~ 2\sqrt{\eps} \cdot \Delta_{t}.
\end{align*}
\end{claim}

Using Claim \ref{claim:decomposition}, we can write $\bx^\ttop$ as
\begin{align*}
\bx^\ttop = &~ \bU_\perp (\bU_\perp)^{\top} \bx^\ttop + \frac{\langle \bx^\ttop, \bz^\ttop_{\bU} \rangle}{\|\bz^\ttop_{\bU}\|_2} \cdot \frac{\bz^\ttop_{\bU}}{\|\bz_{\bU}^\ttop\|_2} + \bz_{\perp\perp}^\ttop \\
= &~ \bx^{*} + \Delta_{t}\cdot  \frac{\bz^\ttop_{\bU}}{\|\bz_{\bU}^\ttop\|_2} + \bz_{\perp\perp}^\ttop \pm O(d^8\sqrt{\lambda}),
\end{align*}
where the first step follows from decomposing $\bx_r^\ttop$ into three parts: the component that is in subspace $\bU_\perp$, the component that is in the same direction as $\bz_{\bU}^\ttop$, and the component that is orthogonal to both $\bU_\perp$ and $\bz^\ttop_{\bU}$ which we denote as $\bz_{\perp\perp}^\ttop := (\mathbf{I} - \bV_{t}\bV_{t}^{\top}) \cdot \bx^\ttop$, the second step follows from the first and second parts of Claim~\ref{claim:decomposition} and that $\|\bz_{\bU}^\ttop\|_2 \geq 1/d^4$, and finally note that using the third part of Claim~\ref{claim:decomposition} we have 
\begin{equation}\label{eq:z_perp_perp_bound}
\|\bz_{\perp\perp}^\ttop\|_2 \leq 2\sqrt{\eps}\cdot \Delta_{t}.
\end{equation}

Consequently, we can write $\by^\ttop$ as
\begin{align}\label{eq:y_rt}
    \by^\ttop = \bx^\ttop - \bx^* = \Delta_{t}\cdot  \frac{\bz^\ttop_{\bU}}{\|\bz_{\bU}^\ttop\|_2} + \bz_{\perp\perp}^\ttop \pm O(d^8\sqrt{\lambda}).
\end{align}

We further divide into two cases based on whether $\|\by^\ttop\|_2 \geq d^3$, i.e., whether the termination condition is satisfied.

\vspace{+2mm}
{\bf \noindent Case 1-1 \ \ } Suppose $\|\by^\ttop\|_2 \geq d^3$. Then it meets the termination condition and we return $\hat{\bz}_{U}^\ttop = \mathbf{0}$. 
In this case, we have
\begin{align*}
d^3 \leq \|\by^\ttop\|_2 = &~ \left\|\Delta_{t}\cdot  \frac{\bz^\ttop_{\bU}}{\|\bz_{ \bU}^\ttop\|_2} + \bz_{\perp\perp}^\ttop\right\|_2  \pm O(d^8\sqrt{\lambda} ) \\
\leq &~\Delta_{t} + \|\bz_{\perp\perp}^\ttop\|_2 \pm O(d^8\sqrt{\lambda} ) \leq (1+2\sqrt{\eps})\Delta_{t}  \pm O(d^8\sqrt{\lambda} ),
\end{align*}
where third step follows from triangle inequality, and the last step follows from Eq.~\eqref{eq:z_perp_perp_bound}. 
Since $\lambda = 1/d^{40}$ and $\eps = 1/100$, we conclude that
\begin{align*}
\frac{1}{2}d^3 \leq \Delta_{t} =  \frac{10 - \langle \bz_{\bU_\perp}^\ttop,\bx^{*}\rangle}{\|\bz_{ \bU}^{(t)}\|_2}
\leq \frac{11}{\|\bz_{\bU}^{(t)}\|_2},
\end{align*}
and therefore, $\|\bz_{\bU}^\ttop\|_2 \leq O(1/d^3)$ and it is fine to return $\hat{\bz}_{\bU}^\ttop = \mathbf{0}$.

\vspace{+2mm}
{\bf \noindent Case 1-2 \ \ } 
Suppose $\|\by^\ttop\|_2 < d^3$. Then the termination condition is not met. 
To compute $\hat{\bz}_\bU^\ttop$, we need to solve Eq.~\eqref{eq:interpolation}. Define $\xi  = \Delta_{t}^{-1} \cdot \|\bz_{\bU}^\ttop\|_2$, and we have
\begin{align}
\|\bz^\ttop - \xi \by^\ttop\|_2^2 = &~ \Big\|\bz_{\bU_\perp}^\ttop - \Delta_{t}^{-1} \cdot \|\bz_{\bU}^\ttop\|_2 \cdot \big( \bz_{\perp\perp}^\ttop \pm O(d^{8}\sqrt{\lambda}) \big) \Big\|_2^2\notag \\
= &~ \|\bz_{\bU_\perp}^\ttop\|_2^2 + \Delta_{t}^{-2} \cdot \|\bz_{\bU}^\ttop\|_2^2 \cdot \|\bz_{\perp\perp}^\ttop\|_2^2 \pm O(d^{8}\sqrt{\lambda}) \label{eq:ub1}.
\end{align}
The first step follows from Eq.~\eqref{eq:y_rt} and $\bz^\ttop = \bz_{\bU}^\ttop + \bz_{ \bU_{\perp}}^\ttop$, the second step holds since $\bz_{\bU_\perp}^\ttop$ is orthogonal to $\bz_{ \perp\perp}^\ttop$, and the error term is still $\pm O(d^8 \sqrt{\lambda})$ since $\|\bz_{\bU}^\ttop\|_2, \|\bz_{\bU_{\perp}}^\ttop\|_2, \|\bz_{\perp\perp}^\ttop\|_2 \leq 1$ and $\Delta_{t} \geq 9$ (Eq.~\eqref{eq:Delta_bound}).

The optimal solution to Eq.~\eqref{eq:interpolation}, denoted as $\xi^{*}$, can be expressed as $\xi^{*} = (1 + \nu)\Delta_t^{-1} \|\bz_\bU^\ttop\|_2$ for some scaling factor $\nu$. Similarly, we have
\begin{align}
\|\bz^\ttop - \xi^{*}\by^\ttop\|_2^2 = &~ \Big\|\bz_{\bU}^\ttop + \bz_{\bU_\perp}^\ttop - (1 + \nu)\Delta_t^{-1}\|\bz_\bU^\ttop\|_2 \cdot \Big(\Delta_{t}\cdot  \frac{\bz^\ttop_{\bU}}{\|\bz_{\bU}^\ttop\|_2} + \bz_{\perp\perp}^\ttop\Big)   \Big\|_2^2 \pm O(d^8\sqrt{\lambda}) \notag \\
= &~ \|\bz_{\bU_\perp}^\ttop\|_2^2 + \nu^2 \|\bz_{\bU}^\ttop\|_2^2 + (1+\nu)^2 \Delta_{t}^{-2} \cdot \|\bz_{\bU}^\ttop\|_2^2 \cdot \|\bz_{\perp\perp}^\ttop\|_2^2 \pm O(d^{8}\sqrt{\lambda}),\label{eq:ub2}
\end{align}
where the first step comes from Eq.~\eqref{eq:y_rt}, the second step follows from $\bz_{\bU_{\perp}}^\ttop$, $\bz_{\bU}^\ttop$ and $\bz^\ttop_{\perp\perp}$ are orthogonal to each other.

Combining Eq.~\eqref{eq:ub1}\eqref{eq:ub2} and the fact that $\xi^{*}$ is the optimal solution to Eq.~\eqref{eq:interpolation}, we have that
\begin{align*}
0 \geq &~ \|\bz^\ttop - \xi^{*} \by^\ttop\|_2^2 - \|\bz^\ttop - \xi\by^\ttop\|_2^2 \\
= &~ \nu^2 \|\bz_{\bU}^\ttop\|_2^2 + (1+\nu)^2 \Delta_{t}^{-2} \cdot \|\bz_{\bU}^\ttop\|_2^2 \cdot \|\bz_{\perp\perp}^\ttop\|_2^2 - \Delta_{t}^{-2} \cdot \|\bz_{\bU}^\ttop\|_2^2 \cdot \|\bz_{\perp\perp}^\ttop\|_2^2 \pm O(d^{8}\sqrt{\lambda})\\
\geq &~ \nu^2 \|\bz_{\bU}^\ttop\|_2^2 - 2|\nu| \cdot \Delta_{t}^{-2} \cdot \|\bz_{\bU}^\ttop\|_2^2 \cdot \|\bz_{\perp\perp}^\ttop\|_2^2  \pm O(d^{8}\sqrt{\lambda})\\
\geq &~ \nu^2 \|\bz_{\bU}^\ttop\|_2^2 - 2|\nu| \cdot 4\eps \cdot \|\bz_{\bU}^\ttop\|_2^2 \pm O(d^{8}\sqrt{\lambda}).
\end{align*}
where the the last step holds due to $\|\bz_{\perp\perp}\|_2 \leq 2\sqrt{\eps} \cdot \Delta_{t}$ (see Eq. \eqref{eq:z_perp_perp_bound}).

Combining the fact that $\|\bz_\bU^\ttop\|\geq 1/d^4$ and choice of parameters, we conclude that $|\nu| \leq 9\eps$. Therefore, the output $\hat{\bz}_{\bU}^\ttop = \xi^{*}\by^\ttop$ satisfies
\begin{align*}
\|\hat{\bz}_{\bU}^\ttop - \bz_\bU^\ttop\|_2 = &~ \|\xi^{*} \cdot \by^\ttop - \bz_\bU^\ttop\|_2\\
= &~ \Big\| (1 + \nu)\Delta_t^{-1}\|\bz_\bU^\ttop\|_2 \cdot \Big(\Delta_{t}\cdot  \frac{\bz^\ttop_{\bU}}{\|\bz_{\bU}^\ttop\|_2} + \bz_{\perp\perp}^\ttop\Big)  - \bz_\bU^\ttop \Big \|_2 \pm O(d^8\sqrt{\lambda})\\
\leq &~ |\nu| \cdot \|\bz_{\bU}^\ttop\|_2 + |1+\nu| \cdot \Delta_t^{-1}\|\bz_\bU^\ttop\|_2 \|\bz_{\perp\perp}^\ttop\|_2 \pm O(d^8\sqrt{\lambda})\\
\leq &~ (|\nu| + 2\sqrt{\eps}|1+\nu|)\|\bz_{\bU}^\ttop\|_2 \pm O(d^8\sqrt{\lambda})\\
\leq &~ \alpha \|\bz_{\bU}^\ttop\|_2 + 1/d^3.
\end{align*}
Here the second step follows from the Eq.~\eqref{eq:y_rt} and the choice of $\xi^{*}$, the third step follows from triangle inequality, the fourth step holds from $\|\bz_{\perp\perp}\|_2 \leq 2\sqrt{\eps}\Delta_{t}$ (see Eq.~\eqref{eq:z_perp_perp_bound}), and the last step follows from $\alpha = 1/3$ and $|\nu| \leq 9 \epsilon$ that we just proved.
This verifies that $\hat{\bz}_{\bU}^\ttop$ is indeed an $(\alpha, 1/d^3)$-approximation of the projection $\bz_{\bU}^\ttop$.

\vspace{+2mm}
{\bf \noindent Case 2 \ \ } Suppose $\|\bz_{\bU}^\ttop\|_2 <  1/d^4$. It suffices to prove the termination condition on Line~\ref{line:termination} of Algorithm~\ref{algo:reduction} holds, i.e., either $\|\by^\ttop\|_2 \geq d^3$ or $|10 - \langle \bz^\ttop, \bx^{(t)}\rangle | \geq 200 d^4\sqrt{\lambda}$. 

Suppose on the contrary that $\|\by^\ttop\|_2 < d^3$ and $|10 - \langle \bz^\ttop, \bx^{(t)}\rangle | < 200 d^4\sqrt{\lambda}$. Then we have
\begin{align*}
10 - O(d^4 \sqrt{\lambda}) \leq  &~ \langle \bz^\ttop , \bx^\ttop\rangle =  \langle \bz_{\bU_\perp}^\ttop , \bx^\ttop\rangle + \langle \bz_{\bU}^\ttop, \by^\ttop \rangle\\
\leq &~ \langle \bz_{\bU_\perp}^\ttop , \bx^\ttop\rangle + \|\bz_{\bU}^\ttop\|_2 \cdot \| \by^\ttop \|_2 \\
\leq &~ \langle \bz_{\bU_\perp}^\ttop , \bx^\ttop\rangle + (1/d^4) \cdot d^3 \\
\leq &~ \|\bz_{\bU_\perp}^\ttop\|_2 \cdot \|\bU_\perp (\bU_\perp)^\top\bx^\ttop \|_2 + 1/d \\
\leq &~ \|\bU_\perp (\bU_\perp)^\top\bx^\ttop \|_2 + 1/d,
\end{align*}
where second step follows from $\by^{(t)} = \bx^{(t)} - \bx^*$ and $\bx^*$ is in the subspace $\bU_{\perp}$, the fourth step follows from the assumptions $\|\bz_{\bU}^\ttop\|_2 <  1/d^4$ and $\|\by^\ttop\|_2 < d^3$, the fifth step holds since $\bz_{\bU_\perp}^\ttop$ lies in the span of $\bU_\perp$, and the last step follows from $\|\bz_{\bU_\perp}^\ttop\|_2 \leq 1$. We conclude that 
\begin{align*}
\|(\bU_\perp)^{\top}\bx^\ttop\|_2 = \|\bU_\perp(\bU_\perp)^{\top}\bx^\ttop\|_2  \geq 9.
\end{align*}
This means 
\[
L^\ttop(\bx^\ttop) \geq \|(\bU_\perp)^\top \bx^\ttop - \frac{1}{\sqrt{d}}\mathbf{1}_{d-d_1}\|^2_2 \geq \big(\|(\bU_\perp)^{\top}\bx^\ttop\|_2 - \|\frac{1}{\sqrt{d}}\mathbf{1}_{d-d_1}\|_2\big)^2 \geq 60.
\] 
This cannot happen because by the definition that $\bx^{*} = \frac{1}{\sqrt{d}} \sum_{j=1}^{d-d_1}\bU_{\perp,j}$, we have
\begin{align*}
L^\ttop(\bx^{*}) = &~ \left\|(\bU_\perp)^{\top}\bx^{*} - \frac{1}{\sqrt{d}} \cdot \mathbf{1}_{d-d_1}\right\|_2^2 + \frac{1}{100}|\langle \bz^\ttop, \bx^{*}\rangle - 10|^2 + \lambda \|\bx^{*}\|_2^2 \\
\leq &~ 0 + \frac{121}{100} +  \lambda \leq 2,
\end{align*}
and this contradicts with $\bx^{\ttop}$ being an $\eps$-approximate solution of the square root of the loss $L^\ttop$. We conclude the proof here.
\end{proof}

Finally, we note that the reduction of Algorithm \ref{algo:reduction} involves one insertion, one deletion and one calls of the dynamic $\eps$-LSR. 
The extra computation it takes is $O(d)$, so it reduces online projection to dynamic $\eps$-LSR.

Combining Lemma \ref{lem:online-projection}, Lemma \ref{lem:hard-amplification} and Lemma \ref{lem:reduction-lsr}, we can finish the proof of Theorem \ref{thm:lower-full}.

\section{Partially dynamic LSR with incremental updates}
\label{sec:upper}
In this section, we present an algorithm for the partially dynamic least-squares regression problem with \emph{incremental updates}.

\begin{theorem}[Partially dynamic LSR with incremental updates, formal version of Theorem \ref{thm:main_UB_informal}]\label{thm:upper}
Let $d, T \in \mathbb{N}$ and $0< \eps, \delta <1/8$.
Assume the least singular value of $\bM^{(0)}$ is at least $\sigma_{\min}$, and the largest singular value of $\bM^{(T)}$ is at most $\sigma_{\max}$. 
For partially dynamic least-squares regression incremental updates, there exists a randomized algorithm (Algorithm \ref{algo:preprocess}--\ref{algo:update-member}) that with probability at least $1-\delta$, maintains an $\eps$-approximation solution for all iterations $t\in [T]$. 
For oblivious adversary, the total update time is at most 
\[
O\Big(\nnz(\bA^{(T)}) \log(\frac{T}{\delta}) + \epsilon^{-4} d^3 \log^2(\frac{\sigma_{\max}}{\sigma_{\min}}) \log^3(\frac{T}{\delta})\Big),
\]  
For adaptive adversary, the total update time is at most 
\[
O\Big(\nnz(\bA^{(T)}) \log(\frac{T}{\delta}) + \epsilon^{-4} d^5 \log^4(\frac{\sigma_{\max}}{\sigma_{\min}}) \log^3(\frac{T}{\delta}) \Big).
\]  
\end{theorem}

\subsection{Data structure}\label{sec:data_structure}
A complete description of our data structure can be found at Algorithm \ref{algo:preprocess}--\ref{algo:update-member}.
Our approach follows the online row sampling framework \cite{cmp20}. 
When a new row arrives, we sample and keep the new row with probability proportional to the online leverage score, which is approximately computed using JL embedding (Algorithm \ref{algo:sample} Line \ref{line:levarage_score}).
If the row is sampled, then we update the data structure (Algorithm \ref{algo:update-member} Line \ref{line:Delta_H}--\ref{line:update-N}, Line \ref{line:update-G}--\ref{line:update-x}) using Woodbury identity, and instantiate a new JL sketch (Line \ref{line:new-JL}--\ref{line:update-wt_B}. See below for the JL lemma); otherwise, we do not perform any updates.

\begin{lemma}[Johnson-Lindenstrauss Lemma \cite{jl84}]\label{lem:JL}
There exists a function $\textsc{JL}(n,m,\epsilon, \delta)$ that returns a random matrix $\bJ \in \R^{k \times n}$ where $k = O(\epsilon^{-2} \log(m/\delta))$, and $\bJ$ satisfies that for any fixed $m$-element subset $V \subset \R^n$,
\begin{align*}
    \Pr\big[\forall \bv \in V, ~ (1 - \epsilon) \|\bv\|_2 \leq \|\bJ \bv\|_2 \leq (1 + \epsilon) \|\bv\|_2\big] \geq 1 - \delta.
\end{align*}
Furthermore, the function $\textsc{JL}$ runs in $O(kn)$ time.
\end{lemma}

\paragraph{Notation} We use superscripts $^{(t)}$ to denote the matrix/vector/scalar maintained by the data structure at the end of the $t$-th iterations. In particular, the superscript $^{(0)}$ represents the variables after the preprocessing step.

\begin{algorithm}[!htbp]
\caption{\textsc{Preprocess} ($\bA$, $\bb$, $\eps$, $\delta$, $T$)}
\label{algo:preprocess}
\begin{algorithmic}[1]
    \State $\bM \leftarrow [\bA, \bb]$ \Comment{Input matrix $\bM \in \R^{(d+1) \times (d+1)}$}
    \State $\bD \leftarrow \bI_{d+1}$  \Comment{Sampling matrix $\bD \in \R^{(d+1)\times(d+1)}$}
    \State $s \leftarrow d+1$ \Comment{The number of sampled rows}
    \State $\bN \leftarrow \bD \cdot \bM$ \Comment{Sampled rows $\bN \in \R^{s \times (d+1)}$} 
    \State $\bH \leftarrow ((\bN)^{\top} \bN )^{-1}$ \label{line:H_0} \Comment{$\bH \in \R^{(d+1) \times (d+1)}$}
    \State $\bB \leftarrow \bN \cdot \bH$ \Comment{$\bB \in \R^{s \times (d+1)}$}
    \State $\bJ \leftarrow \textsc{JL}(s, T, \frac{1}{100}, \frac{\delta}{2 T^2})$ \Comment{JL embedding $\bJ \in \R^{O(\log(T/\delta)) \times s}$}
    \State $\wt{\bB} \leftarrow \bJ \cdot \bB$ \label{line:wt_B_0} \Comment{Used for online LS estimation $\wt{\bB}\in \R^{O(\log(T/\delta)) \times (d+1)}$}
    \State $\bG \leftarrow (\bA^{\top}\bD^2 \bA)^{-1}$ \label{line:G_0} \Comment{$\bG \in \R^{d \times d}$}
    \State $\bu \leftarrow \bA^{\top} \bD^2 \bb$ \Comment{$\bu \in \R^{d}$}
    \State $\bx \leftarrow \bG\cdot \bu$ \Comment{(Approximate) solution $\bx \in \R^d$}
\end{algorithmic}
\end{algorithm}

\begin{algorithm}[!hbtp]
\caption{\textsc{Insert} ($\ba, \beta$) \Comment{Insert a new row $(\ba, \beta)\in \R^d \times \R$}}
\label{algo:update}
\begin{algorithmic}[1]
\State $\boldm \leftarrow [\ba^{\top}, \beta]^{\top}$ \Comment{$\boldm \in \R^{d+1}$} 
\State $\nu \leftarrow \textsc{Sample}(\boldm)$\Comment{$\nu \in \R$}\label{line:sample_in_update}
\State $\bD \leftarrow
\begin{bmatrix}
\bD & 0 \\
0 & \nu
\end{bmatrix} 
$
\State \textbf{if} $\nu \neq 0$ \textbf{then} \textsc{UpdateMembers}($\boldm$)  \label{line:if_start}
\State \Return $\bx$
\end{algorithmic}
\end{algorithm}

\begin{algorithm}[!htbp]
\caption{\textsc{Sample} ($\boldm$)}
\label{algo:sample}
\begin{algorithmic}[1]
\State $C_{\mathrm{obl}} \leftarrow 10 \epsilon^{-2} \log(2T/\delta)$, $C_{\mathrm{adv}} \leftarrow 32 (1 + \epsilon) d \log(\frac{\sigma_{\max}}{\sigma_{\min}}) \cdot C_{\mathrm{obl}}$
\State $\tau \leftarrow \|\wt{\bB} \cdot \boldm\|_2^2$ \Comment{Approximate online LS} \label{line:approximate-ols}
\label{line:levarage_score}
\If{\texttt{Oblivious adversary}} \Comment{Oblivious adversary}
\State $p \leftarrow \min\{C_{\mathrm{obl}} \cdot \tau, 1\}$
\Else \Comment{Adaptive adversary}
\State $p \leftarrow \min\{C_{\mathrm{adv}} \cdot \tau, 1\}$
\EndIf 
\label{line:prob}
\State $\nu \leftarrow 1/\sqrt{p}$ with probability $p$, and $\nu \leftarrow 0$ otherwise\label{line:sample}
\end{algorithmic}
\end{algorithm}

\begin{algorithm}[!htbp]
\caption{\textsc{UpdateMembers} ($\boldm$)}
\label{algo:update-member}
\begin{algorithmic}[1]
\Statex  \texttt{// Update spectral approximation}
\State $s \leftarrow s + 1$ \Comment{The number of sampled rows}
\State $\Delta \bH \leftarrow - \frac{\bH \boldm \boldm^{\top} \bH / p}{1 + \boldm^{\top} \bH \boldm / p}$  \label{line:Delta_H} 
\State $\bH \leftarrow \bH +\Delta \bH$ \Comment{Update $\bH \in \R^{(d+1) \times (d+1)}$}
\State $\bB\leftarrow [ ( \bB + \bN \cdot \Delta \bH )^{\top}, ~ \bH \cdot \boldm / \sqrt{p}]^{\top}$\label{line:update-B} \Comment{Update $\bB \in \R^{s \times (d+1)}$}
\State $\bN \leftarrow [\bN^{\top}, \boldm / \sqrt{p}]^{\top}$ \label{line:update-N} \Comment{Update $\bN \in \R^{s \times (d+1)}$}
\State $\bJ \leftarrow \textsc{JL}(s, T, \frac{1}{100}, \frac{\delta}{2 T^2})$ \Comment{Instantiate a new JL sketch, $\bJ \in \R^{O(\log(T/\delta)) \times s}$}\label{line:new-JL}
\State $\wt{\bB}\leftarrow \bJ \cdot \bB$\label{line:update-wt_B} \Comment{Update $\wt{\bB} \in \R^{O(\log(T/\delta)) \times (d+1)}$}
\Statex  \texttt{// Update solution }
\State $\bG \leftarrow \bG - \frac{\bG \ba \ba^{\top} \bG / p}{1 + \ba^{\top} \bG \ba / p}$ \label{line:update-G}
\Comment{Woodbury identity, update $\bG \in \R^{d \times d}$}
\State $\bu \leftarrow \bu + \beta \cdot \ba/ p$\label{line:update-u}\Comment{Update $\bu \in \R^d$}
\State $\bx \leftarrow \bG\cdot \bu$ \Comment{Update $\bx \in \R^d$} \label{line:update-x}
\end{algorithmic}
\end{algorithm}

We summarize all variables maintained by our data structure and their closed-form formulas. The proof can be found in Appendix \ref{sec:upper-app}.

\begin{lemma}[Closed-form formulas]
\label{lem:close_form_formula_algorithm}
At the $t$-th iteration of $\textsc{Insert}$ (Algorithm~\ref{algo:update}), we have
\begin{enumerate}
    \item $\bM^{(t)} = [\bA^{(t)}, \bb^{(t)}] \in \R^{(d+t+1) \times (d+1)}$ is the input matrix.
    \item $\bD^{(t)} \in \R^{(d+t+1) \times (d+t+1)}$ is a diagonal matrix with $s^{(t)}$ non-zero entries.
    \item $\bN^{(t)} = (\bD^{(t)} \bM^{(t)})_{S^{(t)}, *} \in \R^{s^{(t)} \times (d+1)}$ takes rows of $\bM^\ttop$, where $S^{(t)} \subset [d+t+1]$ is the set of non-zero entries of $\bD^{(t)}$.
    \item $\bH^{(t)} = \big( (\bN^{(t)})^{\top} \bN^{(t)} \big)^{-1} \in \R^{(d+1) \times (d+1)}$.
    \item $\bB^{(t)} = \bN^{(t)} \bH^{(t)} \in \R^{s^{(t)} \times (d+1)}$.
    \item $\wt{\bB}^{(t)} = \bJ^{(t)} \cdot \bB^{(t)} \in \R^{O(\log(T/\delta)) \times (d+1)}$ is used for approximately estimating the online leverage score.
    \item $\bG^{(t)} = \big( (\bA^{(t)})^{\top} (\bD^{(t)})^2 \bA^{(t)} \big)^{-1} \in \R^{d \times d}$.
    \item $\bu^{(t)} = (\bA^{(t)})^{\top}(\bD^{(t)})^2 \bb^{(t)} \in \R^d$.
    \item $\bx^{(t)} = \big( (\bA^{(t)})^{\top} (\bD^{(t)})^2 \bA^{(t)} \big)^{-1} \cdot (\bA^{(t)})^{\top} (\bD^{(t)})^2 \bb^{(t)} \in \R^d$ is the maintained solution. 
\end{enumerate}
\end{lemma}

\subsection{Warm up: Analysis for oblivious adversary}
\label{sec:correct-oblivious}

We first prove the correctness against an oblivious adversary (i.e., our data structure maintains an $\eps$-approximate solution w.h.p.) and the runtime analysis is deferred to Section \ref{sec:time}.
The proof follows easily from the guarantee of online leverage score sampling \cite{cmp20} and the JL sketch, and it serves as a warm up for the more complicated algorithm against an adaptive adversary.
As we shall see later, both guarantees become nontrivial when facing an adaptive adversary.

The key advantage for the oblivious setting is that we can fix the input sequence $\boldm^{(1)}, \ldots, \boldm^{(T)}$ for analysis. We exploit the following guarantee of online leverage score sampling, which is a direct corollary from matrix Freedman inequality. For completeness we include a proof in Appendix~\ref{sec:upper-app}.
\begin{lemma}[Online leverage score sampling, adapted from Lemma 3.3 of \cite{cmp20}]\label{lem:spectral-online-leverage-score}
Let $\epsilon, \delta \in (0,1/2)$ be two parameters. Let $\boldm^{(1)}, \ldots, \boldm^{(T)} \in \R^{d+1}$ be a fixed sequence and let $\tauo^\ttop$ be the online leverage score of the $t$-th row, i.e., $\tauo^\ttop:= (\boldm^\ttop)^\top ((\bM^{(t-1)})^\top\bM^{(t-1)})^{-1}\boldm^\ttop$. 
Suppose an algorithm samples the $t$-th row with probability\footnote{The sampling probability could depend on the result of previous sampling outcomes.}
\begin{align*}
p_t \geq \min\{3 \epsilon^{-2} \tauo^\ttop \log(d / \delta), 1\}.
\end{align*}
Define $\nu_t \in \R$ as
\begin{align*}
    \nu_t = 
    \begin{cases}
    \frac{1}{\sqrt{p_t}}, & \text{if the } t\text{-th row is sampled},  \\
    0, & \text{otherwise.}
    \end{cases}
\end{align*}
Then with probability at least $1 - \delta$, $(\bM^{(0)})^\top \bM^{(0)} + \sum_{t=1}^{T}\nu_t^2\cdot \boldm^\ttop (\boldm^\ttop)^\top$ is an $\epsilon$-spectral approximation of $(\bM^{(T)})^\top \bM^{(T)}$.
\end{lemma}

Our data structure maintains an $\eps$-spectral approximation of $(\bM^{(t)})^\top \bM^{(t)}$ and uses it to approximate the online leverage score.

\begin{lemma}[Spectral approximation]
\label{lem:sampling_probability}
With probability at least $1 - \delta/T$, for any $t \in [T]$,
\begin{align}\label{eq:tau_leverage_score}
    0.9 (1-\eps) \cdot \tauo^\ttop \leq \tau^{(t)} \leq 1.1 (1+\eps) \cdot \tauo^\ttop
\end{align}
and for any $t \in [0: T]$ 
\begin{align}\label{eq:spectral_approximation}
    (\bM^{(t)})^{\top} (\bD^{(t)})^2 \bM^{(t)} \approx_{\epsilon} (\bM^{(t)})^{\top} \bM^{(t)}.
\end{align}
\end{lemma}
\begin{proof}
We prove the claim inductively. Let $\delta' = \frac{\delta}{2T^2}$, the induction hypothesis is that with probability $1 - 2 t \delta'$, Eq.~\eqref{eq:spectral_approximation} holds for all $t' \in [0:t]$ and Eq.~\eqref{eq:tau_leverage_score} holds for all $t' \in [t]$. 
The base case $t=0$ holds trivially as $\bD^{(0)} = \bI_{d+1}$. Suppose the induction hypothesis continues to hold for $t-1$, then for the $t$-th iteration, we have
\begin{align}
    \|\bB^{(t-1)} \cdot \boldm^{(t)}\|_2^2 = &~ (\boldm^{(t)})^{\top} \cdot (\bH^{(t-1)})^{\top} (\bN^{(t-1)})^{\top} \bN^{(t-1)} \bH^{(t-1)} \cdot \boldm^{(t)} \notag \\
    = &~ (\boldm^{(t)})^{\top} \cdot \big((\bN^{(t-1)})^{\top} \bN^{(t-1)} \big)^{-1} \cdot \boldm^{(t)} \notag \\
    = &~ (\boldm^{(t)})^{\top} \cdot \big( (\bM^{(t-1)})^{\top} (\bD^{(t-1)})^2 \bM^{(t-1)} \big)^{-1} \cdot \boldm^{(t)} \notag \\
    = &~ (1 \pm \epsilon) \cdot (\boldm^{(t)})^{\top} \cdot \big( (\bM^{(t-1)})^{\top} \bM^{(t-1)} \big)^{-1} \cdot \boldm^{(t)} \notag \\
    = &~ (1 \pm \epsilon) \cdot \tauo^\ttop \label{eq:jl1}
\end{align}
The first three steps follow from the closed-form formula (Lemma \ref{lem:close_form_formula_algorithm}), the fourth step holds due to the induction hypothesis and the last step comes from the definition of online leverage score.

Meanwhile, using the JL Lemma (Lemma~\ref{lem:JL}), we have that with probability at least $1 - \delta'$, 
\begin{align}
    \tau^\ttop = \|\wt{\bB}^{(t-1)} \cdot \boldm^{(t)}\|_2^2  = \|\mathbf{J}^{(t-1)} \bB^{(t-1)} \cdot \boldm^{(t)}\|_2^2 = (1\pm 0.1) \|\bB^{(t-1)} \cdot \boldm^\ttop\|_2^2\label{eq:jl2}
\end{align}
where the last step holds due to the JL Lemma and $\boldm^\ttop$, $\bB^{(t-1)}$ are independent of the entries of $\bJ^{(t-1)}$.

Combining Eq.~\eqref{eq:jl1}\eqref{eq:jl2}, we finish the induction of Eq.~\eqref{eq:tau_leverage_score}. The spectral approximation guarantee of Eq.~\eqref{eq:spectral_approximation} follows directly from Lemma \ref{lem:spectral-online-leverage-score} and our choice of parameters. We finish the proof here.\qedhere
\end{proof}

It is well known that spectral approximations of $(\bM^{(t)})^{\top} \bM^{(t)}$ give approximate solutions to least squares regressions \cite{w14}, so we have proved the correctness of our algorithm.
\begin{lemma}[Correctness of Algorithm \ref{algo:preprocess}--\ref{algo:update-member}, oblivious adversary]\label{lem:correctness_algorithm}
With probability at least $1 - \delta/T$, in each iteration, \textsc{Insert} of Algorithm~\ref{algo:update} outputs a vector $\bx^{(t)} \in \R^d$ such that 
    \[
    \|\bA^{(t)} \bx^{(t)} - \bb^{(t)}\|_2 \leq (1+\eps) \min_{\bx \in \R^d} \|\bA^{(t)} \bx - \bb^{(t)} \|_2.
    \]
against an oblivious adversary.
\end{lemma}
\begin{proof}
Note by Lemma~\ref{lem:close_form_formula_algorithm}, one has $\bx^{(t)} = \big( (\bA^{(t)})^{\top} (\bD^{(t)})^2 \bA^{(t)} \big)^{-1} \cdot (\bA^{(t)})^{\top} (\bD^{(t)})^2 \bb^{(t)}$, which is the closed-form minimizer of $\|\bD^{(t)} \bA^{(t)} \bx - \bD \bb^{(t)}\|_2$.
From Eq.~\eqref{eq:spectral_approximation} in Lemma~\ref{lem:sampling_probability}, we know that with probability at least $1 - \delta/T$, we have $(\bM^{(t)})^{\top} (\bD^{(t)})^2 \bM^{(t)} \approx_{\epsilon} (\bM^{(t)})^{\top} \bM^{(t)}$. 
Using this spectral approximation and Lemma~\ref{lem:approx_l2_regression_from_spectral_approx}, we conclude the proof of this lemma.
\end{proof}

\subsection{Analysis for adaptive adversary}\label{sec:upper_robust}

Next, we prove our data structure (Algorithm~\ref{algo:preprocess} --\ref{algo:update-member}) is adversarially robust and it works against an adaptive adversary when using a larger sampling constant $C_{\mathrm{adv}} = 32 (1 + \epsilon) d \log(\frac{\sigma_{\max}}{\sigma_{\min}}) \cdot C_{\mathrm{obl}}$.

First, we prove that online leverage score sampling works against an adaptive adversary. In contrast with the counterpart Lemma \ref{lem:sampling_probability} of the oblivious setting, the sequence $\boldm^{(1)}, \ldots, \boldm^{(T)}$ is not fixed but chosen adaptively based on previous outcomes. The proof becomes more challenging due to this adaptivity.
\begin{lemma}[Intrinsic robustness of online leverage score sampling]
\label{lem:intrinsic_new}
Let $\eps, \delta \in (0, 1/8)$. 
Let $\boldm^{(1)}, \ldots, \boldm^{(T)} \in \R^{d+1}$ be an adaptive sequence of row vectors chosen by an adaptive adversary and let $\tauo^\ttop$ be the online leverage score of the $t$-th row, i.e., $\tauo^\ttop = (\boldm^{(t)})^\top ((\bM^{(t-1)})^\top \bM^{(t-1)})^{-1} \boldm^{(t)}$. If an algorithm samples the $t$-th row with probability 
\[
p_{t} \geq \min\{\alpha \cdot \tauo^\ttop, 1\}, \text{ where } \alpha = 300 d \epsilon^{-2} \log(\frac{400d \sigma_{\max}}{\epsilon \delta \sigma_{\min} }),
\]
that is,
\begin{align*}
\bX^\ttop =  \left\{
\begin{matrix}
\frac{1}{p_t} \cdot  \boldm^\ttop (\boldm^\ttop)^\top & \text{w.p. } p_t,\\
\mathbf{0} & \text{w.p. } 1 - p_t.\\
\end{matrix}
\right.
\end{align*}
Then with probability at least $1-\delta$, its output $\bY := (\bM^{(0)})^\top \bM^{(0)} + \sum_{t=1}^{T}\bX^{(t)}$ is an $\eps$-spectral approximation of $(\bM^{(T)})^\top \bM^{(T)}$.
\end{lemma}

We note Lemma \ref{lem:intrinsic_new} improves Lemma B.2 of \cite{bhm+21}. Concretely, Lemma B.2 in \cite{bhm+21} has an $\wt{O}(d (\frac{\sigma_{\max}}{\sigma_{\min}})^2)$ overhead comparing with the oblivious case, while we reduce the overhead to $\wt{O}(d \log (\frac{\sigma_{\max}}{\sigma_{\min}}))$, which has only polylogarithmic dependence on the condition number $\frac{\sigma_{\max}}{\sigma_{\min}}$.

We make use of the Freedman's inequality for martingales. 
\begin{lemma}[Freedman's inequality, \cite{freedman1975tail}]\label{thm:freedman}
Consider a martingale $Y_0, Y_1, \cdots, Y_n$ with difference sequence $X_1, X_2, \cdots, X_n$, i.e., $Y_0 = 0$, and for all $i \in [n]$, $Y_i = Y_{i-1} + X_i$ and $\E_{i-1}[Y_i] = Y_{i-1}$. Suppose $|X_i| \leq R$ almost surely for all $i \in [n]$. Define the predictable quadratic variation process of the martingale as $W_i = \sum_{j=1}^{i} \E_{j-1}[X_j^2]$, for ill $i \in [n]$. Then for all $u \geq 0$, $\sigma^2 > 0$,
\[
\Pr\left[ \exists i \in [n]: |Y_i| \geq u \text{ and } W_i \leq \sigma^2 \right] \leq 2 \exp\left(-\frac{u^2/2}{\sigma^2 + R u / 3}\right)
\]
\end{lemma}

We now turn to the proof of Lemma \ref{lem:intrinsic_new}.
The first step is similar to \cite{bhm+21}, we take an union bound over the $\eps$-net of unit vectors $\bx\in \R^{d+1}$ and reduce to the scalar case (the union bound gives an $\wt{O}(d)$ overhead).
When applying Freedman's inequality, the term $\|\bM^{(T)} \bx\|_2$ shows up, and it is unknown since the rows of $\bM^{(T)}$ are chosen adaptively. 
\cite{bhm+21} uses a straightforward bound of $\sigma_{\min} \|\bx\|_2 \leq \|\bM^{(T)} \bx\|_2 \leq \sigma_{\max} \|\bx\|_2$, which results in another $O((\frac{\sigma_{\max}}{\sigma_{\min}})^2)$ overhead. 
We instead consider $O(\frac{\sigma_{\min}}{\sigma_{\max}})$ number of (truncated) martingales and prove that one of them correctly guesses the value of $\|\bM^{(T)} \bx\|_2$. By taking an union bound over these $O(\frac{\sigma_{\min}}{\sigma_{\max}})$ martingales, it only has an $O(\log(\frac{\sigma_{\max}}{\sigma_{\min}}))$ overhead.

\begin{proof}[Proof of Lemma~\ref{lem:intrinsic_new}]
We first introduce some notations. 
Let 
\[
\kappa = \frac{\sigma_{\max}}{\sigma_{\min}}, \quad \delta' = \delta \cdot \left(\frac{200 \kappa d}{\epsilon}\right)^{-d}, \quad \text{and} \quad \delta'' = \frac{\delta'}{10 \kappa}.
\]
For any $t \in [T]$, let $\mathcal{F}_{t}$ be the $\sigma$-algebra generated by the adaptive sequence $\boldm^{(1)}, \cdots, \boldm^{(t+1)}$ and $\bX^{(1)}, \cdots, \bX^{(t)}$. Note that $\mathcal{F}_0 \subseteq \mathcal{F}_1 \subseteq \cdots \subseteq \mathcal{F}_T$ is a filtration. We use the notation $\E_{t}[\cdot] = \E[\cdot \mid \mathcal{F}_t]$ to denote the expectation conditioned on $\mathcal{F}_t$.

\vspace{+2mm}
{\bf Step 1.} The key step is to prove that for any fixed vector $\bx \in \R^d$, with probability at least $1 - \delta''$, 
\begin{align}
|\bx^{\top} \bY \bx - \|\bM^{(T)} \bx\|_2^2 | \leq \frac{\epsilon}{2} \cdot \|\bM^{(T)}\bx\|_2^2.\label{eq:upper-goal1} 
\end{align}
To this end, define a set 
\[
\mathcal{S} := \Big\{k \cdot \frac{\sigma_{\min}}{10} \|\bx\|_2 ~\Big|~ k \in \mathbb{Z}_{\geq 1}, \text{ and } \sigma_{\min} \leq k \cdot \frac{\sigma_{\min}}{10} \leq \sigma_{\max}\Big\}.
\]
The size of $\mathcal{S}$ is $|\mathcal{S}| = 10 \kappa$. 

For any value $s \in \mathcal{S}$, define the random sequence $\{\overline{x}_{s}^\ttop\}_{t\in [T]}$
\begin{align}
\overline{x}_s^{(t)} =
\begin{cases}
\bx^{\top} \bX^{(t)} \bx - (\bx^{\top} \boldm^{(t)})^2 & \text{if } \|\bM^{(t)} \bx\|_2 \leq s, \\
0 & \text{otherwise.}
\end{cases} \label{eq:definition-ox}
\end{align}
and  $\{\overline{y}_{s}^\ttop\}_{t\in [0:T]}$
\[
\overline{y}_s^{(0)} = 0\quad \text{and} \quad \overline{y}_s^{(t)} = \overline{y}_s^{(t-1)} + \overline{x}_s^{(t)},\quad \forall t \in [T].
\]

The sequence $\{\overline{y}_s^{(t)}\}_{t\in [0:T]}$ forms a martingale. To see this, by the definition of $\bX^{(t)}$, we have $\E_{t-1}[\bX^{(t)}] = \boldm^{(t)} (\boldm^{(t)})^{\top}$ and $\E_{t-1}[\bx^{\top} \bX^{(t)} \bx] = (\bx^{\top} \boldm^{(t)})^2$. This means $\E_{t-1}[\overline{x}_s^{(t)}] = 0$, and therefore
\[
\E_{t-1}[\overline{y}_s^{(t)}] = \E_{t-1}[\overline{y}_s^{(t-1)} + \overline{x}_s^{(t)}] = \overline{y}_s^{(t-1)}.
\]

We wish to apply the Freedman inequality to $\{\overline{y}_s^\ttop\}_{t \in [T]}$, and we bound the maximum deviation and the variance separately.
\begin{itemize}
\item First, we prove $|\overline{x}_s^{(t)}| < \frac{s^2}{\alpha}$. 

If $\|\bM^{(t)} \bx\|_2 > s$, then  $\overline{x}_s^{(t)} = 0$ due to the definition in Eq.~\eqref{eq:definition-ox}.
If $p_t = 1$, then we have $\bX^{(t)} = \boldm^{(t)} (\boldm^{(t)})^{\top}$ with probability $1$, and therefore, $\overline{x}_s^{(t)} = \bx^{\top} \bX^{(t)} \bx - (\bx^{\top} \boldm^{(t)})^2 = 0$.
Finally, suppose $\|\bM^{(t)} \bx\|_2 \leq s$ and $p_t < 1$, we have $p_t \geq \alpha \cdot \tauo^\ttop$, and
\begin{align*}
|\overline{x}_s^{(t)}| \leq \frac{1}{p_t} \cdot (\bx^{\top} \boldm^{(t)})^2 \leq \frac{1}{\alpha \cdot \tauo^\ttop} \cdot (\bx^{\top} \boldm^{(t)})^2 
\leq \frac{\|\bM^{(t)} \bx\|_2^2}{\alpha} \leq \frac{s^2}{\alpha},
\end{align*}
where the first step follows from the definition of $\overline{x}_s^{(t)}$ and $0 < \frac{1}{p_t} - 1 < \frac{1}{p_t}$, the third step follows from the property of online leverage score (Fact \ref{fact:online_leverage_score}).
We conclude with $|\overline{x}_s^{(t)}| \leq \frac{s^2}{\alpha}$.

\item Then, we bound the variance. Similarly, if $\|\bM^{(t)} \bx\|_2 > s$ or $p_t = 1$, then we have $\overline{x}_s^{(t)} = 0$.
Otherwise, suppose $\|\bM^{(t)} \bx\|_2 \leq s$ and $p_t < 1$, then we have $p_t \geq \alpha \cdot \tauo^\ttop$, and
\begin{align*}
\E_{t-1}[(\overline{x}_s^{(t)})^2] = &~ p_t \cdot (\frac{1}{p_t}-1)^2 (\bx^{\top} \boldm^{(t)})^4 + (1 - p_t) \cdot (\bx^{\top} \boldm^{(t)})^4 \\
\leq &~ \frac{1}{p_t} \cdot (\bx^{\top} \boldm^{(t)})^4 
\leq \frac{\|\bM^{(t)} \bx\|_2^2 \cdot (\bx^{\top} \boldm^{(t)})^2}{\alpha} \leq \frac{s^2 \cdot (\bx^{\top} \boldm^{(t)})^2}{\alpha},
\end{align*}
where the third step follows from $p_t \geq \alpha \cdot \tauo^\ttop$ and Fact~\ref{fact:online_leverage_score}, the last step follow from the assumption of $\|\bM^{(t)} \bx\|_2 \leq s$.
Hence, we conclude that
\[
\E_{t-1}[(\overline{x}_s^{(t)})^2] \leq \frac{s^2 \cdot (\bx^{\top} \boldm^{(t)})^2}{\alpha} \cdot \mathbf{1}_{\|\bM^{(t)} \bx\|_2 \leq s},
\] 
where the indicator variable $\mathbf{1}_{\|\bM^{(t)} \bx\|_2 \leq s}$ is $1$ if $\|\bM^{(t)} \bx\|_2 \leq s$ and $0$ otherwise.

Let $t^*$ be the largest index in $[T]$ such that $\|\bM^{(t^*)} \bx\|_2 \leq s$, then we have
\begin{align*}
\sum_{t=1}^{t^{*}} \E_{t-1}[(\overline{x}_s^{(t)})^2] \leq \sum_{t=1}^{t^{*}} \frac{s^2 \cdot (\bx^{\top} \boldm^{(t)})^2}{\alpha} \cdot \mathbf{1}_{\|\bM^{(t)} \bx\|_2 \leq s} = \frac{s^2 \cdot \|\bM^{(t^*)} \bx\|_2^2}{\alpha} \leq \frac{s^4}{\alpha}.
\end{align*}
\end{itemize}

Now we can apply Freedman's inequality (Lemma~\ref{thm:freedman}) to the sequence $\overline{y}_s^{(0)}, \overline{y}_s^{(1)}, \cdots, \overline{y}_s^{(T)}$ with parameters $R = \frac{s^2}{\alpha}$, $\sigma^2 = \frac{s^4}{\alpha}$, and $u = \frac{\epsilon}{8} s^2$:
\begin{align*}
\Pr\Big[|\overline{y}_s^{(T)}| \geq  \frac{\epsilon}{8} s^2\Big] \leq &~ 2 \exp\left(-\frac{u^2/2}{\sigma^2 + Ru / 3}\right) = 2 \exp\left(-\frac{\epsilon^2 s^4/128}{s^4/\alpha + \epsilon s^4 / (24 \alpha)}\right)\\
\leq &~ 2 \exp\left(-\epsilon^2 \alpha / 200\right) \leq \delta'',
\end{align*}
where the last step follows from the choice of $\alpha = 300 d \epsilon^{-2} \log(\frac{400 \kappa d}{\epsilon \delta}) > 200 \epsilon^{-2} \log(2/\delta'')$.

Taking a union bound over $\mathcal{S}$, and since $|\mathcal{S}| = 10 \kappa$, we have that with probability at least $1 - \delta'' \cdot 10 \kappa = 1 - \delta'$, 
\begin{align}
|\overline{y}_s^{(T)}| < \frac{\epsilon}{8}\cdot s^2, \quad  \forall s \in \mathcal{S}.
\label{eq:union1}
\end{align}
Meanwhile, for any realization of $\boldm^{(1)}, \cdots, \boldm^{(T)}$, there must exist an $s^* \in \mathcal{S}$ such that 
\[
\|\bM^{(T)} \bx\|_2 \leq s^* \leq 2 \|\bM^{(T)} \bx\|_2.
\]
This is because (1) $\sigma_{\min} \cdot \|\bx\|_2 \leq \|\bM^{(T)} \bx\|_2 \leq \sigma_{\max} \cdot \|\bx\|_2$,  
and (2) the gap between two values in $\mathcal{S}$ is $\frac{\sigma_{\min}}{10} \|\bx\|_2 \leq \frac{1}{10} \|\bM^{(T)} \bx\|_2$. 

We have $\overline{x}_{s^*}^{(t)} = \bx^{\top} \bX^{(t)} \bx - (\bx^{\top} \boldm^{(t)})^2$ for all $t \in [T]$. 
Conditioned on the high probability event of Eq.~\eqref{eq:union1}, we conclude that
\[
|\bx^{\top} \bY \bx - \|\bM^{(T)} \bx\|_2^2 | = |\overline{y}_{s^*}^{(T)}| \leq \frac{\epsilon}{8} (s^*)^2 \leq \frac{\epsilon}{2} \|\bM^{(T)} \bx\|_2^2.
\]
This proves Eq.~\eqref{eq:upper-goal1}

\vspace{+2mm}
{\bf Step 2} Next we prove that $\bY$ is an $\eps$-spectral approximation to $(\bM^{(T)})^\top \bM^{(T)}$ with high probability.
Define the set 
\[
\mathcal{B} := \{(x_1, x_2, \cdots, x_d) \in [-1, 1]^d \mid \forall i, x_i = k \cdot \frac{\epsilon}{100 \kappa d} \text{ for some } k \in \mathbb{Z}\}.
\]
Note that the size of $\mathcal{B}$ is $(\frac{200 \kappa d}{\epsilon})^d$.

Taking a union bound over $\mathcal{B}$, we have that with probability at least $1 - \delta' \cdot (\frac{200 \kappa d}{\epsilon})^d = 1 - \delta$, 
\begin{align}
|\bx^{\top} \bY \bx - \|\bM^{(T)} \bx\|_2^2 | \leq \frac{\epsilon}{2} \|\bM^{(T)} \bx\|_2^2, \quad \forall \bx \in \mathcal{B}.\label{eq:adaptive-upper3}
\end{align}

We will condition on this event. For any unit vector $\bx^{*} \in \R^{d}$, there must exist some $\bx \in \mathcal{B}$ such that $|x_i - x^*_i| \leq \frac{\epsilon}{100 \kappa d}$ for all $i \in [d]$, and therefore $\|\bx - \bx^*\|_2 \leq \frac{\epsilon}{100 \kappa}$.  
Now, we have
\begin{align*}
|\|\bM^{(T)} \bx\|_2 - \|\bM^{(T)} \bx^*\|_2| \leq &~ \|\bM^{(T)} (\bx - \bx^*)\|_2 \\
\leq &~ \sigma_{\max} \cdot \frac{\epsilon}{100 \kappa} \leq \frac{\epsilon}{100} \cdot \|\bM^{(T)} \bx^{*}\|_2,
\end{align*}
where the second step follows from the largest singular value of $\bM^{(T)}$ is at most $\sigma_{\max}$ and $\|\bx - \bx^*\|_2 \leq \frac{\epsilon}{100 \kappa}$, the last step follows from the least singular value of $\bM^{(T)}$ is at least $\sigma_{\min}$ and $\kappa = \frac{\sigma_{\max}}{\sigma_{\min}}$. We then have 
\begin{align}
|\|\bM^{(T)} \bx\|_2^2 - \|\bM^{(T)} \bx^*\|_2^2| = &~ (\|\bM^{(T)} \bx\|_2 + \|\bM^{(T)} \bx^*\|_2) \cdot |\|\bM^{(T)} \bx\|_2 - \|\bM^{(T)} \bx^*\|_2| \notag \\
\leq &~ \frac{\epsilon}{40} \|\bM^{(T)} \bx^{*}\|_2^2.\label{eq:adaptive-upper1}
\end{align}

Similarly, we have 
\[
|\|\bY^{1/2} \bx\|_2 - \|\bY^{1/2} \bx^*\|_2| \leq \frac{\epsilon}{50} \cdot \|\bM^{(T)} \bx^{*}\|_2,
\]
this comes from the fact that the largest singular value of $\bY$ is at most $4 \sigma_{\max}^2$ (see Claim \ref{claim:upper-tech1}). Hence, 
\begin{align}
\label{eq:adaptive-upper2}
|\bx^{\top} \bY \bx - (\bx^*)^{\top} \bY \bx^*| \leq \frac{\epsilon}{20} \cdot \|\bM^{(T)} \bx^{*}\|_2^2.
\end{align}
Combining Eq.~\eqref{eq:adaptive-upper3}\eqref{eq:adaptive-upper1}\eqref{eq:adaptive-upper2}, we conclude with
\[
|\bx^{\top} \bY \bx - \|\bM^{(T)} \bx\|_2^2 | \leq \epsilon \|\bM^{(T)} \bx\|_2^2.
\]
This implies that $\bY$ is an $\epsilon$-spectral approximation of $(\bM^{(T)})^{\top} \bM^{(T)}$. We conclude the proof here.
\end{proof}

We next prove JL estimates are within $(1\pm 0.01)$ of the online leverage scores, even under the adversarial input sequence.

\begin{lemma}[Robustness of JL estimation]\label{lem:JL_estimates}
With probability at least $1-\delta/T$, the JL estimates in Algorithm~\ref{algo:sample} satisfy
\[
\|\tilde{\bB}^{(t-1)} \cdot \boldm^{(t)}\|_2 = (1 \pm 0.01) \|\bB^{(t-1)} \cdot \boldm^{(t)}\|_2, ~~ \forall t \in [T]
\]
under adversarial input sequence.
\end{lemma}
\begin{proof}
We can assume the adversary is deterministic by using Yao's minimax principle. That is to say, the row $\boldm^{(t)} = (\ba^{(t)}, \beta^{(t)})$ at the $t$-th iteration is fixed once the previous rows $\boldm^{(1)}, \cdots, \boldm^{(t-1)}$ and the previous outputs $\bx^{(1)}, \cdots, \bx^{(t-1)}$ are given.

Let $t^*_k \in [T]$ be the time that our data structure samples and keeps the $k$-th row. 
Our algorithm uses a new JL matrix at the end of $t^*_k$-th iteration (Line \ref{line:new-JL} of Algorithm \ref{algo:update-member}), after the algorithm outputs the solution $\bx^{(t^*_k)}$. 
With a slight abuse of notation, we denote the JL matrix used in iterations $t \in [t^*_k+1: t^*_{k+1}]$ as $\bJ^{(k)}$.  Our goal is to prove that $\bJ^{(k)}$ ensures
\[
\|\tilde{\bB}^{(t-1)}\boldm^\ttop\|_2 =  \|\bJ^{(k)}\bB^{(t-1)}\boldm^\ttop\|_2 = (1\pm 0.01)\|\bB^\ttop\boldm^\ttop\|_2, \quad \forall t \in [t^*_k+1, t^*_{k+1}]
\]
with probability at least $1- \delta/T^2$.

Let $\mathcal{A}_{\text{JL}}$ denote our data structure. 
For each $k$, define a ``hybrid'' algorithm $\mathcal{A}_{\text{exact}}^{(k)}$ that is the same as $\mathcal{A}_{\text{JL}}$ in iterations $1, 2, \cdots, t^*_k$, but starting from the $(t^*_k + 1)$-th iteration, $\mathcal{A}_{\text{exact}}^{(k)}$ ignores the JL matrix and set $\tau^\ttop_{\text{exact}} = 0.9 \|\bB^\ttop \boldm^\ttop\|_2$ in replace of Line \ref{line:approximate-ols} of Algorithm \ref{algo:sample}. 
The two algorithms $\mathcal{A}_{\text{JL}}$ and $\mathcal{A}_{\text{exact}}^{(k)}$ start from the same status in the $(t^*_k+1)$-th iteration. 
We also let them use the same random bits to perform the sampling step (Line~\ref{line:sample} in Algorithm~\ref{algo:sample})  
i.e., they share a uniformly random number $u \in [0,1]$ and each set $\nu = 1/\sqrt{p}$ if $u > p$ and $\nu \leftarrow 0$ otherwise, though their sampling probabilities $p$ are different.

We prove that the outputs of $\mathcal{A}_{\text{JL}}$ and $\mathcal{A}_{\text{exact}}^{(k)}$ are the same in iterations $t \in [t^*_k+1, t^*_{k+1}-1]$ with probability $1 - (t^*_{k+1}-t^*_k-1) \cdot \delta/T^2$. 
We prove by induction. In the base case of $(t^*_k+1)$-th iteration, $\bJ^{(k)}$ is a new random matrix and its entries are independent of all previous rows and outputs, hence the entries of $\bJ^{(k)}$ are also independent of $\bB^{(t^*_k)}$ and $\boldm^{(t^*_k+1)}$, 
by Lemma \ref{lem:JL}, with probability $1 - \delta/T^2$, 
\[
 \|\bJ^{(k)} \bB^{(t^*_k)} \boldm^{(t^*_k+1)}\|_2^2 \in (1\pm 0.01) \|\bB^{(t^*_k)} \boldm^{(t^*_k+1)}\|_2^2, 
\]
Therefore, we have $\tau^{(t_k^*)} \geq \tau^{(t_k^{*})}_{\text{exact}}$ and the sampling probability of $\mathcal{A}_{\text{JL}}$ is at least that of $\mathcal{A}_{\text{exact}}^{(k)}$. Since we use the same random bit for $\mathcal{A}_{\text{exact}}^{(k)}$ and $\mathcal{A}_{\text{JL}}$, this means if the $(t^*_k+1)$-th row is sampled in $\mathcal{A}_{\text{exact}}^{(k)}$, then it's also sampled in $\mathcal{A}_{\text{JL}}$, but we know that a new row is only sampled in $\mathcal{A}_{\text{JL}}$ until the $(t^*_{k+1})$-th iteration, so neither $\mathcal{A}_{\text{JL}}$ nor $\mathcal{A}_{\text{exact}}^{(k)}$ samples the $(t^*_k+1)$-th row. 

The induction step is similar and suppose the outputs of $\mathcal{A}_{\text{JL}}$ and $\mathcal{A}_{\text{exact}}^{(k)}$ are the same up to iteration $(t-1) \in [t_k^{*}: t_{k+1}^{*}-1]$, for the $t$-th iteration, we have that
\begin{itemize}
\item Since the $t_{k}^{*} + 1, \ldots, (t-1)$-th rows are not sampled, the matrix $\bB^{(t-1)} = \bB^{(t^*_k)}$ is not updated, so $\bB^{(t-1)}$ remains independent of the JL matrix $\bJ^{(k)}$.
\item Since $\mathcal{A}_{\text{JL}}$ and $\mathcal{A}_{\text{exact}}^{(k)}$ output the same vector to the adversary, it means the next query $\boldm^{(t)}$ chosen by the adversary is the same for both algorithms. The next query $\boldm^{(t)}$ for $\mathcal{A}_{\text{exact}}^{(k)}$ is fixed given the transcript of $\mathcal{A}_{\text{exact}}^{(k)}$, and since $\mathcal{A}_{\text{exact}}^{(k)}$ is completely agnostic of the JL matrix $\bJ^{(k)}$, this means $\boldm^{(t)}$ is also independent of $\bJ^{(k)}$. 
\end{itemize}
Combining the above two facts and by the property of JL sketch (by Lemma \ref{lem:JL}),  with probability $1 - \delta/T^2$, 
\[
 \|\bJ^{(t)} \bB^{(t-1)} \boldm^{(t)}\|_2^2 = (1\pm 0.01) \|\bB^{(t-1)} \boldm^{(t)}\|_2^2, 
\]
and $\tau^{(t)} \geq \tau^{(t)}_{\text{exact}}$. 
By a similar argument, since both algorithm use the same random bits for sampling, $\mathcal{A}_{\text{exact}}^{(k)}$ would not keep $\boldm^{(t)}$ before the $t_{k+1}^{*}$-th iteration. We note it is possible that two algorithms $\mathcal{A}_{\text{JL}}$ and $\mathcal{A}_{\text{exact}}^{(k)}$ differs from each other at the $t_{k+1}^{*}$-th iteration, (i.e., the $t^*_{k+1}$-th row is sampled in $\mathcal{A}_{\text{JL}}$, but it may not be sampled in $\mathcal{A}_{\text{exact}}^{(k)}$), this is fine because after outputting $\bx^{(t^*_{k+1})}$, our data structure immediately update the JL matrix to $\bJ^{(k+1)}$, and we would apply the same argument for $\mathcal{A}_{\text{exact}}^{(k+1)}$. Finally, we note that in each iteration the failure probability is $\delta/T^2$, so using a union bound, the total failure probability is at most $\delta/T$. We conclude the proof here.
\end{proof}

Combining Lemma \ref{lem:intrinsic_new} and Lemma \ref{lem:JL_estimates}, we conclude with the following lemma.

\begin{lemma}[Spectral approximation, adaptive adversary]
\label{lem:spectral-approximation-robust}
The following holds against an adaptive adversary: With probability at least $1 - \delta/T$, for any $t \in [T]$,
\begin{align}
    0.9 (1-\eps) \cdot \tauo^\ttop \leq \tau^{(t)} \leq 1.1 (1+\eps) \cdot \tauo^\ttop \label{eq:spectral-robust-1}
\end{align}
and for any $t \in [0: T]$ 
\begin{align}
    (\bM^{(t)})^{\top} (\bD^{(t)})^2 \bM^{(t)} \approx_{\epsilon} (\bM^{(t)})^{\top} \bM^{(t)}.\label{eq:spectral-robust-2}
\end{align}
Here $\tauo^\ttop = (\boldm^{(t)})^{\top} ((\bM^{(t-1)})^\top \bM^{(t-1)})^{-1}\boldm^{(t)}$ is the online leverage score of the $t$-th row.
\end{lemma}
\begin{proof}
With Lemma \ref{lem:intrinsic_new} and Lemma \ref{lem:JL_estimates} in hand, the proof is similar to the oblivious case (Lemma \ref{lem:sampling_probability}). We prove Eq.~\eqref{eq:spectral-robust-1}\eqref{eq:spectral-robust-2} inductively. The base case holds trivially and suppose it continues to hold up to iteration $(t-1)$. In the $t$-th iteration, by Lemma \ref{lem:intrinsic_new}, we have
\begin{align*}
\tau^\ttop = \|\tilde{\bB}^{(t-1)} \cdot \boldm^{(t)}\|_2 = (1 \pm 0.01) \|\bB^{(t-1)} \cdot \boldm^{(t)}\|_2 = (1 \pm 0.01)(1\pm \eps) \tauo^\ttop. 
\end{align*}
Here the last step follows from the same calculation as Eq.~\eqref{eq:jl1} and the inductive hypothesis on spectral approximation. This finishes the first part of induction. For the second part (i.e., Eq.~\eqref{eq:spectral-robust-2}), it follows from the inductive hypothesis (i.e., Eq.~\eqref{eq:spectral-robust-1}) and the robustness of online leverage score sampling (Lemma \ref{lem:intrinsic_new}). 
\end{proof}

The correctness of our data structure follows directly from Lemma \ref{lem:spectral-approximation-robust} and Lemma \ref{lem:approx_l2_regression_from_spectral_approx}, we summarize below.
\begin{lemma}[Correctness of Algorithm \ref{algo:preprocess}--\ref{algo:update-member}, adaptive adversary]\label{lem:correctness_algorithm-adaptive}
With probability at least $1 - \delta/T$, in each iteration, \textsc{Insert} of Algorithm~\ref{algo:update} outputs a vector $\bx^{(t)} \in \R^d$ such that 
    \[
    \|\bA^{(t)} \bx^{(t)} - \bb^{(t)}\|_2 \leq (1+\eps) \min_{\bx \in \R^d} \|\bA^{(t)} \bx - \bb^{(t)} \|_2
    \]
against an adaptive adversary.
\end{lemma}

\subsection{Runtime analysis}
\label{sec:time}

Finally, we bound the total update time of our algorithm. 
At each iteration, if the $t$-th row is sampled, the \textsc{Insert} procedure makes a call to \textsc{UpdateMembers} (Algorithm~\ref{algo:update-member}), and it takes $\wt{O}(s^{(t)} d \log(T/\delta))$ time, where $s^{(t)} $ is the number of sampled rows. The most expensive step is (Line \ref{line:new-JL} and Line \ref{line:update-wt_B}), where we instantiate a new JL matrix. All other steps can be done in $O(d^2)$ time.
If the $t$-th row is not sampled, then the \textsc{Insert} procedure only needs to compute the approximate leverage score $\tau^{(t)}$, and it takes $\wt{O}(\log(T/\delta) \cdot \nnz(\ba^{(t)}))$ time. 

We summarize the above observations and the calculations can be found in Appendix \ref{sec:upper-app}.
\begin{lemma}[Update time]\label{lem:worst_case_query_time}
At the $t$-th iteration of the \textsc{Insert} procedure (Algorithm~\ref{algo:update}),
\begin{itemize}
    \item If the $t$-th row is not sampled, then \textsc{Insert} takes $O\big(\log(T/\delta) \cdot \nnz(\ba^{(t)})\big)$ time.
    \item If the $t$-th row is sampled, then \textsc{Insert} takes $O\big(s^{(t)}d \log(T/\delta)\big)$ time.
\end{itemize}
\end{lemma}

It is clear that we want to bound the number of sampled rows. First, we have the following bound on the sum of online leverage score.
\begin{lemma}[Sum of online leverage scores \cite{cmp20}]\label{lem:online_leverage_score_sum}
If the largest singular value of matrix $\bM^{(T)}$ is at most $\sigma_{\max}$, and the least singular value of matrix $\bM^{(0)}$ is at least $\sigma_{\min}$, then
\[
\sum_{t=1}^T \tauo^\ttop \leq O(d \log(\sigma_{\max}/ \sigma_{\min})).
\]
where $\tauo^\ttop = (\boldm^{(t)})^{\top} ((\bM^{(t-1)})^\top \bM^{(t-1)})^{-1}\boldm^{(t)}$ is the online leverage score of the $t$-th row.
\end{lemma}

Using Lemma \ref{lem:online_leverage_score_sum}, we have the following lemma on the number of sampled rows, and we again defer its proof to Appendix \ref{sec:upper-app}.
\begin{lemma}[Number of sampled rows]
\label{lem:number-row}
With probability at least $1 -\delta/T$, for oblivious adversary, the total number of sampled rows is at most
\begin{align}
O\left(d\epsilon^{-2} \log(T / \delta) \log(\sigma_{\max}/\sigma_{\min}) \right). \label{eq:number-row-oblivious}
\end{align}
For adaptive adversary, the total number of sampled rows is at most
\begin{align}
O\left(d^2 \epsilon^{-2} \log(T/\delta)\log^2(\sigma_{\max}/\sigma_{\min})\right). \label{eq:number-row-adaptive}
\end{align}
\end{lemma}

Combining Lemma \ref{lem:worst_case_query_time} and Lemma \ref{lem:number-row}, we bound the amortized update time of our data structure.
\begin{lemma}[Amortized update time]\label{lem:amortized_query_time}
With probability at least $1-\delta/T$, for oblivious adversary, 
the total running time of $\textsc{Insert}$ over $T$ iterations is at most
\[
O\big(\nnz(\bA^{(T)}) \log(T/\delta) + \epsilon^{-4} d^3 \log^2(\sigma_{\max} / \sigma_{\min}) \log^3(T / \delta)\big).
\]
for adaptive adversary, 
the total running time of $\textsc{Insert}$ over $T$ iterations is at most
\[
O\big(\nnz(\bA^{(T)}) \log(T/\delta) + \epsilon^{-4} d^5 \log^4(\sigma_{\max} / \sigma_{\min}) \log^3(T / \delta)\big).
\]
\end{lemma}

This concludes the runtime analysis and we finish the proof of Theorem \ref{thm:upper}. Finally we remark on the space usage of our data structure.

\begin{remark}[Space usage]\label{rem:space}
Since the largest matrices that the data structure maintains and updates in each iteration are $\bB^{(t)}, \bN^{(t)} \in \R^{s^{(t)} \times (d+1)}$, it is straightforward to see that the total space used by the data structure is bounded by $O(s^{(T)} \cdot d)$. Hence by Lemma \ref{lem:number-row}, the total space is bounded by $O\big(d^2 \epsilon^{-2}  \log(\sigma_{\max} / \sigma_{\min}) \log(T / \delta) \big)$ for oblivious adversary and $O\big(d^3 \epsilon^{-2}  \log^2(\sigma_{\max} / \sigma_{\min}) \log(T / \delta) \big)$ for adaptive adversary.

\end{remark}

\section{Lower bound for partially dynamic LSR}
\label{sec:hard}

When the update is incremental, we prove an $\Omega(d^{2 -o(1)})$ amortized time lower bound for any algorithm with high precision solution.

\begin{theorem}[Hardness of partially dynamic LSR with high precision, formal version of Theorem \ref{thm_exact_LB_informal}]
\label{thm:lower}
Let $d$ be a sufficiently large integer, $T = \poly(d)$ and $\eps = \frac{1}{d^8T^2}$. Let $\gamma > 0$ be any constant. Assuming the $\omv$ conjecture is true, then any dynamic algorithm that maintains an $\eps$-approximate solution of the least squares regression under incremental update requires at least $\Omega(d^{2-\gamma})$ amortized time per update.
\end{theorem}

\begin{proof}
We reduce from the problem of Lemma~\ref{lem:omv-real}.
Given a PSD matrix $\bH \in \R^{d\times d}$ for the problem of Lemma~\ref{lem:omv-real}, where $1\leq \lambda_1(\bH) \leq \lambda_{d}(\bH) \leq 3$, we compute $\bA^{\top}\bA = \bH^{-1}$.
In the least squares regression problem, we take $\bA$ to be the initial matrix, and let $\bb = \mathbf{0}_d \in \R^{d}$ be the initial label vector.
The preprocessing step takes $O(d^\omega)$ time.

In the online stage of $\omv$, let the query at the $t$-th step be $\bz^{(t)} \in \R^d$ where $\|\bz^{(t)}\|_2 \leq 1$. 
We scale the vector to construct 
\[
\ba^{(t)} = \frac{1}{d^2\sqrt{T}}\cdot \bz^{(t)} \in \R^{d},
\]
and use $(\ba^{(t)}, 1)\in \R^{d} \times \R$ as the update to the regression problem at the $t$-th step.
Let $\bx^{(t)} \in \R^d$ be the solution returned by the dynamic algorithm for least squares regression at the $t$-th step. We prove that one can answer the matrix-vector query with 
\[
\by^{(t)} = d^2\sqrt{T} (\bx^{(t)} - \bx^{(t-1)}),
\]
and we have the guarantee that
\begin{align}
\|\by^{(t)} - \bH \bz^{(t)}\|_2\leq O(1/d^2).
\label{eq:lower-goal}
\end{align}
By Lemma~\ref{lem:omv-real}, this is impossible if the $\omv$ conjecture is true.

We first introduce some notations.
Let $\bA^{(t)} \in \R^{(d+t) \times d}$ be the data matrix after the $t$-th time step, $\bb^{(t)} \in \R^{d+t}$ be the labels, and $\bH^{(t)} := ((\bA^{(t)})^{\top} \bA^{(t)})^{-1} \in \R^{d\times d}$.
For simplicity, we also define $\bA^{(0)} = \bA$, $\bb^{(0)} = \bb$, $\bH^{(0)} = \bH$.
For any $t\in [T]$, let $\bx_{\star}^{(t)}$ be the optimal solution at the $t$-th step, and it has the closed-form $\bx_{\star}^{(t)} = ((\bA^{(t)})^{\top} \bA^{(t)})^{-1} (\bA^{(t)})^{\top} \cdot \bb^{(t)}$. The proof divides into three steps:
\begin{itemize}
    \item Step 1. $\bx^{(t)}$ and $\bx_{\star}^{(t)}$ are close, i.e., $\bx^{(t)} = \bx_{\star}^{(t)} \pm O(\frac{1}{d^4\sqrt{T}})$.
    \item Step 2. $\bx^{(t)} - \bx^{(t-1)}$ recovers $\bH^{(t-1)} \ba^{(t)}$, i.e., $\bx^{(t)} - \bx^{(t-1)} = \bH^{(t-1)} \ba^{(t)} \pm O(\frac{1}{d^4\sqrt{T}})$.
    \item Step 3. $\bH^{(t-1)} \ba^{(t)}$ is close to $\bH \ba^{(t)}$, i.e., $\bH^{(t-1)} \ba^{(t)} = \bH \ba^{(t)} \pm O(\frac{1}{d^6\sqrt{T}})$.
\end{itemize}
In particular, Step 2 and 3 directly implies Eq.~\eqref{eq:lower-goal}.

We first prove a useful bound on the singular values of $\bA^{(t)}$ and $\bH^{(t)}$ for all $t \in [T]$.
First note that the assumption $1\leq \lambda_1(\bH) \leq \lambda_{d}(\bH) \leq 3$ implies that $\frac{1}{\sqrt{3}} \leq \lambda_d(\bA) \leq \lambda_1(\bA) \leq 1$.
Since $\|\ba^{(t)}\|_2 = \|\frac{1}{d^2\sqrt{T}}\cdot \bz^{(t)}\|_2 \leq \frac{1}{d^2\sqrt{T}}$ for any $t\in [T]$, we have that for any $\bx \in \R^d$ with $\|\bx\|_2 = 1$,
\begin{align*}
    \bx^{\top} (\bA^{(T)})^{\top} \bA^{(T)} \bx = &~ \bx^{\top} \bA^{\top} \bA \bx + \sum_{t=1}^T (\bx^{\top} \ba^{(t)})^2 \\
    \leq &~ \|\bx\|_2^2 + \sum_{t=1}^T \|\bx\|_2^2\cdot \|\ba^{(t)}\|_2^2 \leq 1 + \sum_{t=1}^T \frac{1}{d^4 \cdot T} \leq 2.
\end{align*}
Thus we have $\frac{1}{2} \leq \lambda_d(\bA) \leq \lambda_d(\bA^{(t)}) \leq \lambda_1(\bA^{(t)}) \leq \lambda_1(\bA^{(T)}) \leq 2$.

Therefore, the matrix $\bH^{(t)} = ((\bA^{(t)})^{\top} \bA^{(t)})^{-1}$ satisfies $\frac{1}{4} \leq \lambda_d(\bH^{(t)}) \leq \lambda_1(\bH^{(t)}) \leq 4$.

\vspace{+2mm}
{\noindent \bf Step 1 \ \ } We prove $\bx^{(t)} = \bx_{\star}^{(t)} \pm \frac{1}{d^4\sqrt{T}}$. Since at each step we add a new label $1$ to the vector $\bb$, we have $\|\bb^{(t)}\|_2 \leq \|\bb^{(T)}\|_2 = \sqrt{T}$.
Since $\bx^{(t)}$ is an $\epsilon$-approximate solution of the least squares regression problem, we have
\begin{align*}
(1+\eps)^2\|\bA^{(t)} \bx_{\star}^{(t)} - \bb^{(t)}\|_2^2 \geq &~ \|\bA^{(t)} \bx^{(t)} - \bb^{(t)}\|_2^2\\
=&~ \|\bA^{(t)} \bx_{\star}^{(t)} - \bb^{(t)}\|_2^2 + \|\bA^{(t)}(\bx^{(t)} - \bx_{\star}^{(t)})\|_2^2\\
\geq&~ \|\bA^{(t)} \bx_{\star}^{(t)} - \bb^{(t)}\|_2^2 + \frac{1}{4}\|\bx^{(t)} -\bx_{\star}^{(t)}\|_2^2.
\end{align*}
The second step follows from $\bA^{(t)} \bx_{\star}^{(t)} - \bb^{(t)} = \big( \bA^{(t)} ((\bA^{(t)})^{\top} \bA^{(t)})^{-1} (\bA^{(t)})^{\top} - \bI \big) \cdot \bb^{(t)} \in \ker[(\bA^{(t)})^{\top}]$ is orthogonal to $\bA^{(t)}(\bx^{(t)} - \bx_{\star}^{(t)}) \in \im[\bA^{(t)}]$, the third step follows from $\lambda_d(\bA^{(t)}) \geq \frac{1}{2}$.

Therefore, we conclude
\begin{align}
    \|\bx^{(t)} -\bx_{\star}^{(t)}\|_2^2 \leq 4(2\eps + \eps^2)\|\bA^{(t)} \bx_{\star}^{(t)} - \bb^{(t)}\|_2^2 \leq 12\eps \cdot \|\bb^{(t)}\|_2^2 \leq 12\eps T \notag\\
    \Rightarrow \quad \|\bx^{(t)} -\bx_{\star}^{(t)}\|_2 < 4\sqrt{\eps T} = O\Big(\frac{1}{d^4\sqrt{T}}\Big)\label{eq:lower-step1}.
\end{align}
The last step follows from $\epsilon = \frac{1}{d^8 T^2}$ in the theorem statement.

\vspace{+2mm}
{\noindent \bf Step 2 \ \ } We prove $\bx^{(t)} - \bx^{(t-1)} = \bH^{(t-1)} \ba^{(t)} \pm O(\frac{1}{d^4\sqrt{T}})$.
By the Woodbury identity, one has
\begin{align}
\bx_{\star}^{(t)} = &~ \big((\bA^{(t)})^{\top}  \bA^{(t)}\big)^{-1}(\bA^{(t)})^{\top} \bb^{(t)}\notag\\
=&~ \big((\bA^{(t-1)})^{\top} \bA^{(t-1)} +\ba^{(t)} (\ba^{(t)})^{\top}\big)^{-1} \cdot ((\bA^{(t-1)})^{\top} \bb^{(t-1)} + \ba^{(t)})\notag\\
=&~ \Big(\bH^{(t-1)} - \bH^{(t-1)} \ba^{(t)} \cdot \big(1+ (\ba^{(t)})^{\top}\bH^{(t-1)} \ba^{(t)}\big)^{-1} \cdot (\ba^{(t)})^{\top}\bH^{(t-1)} \Big) \cdot ((\bA^{(t-1)})^{\top} \bb^{(t-1)} + \ba^{(t)})\notag\\
=&~ \bx_{\star}^{(t-1)} + \bH^{(t-1)} \ba^{(t)} \Big(1 - \big(1+ (\ba^{(t)})^{\top}\bH^{(t-1)} \ba^{(t)}\big)^{-1} (\ba^{(t)})^{\top} (\bx_{\star}^{(t-1)} + \bH^{(t-1)} \ba^{(t)}) \Big)\label{eq:lower-step2}
\end{align}
The first step follows from $\bx_{\star}^{(t)}$ is the optimal solution at step $t$, and the third step follows from the Woodbury identity and $((\bA^{(t-1)})^{\top} \bA^{(t-1)})^{-1} = \bH^{(t-1)}$. We use $\bx_{\star}^{(t-1)} = \bH^{(t-1)} (\bA^{(t-1)})^{\top} \bb^{(t-1)}$ in the fourth step. 

Consequently, we have
\begin{align}
    &~ \bx^{(t)} - \bx^{(t-1)} \notag\\
    = &~ \bx_{\star}^{(t)} - \bx_{\star}^{(t-1)} \pm O\Big(\frac{1}{d^4\sqrt{T}}\Big) \notag\\
    = &~ \bH^{(t-1)} \ba^{(t)}  - \bH^{(t-1)} \ba^{(t)} \big(1+ (\ba^{(t)})^{\top}\bH^{(t-1)} \ba^{(t)}\big)^{-1} (\ba^{(t)})^{\top} (\bH^{(t-1)} \ba^{(t)}+\bx_{\star}^{(t-1)}) \pm O\Big(\frac{1}{d^4\sqrt{T}}\Big)\notag\\
    = &~ \bH^{(t-1)} \ba^{(t)} \pm O\Big(\frac{1}{d^4\sqrt{T}}\Big) \pm O\Big(\frac{1}{d^4\sqrt{T}}\Big)\notag\\
    = &~ \bH^{(t-1)} \ba^{(t)} \pm O\Big(\frac{1}{d^4\sqrt{T}}\Big). \label{eq:lower-step3}
\end{align}
The first step comes from Eq.~\eqref{eq:lower-step1}, the second step follows from Eq.~\eqref{eq:lower-step2}, the third step follows from
\begin{align*}
&~ \|\bH^{(t-1)} \ba^{(t)} \cdot (1+ (\ba^{(t)})^{\top}\bH^{(t-1)} \ba^{(t)})^{-1} \cdot (\ba^{(t)})^{\top} \cdot (\bH^{(t-1)} \ba^{(t)}+\bx_{\star}^{(t-1)})\|_2\\
\leq &~ \|\bH^{(t-1)} \| \cdot \|\ba^{(t)}\|_2 \cdot (1+ (\ba^{(t)})^{\top}\bH^{(t-1)} \ba^{(t)})^{-1} \cdot \| (\ba^{(t)})^{\top}\|_2 \cdot (\|\bH^{(t-1)} \ba^{(t)}\|_2+\|\bx_{\star}^{(t-1)}\|_2)\\
 \leq &~ 4\cdot \frac{1}{d^2\sqrt{T}}\cdot 1 \cdot \frac{1}{d^2\sqrt{T}}\cdot (\frac{4}{d^2\sqrt{T}} + 8\sqrt{T}) = O\Big(\frac{1}{d^4\sqrt{T}}\Big)
\end{align*}
where we use $\|\bH^{(t-1)} \|\leq 4$, $\|\ba^{(t)}\|_2 \leq \frac{1}{d^2\sqrt{T}}$, $(1+ (\ba^{(t)})^{\top}\bH^{(t-1)} \ba^{(t)})^{-1} \leq 1$, $\|\bH^{(t-1)} \ba^{(t)}\|_2 \leq \|\bH^{(t-1)} \|\|\ba^{(t)}\|_2 \leq \frac{4}{d^2\sqrt{T}}$, $\|(\bA^{(t)})^{\top} \|\leq 2$ and
\begin{align*}
\|\bx_{\star}^{(t-1)}\|_2 = &~ \bH^{(t-1)} (\bA^{(t-1)})^{\top} \bb^{(t-1)} \\
\leq &~ \|\bH^{(t-1)}\| \cdot \|(\bA^{(t-1)})^{\top} \| \cdot \|\bb^{(t-1)}\|_2
\leq 4 \cdot 2 \cdot \|\bb^{(t-1)}\|_2\leq 8\sqrt{T}. 
\end{align*}

{\noindent \bf Step 3 \ \ } We prove $\bH^{(t-1)} \ba^{(t)} = \bH \ba^{(t)} \pm O(\frac{1}{d^4\sqrt{T}})$. Denote $\bU = [\ba^{(1)}, \ldots, \ba^{(t-1)}] \in \R^{d\times t}$, we have $\|\bU\| = \|\bU^{\top}\| \leq \|\bU\|_{F} \leq \frac{1}{d^2}$ since $\|\ba^{(i)}\|_2 \leq \frac{1}{d^2 \sqrt{T}}$ for all $i \in [t-1]$. Then we have that 
\begin{align}
 \bH^{(t-1)} \ba^{(t)} 
= &~ (\bA^{\top}\bA + \bU \bU^{\top})^{-1}\ba^{(t)} \notag\\
= &~ (\bH - \bH \bU(\bI + \bU^{\top}\bH \bU)^{-1}\bU^{\top}\bH) \cdot \ba^{(t)} \notag\\
= &~ \bH \ba^{(t)} - \bH \bU(\bI + \bU^{\top}\bH \bU)^{-1}\bU^{\top}\bH \ba^{(t)} \notag\\
= &~ \bH \ba^{(t)} \pm O\Big(\frac{1}{d^6\sqrt{T}}\Big) .\label{eq:lower-step4}
\end{align}
The second step follows from the Woodbury identity and $\bH = (\bA^{\top}\bA)^{-1}$.
The fourth step follows from
\begin{align*}
\| \bH \bU(\bI + \bU^{\top}\bH \bU)^{-1}\bU^{\top}\bH \ba^{(t)} \|_2 \leq &~ \|\bH\| \cdot \|\bU\| \cdot \|(\bI + \bU^{\top}\bH \bU)^{-1}\| \cdot \|\bU^{\top}\| \cdot \|\bH\| \cdot \|\ba^{(t)}\|_2\\
\leq &~ 4 \cdot \frac{1}{d^2} \cdot 1 \cdot\frac{1}{d^2}\cdot 4\cdot \frac{1}{d^2\sqrt{T}} = \frac{16}{d^6\sqrt{T}},
\end{align*}
as $\|\bH\| \leq 4$, $\|\bU\| = \|\bU^{\top}\| \leq \|\bU\|_{F} \leq \frac{1}{d^2}$, $\|(\bI + \bU^{\top}\bH \bU)^{-1}\| \leq 1$, and $\|\ba^{(t)}\|_2 \leq \frac{1}{d^2 \sqrt{T}}$.

\vspace{+2mm}
{\noindent \bf Combining three steps \ \ } We conclude that
\begin{align*}
    \|d^2\sqrt{T}(\bx^{(t)} - \bx^{(t-1)}) - \bH \bz^{(t)}\|_2 = d^2\sqrt{T} \|(\bx^{(t)} - \bx^{(t-1)}) - \bH \ba^{(t)}\|_2 = O(1/d^2)
\end{align*}
where the first step follows from $\ba^{(t)} = \frac{1}{d^2\sqrt{T}}\cdot \bz^{(t)}$, and the second step follows from Eq.~\eqref{eq:lower-step3} and Eq.~\eqref{eq:lower-step4}.
Hence one can recover the matrix-vector query from solutions of dynamic least squares regression. 
On the other side, the reduction only takes $O(d)$ time per update.
Hence, we have shown an $\Omega(d^{2-\gamma})$ lower bound on the amortized running time for incremental least squares regression problem under the $\omv$ conjecture.
\end{proof}

\section*{Acknowledgement}
The authors would like to thank Jan van den Brand, David Woodruff, Fred Zhang, Qiuyi Zhang, Joel Tropp for useful discussion over the project. In particular, the authors would like to thank David Woodruff for discussion on the size of JL sketch, thank Jan van den Brand for discussion over the robustness of the JL trick , and thank Joel Tropp for discussion on the matrix Chernoff bound.

Shunhua Jiang is supported by NSF CAREER award CCF-1844887 and Google PhD fellowship. Binghui Peng is supported by NSF CCF-1703925, IIS-1838154, CCF-2106429, CCF-2107187, CCF-1763970, CCF-2212233. Omri Weinstein is supported by NSF CAREER award CCF-1844887, ERC Starting grant  101039914, and ISF grant 3011005535.

\bibliographystyle{alpha}
\bibliography{ref}

\newcommand{\etalchar}[1]{$^{#1}$}
\begin{thebibliography}{BEJWY22}

\bibitem[AC06]{ac06}
Nir Ailon and Bernard Chazelle.
\newblock Approximate nearest neighbors and the fast johnson-lindenstrauss
  transform.
\newblock In {\em Proceedings of the thirty-eighth annual ACM symposium on
  Theory of computing}, pages 557--563, 2006.

\bibitem[ACSS20]{acss20}
Josh Alman, Timothy Chu, Aaron Schild, and Zhao Song.
\newblock Algorithms and hardness for linear algebra on geometric graphs.
\newblock In {\em 2020 IEEE 61st Annual Symposium on Foundations of Computer
  Science (FOCS)}, pages 541--552. IEEE, 2020.

\bibitem[ACW17]{acw17}
Haim Avron, Kenneth~L Clarkson, and David~P Woodruff.
\newblock Faster kernel ridge regression using sketching and preconditioning.
\newblock {\em SIAM Journal on Matrix Analysis and Applications},
  38(4):1116--1138, 2017.

\bibitem[AGGS22]{aggs22}
Vahid~R Asadi, Alexander Golovnev, Tom Gur, and Igor Shinkar.
\newblock Worst-case to average-case reductions via additive combinatorics.
\newblock In {\em Proceedings of the 54th Annual ACM SIGACT Symposium on Theory
  of Computing}, pages 1566--1574, 2022.

\bibitem[AHK12]{AHK06}
Sanjeev Arora, Elad Hazan, and Satyen Kale.
\newblock The multiplicative weights update method: a meta-algorithm and
  applications.
\newblock {\em Theory of Computing}, 8(6):121--164, 2012.

\bibitem[AKPS19]{akps19}
Deeksha Adil, Rasmus Kyng, Richard Peng, and Sushant Sachdeva.
\newblock Iterative refinement for $\ell_p$-norm regression.
\newblock In {\em Proceedings of the Thirtieth Annual ACM-SIAM Symposium on
  Discrete Algorithms}, pages 1405--1424. SIAM, 2019.

\bibitem[AW21]{aw21}
Josh Alman and Virginia~Vassilevska Williams.
\newblock A refined laser method and faster matrix multiplication.
\newblock In {\em Proceedings of the 2021 ACM-SIAM Symposium on Discrete
  Algorithms (SODA)}, pages 522--539. SIAM, 2021.

\bibitem[BCIS18]{bcis18}
Arturs Backurs, Moses Charikar, Piotr Indyk, and Paris Siminelakis.
\newblock Efficient density evaluation for smooth kernels.
\newblock In {\em 2018 IEEE 59th Annual Symposium on Foundations of Computer
  Science (FOCS)}, pages 615--626. IEEE, 2018.

\bibitem[BDM{\etalchar{+}}20]{bdm+20}
Vladimir Braverman, Petros Drineas, Cameron Musco, Christopher Musco, Jalaj
  Upadhyay, David~P Woodruff, and Samson Zhou.
\newblock Near optimal linear algebra in the online and sliding window models.
\newblock In {\em 2020 IEEE 61st Annual Symposium on Foundations of Computer
  Science (FOCS)}, pages 517--528. IEEE, 2020.

\bibitem[BEJS21]{workshop2021}
Omri Ben-Eliezer, Rajesh Jayaram, and Uri Stemmer.
\newblock Stoc 2021 workshop: Robust streaming, sketching, and sampling, 2021.

\bibitem[BEJWY22]{bjwy22}
Omri Ben-Eliezer, Rajesh Jayaram, David~P Woodruff, and Eylon Yogev.
\newblock A framework for adversarially robust streaming algorithms.
\newblock {\em ACM Journal of the ACM (JACM)}, 69(2):1--33, 2022.

\bibitem[BHM{\etalchar{+}}21]{bhm+21}
Vladimir Braverman, Avinatan Hassidim, Yossi Matias, Mariano Schain, Sandeep
  Silwal, and Samson Zhou.
\newblock Adversarial robustness of streaming algorithms through importance
  sampling.
\newblock {\em Advances in Neural Information Processing Systems},
  34:3544--3557, 2021.

\bibitem[BIS17]{bis17}
Arturs Backurs, Piotr Indyk, and Ludwig Schmidt.
\newblock On the fine-grained complexity of empirical risk minimization: Kernel
  methods and neural networks.
\newblock {\em Advances in Neural Information Processing Systems}, 30, 2017.

\bibitem[BKS17]{bks17}
Christoph Berkholz, Jens Keppeler, and Nicole Schweikardt.
\newblock Answering conjunctive queries under updates.
\newblock In {\em proceedings of the 36th ACM SIGMOD-SIGACT-SIGAI symposium on
  Principles of database systems}, pages 303--318, 2017.

\bibitem[BLL{\etalchar{+}}21]{bll+21}
van den~Jan Brand, Yin~Tat Lee, Yang~P Liu, Thatchaphol Saranurak, Aaron
  Sidford, Zhao Song, and Di~Wang.
\newblock Minimum cost flows, mdps, and $\ell_1$-regression in nearly linear
  time for dense instances.
\newblock In {\em Proceedings of the 53rd Annual ACM SIGACT Symposium on Theory
  of Computing (STOC)}, pages 859--869, 2021.

\bibitem[BLSS20]{blss20}
van den~Jan Brand, Yin~Tat Lee, Aaron Sidford, and Zhao Song.
\newblock Solving tall dense linear programs in nearly linear time.
\newblock In {\em Proceedings of the 52nd Annual ACM SIGACT Symposium on Theory
  of Computing (STOC)}, pages 775--788, 2020.

\bibitem[Bub15]{b15}
S{\'e}bastien Bubeck.
\newblock Convex optimization: Algorithms and complexity.
\newblock {\em Foundations and Trends in Machine Learning}, 8(3-4):231--357,
  2015.

\bibitem[Chu90]{k90}
Charles~K. Chui.
\newblock Estimation, control, and the discrete kalman filter (donald e.
  calin).
\newblock {\em SIAM Review}, 32(3):493--494, 1990.

\bibitem[CKL18]{ckl18}
Diptarka Chakraborty, Lior Kamma, and Kasper~Green Larsen.
\newblock Tight cell probe bounds for succinct boolean matrix-vector
  multiplication.
\newblock In {\em Proceedings of the 50th Annual ACM SIGACT Symposium on Theory
  of Computing}, pages 1297--1306, 2018.

\bibitem[CLM{\etalchar{+}}15]{clmmps15}
Michael~B Cohen, Yin~Tat Lee, Cameron Musco, Christopher Musco, Richard Peng,
  and Aaron Sidford.
\newblock Uniform sampling for matrix approximation.
\newblock In {\em Proceedings of the 2015 Conference on Innovations in
  Theoretical Computer Science (ITCS)}, pages 181--190. ACM, 2015.

\bibitem[CLN{\etalchar{+}}22]{cohen2022robustness}
Edith Cohen, Xin Lyu, Jelani Nelson, Tam{\'a}s Sarl{\'o}s, Moshe Shechner, and
  Uri Stemmer.
\newblock On the robustness of countsketch to adaptive inputs.
\newblock In {\em International Conference on Machine Learning}, pages
  4112--4140. PMLR, 2022.

\bibitem[CLS21]{cls21}
Michael~B Cohen, Yin~Tat Lee, and Zhao Song.
\newblock Solving linear programs in the current matrix multiplication time.
\newblock {\em Journal of the ACM (JACM)}, 68(1):1--39, 2021.

\bibitem[CMP20]{cmp20}
Michael~B Cohen, Cameron Musco, and Jakub Pachocki.
\newblock Online row sampling.
\newblock {\em Theory of Computing}, 16(15):1--25, 2020.

\bibitem[CPP22]{cpp22}
Xi~Chen, Christos Papadimitriou, and Binghui Peng.
\newblock Memory bounds for continual learning.
\newblock In {\em 2022 IEEE 63rd Annual Symposium on Foundations of Computer
  Science (FOCS)}, pages 519--530. IEEE, 2022.

\bibitem[CS17]{cs17}
Moses Charikar and Paris Siminelakis.
\newblock Hashing-based-estimators for kernel density in high dimensions.
\newblock In {\em 2017 IEEE 58th Annual Symposium on Foundations of Computer
  Science (FOCS)}, pages 1032--1043. IEEE, 2017.

\bibitem[CSWZ23]{cherapanamjeri2023optimal}
Yeshwanth Cherapanamjeri, Sandeep Silwal, David~P Woodruff, and Samson Zhou.
\newblock Optimal algorithms for linear algebra in the current matrix
  multiplication time.
\newblock In {\em Proceedings of the 2023 Annual ACM-SIAM Symposium on Discrete
  Algorithms (SODA)}, pages 4026--4049. SIAM, 2023.

\bibitem[CV95]{cortes1995support}
Corinna Cortes and Vladimir Vapnik.
\newblock Support-vector networks.
\newblock {\em Machine learning}, 20:273--297, 1995.

\bibitem[CW09]{cw09}
Kenneth~L Clarkson and David~P Woodruff.
\newblock Numerical linear algebra in the streaming model.
\newblock In {\em Proceedings of the forty-first annual ACM symposium on Theory
  of computing}, pages 205--214, 2009.

\bibitem[CW17]{cw17}
Kenneth~L Clarkson and David~P Woodruff.
\newblock Low-rank approximation and regression in input sparsity time.
\newblock {\em Journal of the ACM (JACM)}, 63(6):1--45, 2017.

\bibitem[Dah16]{d16}
S{\o}ren Dahlgaard.
\newblock On the hardness of partially dynamic graph problems and connections
  to diameter.
\newblock In {\em 43rd International Colloquium on Automata, Languages, and
  Programming (ICALP 2016)}. Schloss Dagstuhl-Leibniz-Zentrum fuer Informatik,
  2016.

\bibitem[DL19]{dl19}
Edgar Dobriban and Sifan Liu.
\newblock Asymptotics for sketching in least squares.
\newblock In {\em Proceedings of the 33rd International Conference on Neural
  Information Processing Systems}, pages 3675--3685, 2019.

\bibitem[Fre75]{freedman1975tail}
David~A Freedman.
\newblock On tail probabilities for martingales.
\newblock {\em the Annals of Probability}, pages 100--118, 1975.

\bibitem[Haz19]{h19}
Elad Hazan.
\newblock Introduction to online convex optimization.
\newblock {\em arXiv preprint arXiv:1909.05207}, 2019.

\bibitem[HFT01]{HFT_book_01}
Trevor Hastie, Jerome~H. Friedman, and Robert Tibshirani.
\newblock {\em The Elements of Statistical Learning: Data Mining, Inference,
  and Prediction}.
\newblock Springer Series in Statistics. Springer, 2001.

\bibitem[HKM{\etalchar{+}}22]{hkm+22}
Avinatan Hassidim, Haim Kaplan, Yishay Mansour, Yossi Matias, and Uri Stemmer.
\newblock Adversarially robust streaming algorithms via differential privacy.
\newblock {\em Journal of the ACM}, 69(6):1--14, 2022.

\bibitem[HKNS15]{hkns}
Monika Henzinger, Sebastian Krinninger, Danupon Nanongkai, and Thatchaphol
  Saranurak.
\newblock Unifying and strengthening hardness for dynamic problems via the
  online matrix-vector multiplication conjecture.
\newblock In {\em Proceedings of the forty-seventh annual ACM symposium on
  Theory of computing}, pages 21--30, 2015.

\bibitem[HS{\etalchar{+}}52]{hestenes1952methods}
Magnus~R Hestenes, Eduard Stiefel, et~al.
\newblock Methods of conjugate gradients for solving linear systems.
\newblock {\em Journal of research of the National Bureau of Standards},
  49(6):409--436, 1952.

\bibitem[HS22]{hs22}
Shuichi Hirahara and Nobutaka Shimizu.
\newblock Hardness self-amplification from feasible hard-core sets.
\newblock In {\em 2022 IEEE 63rd Annual Symposium on Foundations of Computer
  Science (FOCS)}, pages 543--554. IEEE, 2022.

\bibitem[HW13]{hw13}
Moritz Hardt and David~P Woodruff.
\newblock How robust are linear sketches to adaptive inputs?
\newblock In {\em Proceedings of the forty-fifth annual ACM symposium on Theory
  of computing}, pages 121--130, 2013.

\bibitem[IP01]{ip01}
Russell Impagliazzo and Ramamohan Paturi.
\newblock On the complexity of k-sat.
\newblock {\em Journal of Computer and System Sciences}, 62(2):367--375, 2001.

\bibitem[JL84]{jl84}
William~B Johnson and Joram Lindenstrauss.
\newblock Extensions of lipschitz mappings into a hilbert space.
\newblock {\em Contemporary mathematics}, 26(189-206):1, 1984.

\bibitem[JSWZ21]{jswz21}
Shunhua Jiang, Zhao Song, Omri Weinstein, and Hengjie Zhang.
\newblock A faster algorithm for solving general lps.
\newblock In {\em Proceedings of the 53rd Annual ACM SIGACT Symposium on Theory
  of Computing (STOC)}, pages 823--832, 2021.

\bibitem[JX22]{jx22}
Ce~Jin and Yinzhan Xu.
\newblock Tight dynamic problem lower bounds from generalized bmm and omv.
\newblock In {\em Proceedings of the 54th Annual ACM SIGACT Symposium on Theory
  of Computing}, pages 1515--1528, 2022.

\bibitem[Kal60]{k60}
Rudolph~Emil Kalman.
\newblock A new approach to linear filtering and prediction problems.
\newblock {\em Journal of Basic Engineering}, 82(1):35--45, 1960.

\bibitem[Law61]{Law61}
Charles~Lawrence Lawson.
\newblock {\em Contributions to the Theory of Linear Least Maximum
  Approximation /}.
\newblock Los Angeles : (S.N.), 1961, 1961.

\bibitem[LG14]{le14}
Fran{\c{c}}ois Le~Gall.
\newblock Powers of tensors and fast matrix multiplication.
\newblock In {\em Proceedings of the 39th international symposium on symbolic
  and algebraic computation}, pages 296--303, 2014.

\bibitem[LR21]{lr21}
Joshua Lau and Angus Ritossa.
\newblock Algorithms and hardness for multidimensional range updates and
  queries.
\newblock In {\em 12th Innovations in Theoretical Computer Science Conference
  (ITCS 2021)}. Schloss Dagstuhl-Leibniz-Zentrum f{\"u}r Informatik, 2021.

\bibitem[LS14]{ls14}
Yin~Tat Lee and Aaron Sidford.
\newblock Path finding methods for linear programming: Solving linear programs
  in $\wt{O}(\sqrt{\mathrm{rank}})$ iterations and faster algorithms for
  maximum flow.
\newblock In {\em 2014 IEEE 55th Annual Symposium on Foundations of Computer
  Science}, pages 424--433. IEEE, 2014.

\bibitem[LSZ19]{lsz19}
Yin~Tat Lee, Zhao Song, and Qiuyi Zhang.
\newblock Solving empirical risk minimization in the current matrix
  multiplication time.
\newblock In {\em Annual Conference on Learning Theory (COLT)}, 2019.

\bibitem[LW17]{lw17}
Kasper~Green Larsen and Ryan Williams.
\newblock Faster online matrix-vector multiplication.
\newblock In {\em Proceedings of the Twenty-Eighth Annual ACM-SIAM Symposium on
  Discrete Algorithms}, pages 2182--2189. SIAM, 2017.

\bibitem[Mad13]{m13}
Aleksander Madry.
\newblock Navigating central path with electrical flows: From flows to
  matchings, and back.
\newblock In {\em 2013 IEEE 54th Annual Symposium on Foundations of Computer
  Science (FOCS)}, pages 253--262. IEEE, 2013.

\bibitem[MM79]{martin1979multivariate}
Nick Martin and Hermine Maes.
\newblock {\em Multivariate analysis}.
\newblock Academic press London, 1979.

\bibitem[NN13]{nn13}
Jelani Nelson and Huy~L Nguy{\^e}n.
\newblock Osnap: Faster numerical linear algebra algorithms via sparser
  subspace embeddings.
\newblock In {\em 2013 ieee 54th annual symposium on foundations of computer
  science}, pages 117--126. IEEE, 2013.

\bibitem[PKP{\etalchar{+}}19]{pkp+19}
German~I Parisi, Ronald Kemker, Jose~L Part, Christopher Kanan, and Stefan
  Wermter.
\newblock Continual lifelong learning with neural networks: A review.
\newblock {\em Neural Networks}, 113:54--71, 2019.

\bibitem[Pla50]{Pla50}
R.~L. Plackett.
\newblock Some theorems in least squares.
\newblock {\em Biometrika}, 37(1/2):149--157, 1950.

\bibitem[PW17]{pilanci2017newton}
Mert Pilanci and Martin~J Wainwright.
\newblock Newton sketch: A near linear-time optimization algorithm with
  linear-quadratic convergence.
\newblock {\em SIAM Journal on Optimization}, 27(1):205--245, 2017.

\bibitem[RG75]{rg75}
Lawrence~R Rabiner and Bernard Gold.
\newblock Theory and application of digital signal processing.
\newblock {\em Englewood Cliffs: Prentice-Hall}, 1975.

\bibitem[Sar06]{s06}
Tamas Sarlos.
\newblock Improved approximation algorithms for large matrices via random
  projections.
\newblock In {\em 2006 47th Annual IEEE Symposium on Foundations of Computer
  Science (FOCS'06)}, pages 143--152. IEEE, 2006.

\bibitem[SS11]{ss11}
Daniel~A Spielman and Nikhil Srivastava.
\newblock Graph sparsification by effective resistances.
\newblock {\em SIAM Journal on Computing}, 40(6):1913--1926, 2011.

\bibitem[Sti81]{GaussLS}
Stephen~M. Stigler.
\newblock {Gauss and the Invention of Least Squares}.
\newblock {\em The Annals of Statistics}, 9(3):465--474, 1981.

\bibitem[Str69]{s69}
Volker Strassen.
\newblock Gaussian elimination is not optimal.
\newblock {\em Numerische mathematik}, 13(4):354--356, 1969.

\bibitem[Tro11]{tropp2011user}
Joel~A Tropp.
\newblock User-friendly tail bounds for matrix martingales.
\newblock 2011.

\bibitem[vdB20]{b20}
Jan van~den Brand.
\newblock A deterministic linear program solver in current matrix
  multiplication time.
\newblock In {\em Proceedings of the Fourteenth Annual ACM-SIAM Symposium on
  Discrete Algorithms}, pages 259--278. SIAM, 2020.

\bibitem[vdBNS19]{jns19}
Jan van~den Brand, Danupon Nanongkai, and Thatchaphol Saranurak.
\newblock Dynamic matrix inverse: Improved algorithms and matching conditional
  lower bounds.
\newblock In {\em 2019 IEEE 60th Annual Symposium on Foundations of Computer
  Science (FOCS)}, pages 456--480. IEEE, 2019.

\bibitem[Wil94]{wilkinson1994rounding}
James~Hardy Wilkinson.
\newblock {\em Rounding errors in algebraic processes}.
\newblock Courier Corporation, 1994.

\bibitem[Wil18]{W18survey}
Virginia~Vassilevska Williams.
\newblock Some open problems in fine-grained complexity.
\newblock {\em SIGACT News}, 49(4):29–35, dec 2018.

\bibitem[Woo14]{w14}
David~P. Woodruff.
\newblock Sketching as a tool for numerical linear algebra.
\newblock {\em Foundations and Trends in Theoretical Computer Science},
  10(1-2):1--157, 2014.

\bibitem[Woo21]{w21}
David Woodruff.
\newblock A very sketchy talk (invited talk).
\newblock In {\em 48th International Colloquium on Automata, Languages, and
  Programming (ICALP 2021)}. Schloss Dagstuhl-Leibniz-Zentrum fuer Informatik,
  2021.

\bibitem[WW18]{WW18}
Virginia~Vassilevska Williams and R.~Ryan Williams.
\newblock Subcubic equivalences between path, matrix, and triangle problems.
\newblock {\em J. {ACM}}, 65(5):27:1--27:38, 2018.

\bibitem[WZ22]{wz22}
David~P Woodruff and Samson Zhou.
\newblock Tight bounds for adversarially robust streams and sliding windows via
  difference estimators.
\newblock In {\em 2021 IEEE 62nd Annual Symposium on Foundations of Computer
  Science (FOCS)}, pages 1183--1196. IEEE, 2022.

\end{thebibliography}

\newpage
\appendix

\section{Kalman's method}
\label{sec:exact}

We review Kalman's approach for dynamic least squares regression, which maintains an exact solution with $O(d^2)$ amortized update time per iteration.

\begin{theorem}
There is an data structure with $O(d^{\omega})$ preprocessing time and $O(d^2)$ update time that maintains an exact solution of the dynamic least squares regression problem.
\end{theorem}
\begin{proof}
For any $t\in [T]$, the data structure maintains $\bH^{(t)} := ((\bA^{\ttop})^\top \bA^{\ttop})^{-1}$ and $ \bu^{(t)} := (\bA^{\ttop})^{\top} \bb^\ttop$. The later one costs $O(d)$ time to update. By the Woodbury identity, if a new row arrives, then the former one satisfies
\begin{align*}
    \bH^{(t)} =&~ ((\bA^{(t-1)})^\top \bA^{(t-1)} + \ba^{(t)} (\ba^{\ttop})^{\top} )^{-1}\\
    =&~ \bH^{(t-1)} - \bH^{(t-1)}\ba^{(t)} \big(1 + (\ba^{(t)})^{\top} \bH^{(t-1)}\ba^{(t)} \big)^{-1} (\ba^{(t)})^{\top} \bH^{(t-1)},
\end{align*}
and it can be updated in $O(d^2)$ time. It is similar when a row is deleted.

The optimal solution $\bx^{(t)}$ at step $t$ satisfies $\bx^{(t)}= ((\bA^{(t)})^{\top} \bA^{(t)})^{-1} (\bA^{(t)})^{\top}\bb^{(t)} = \bH^{(t)}\bu^{(t)}$, and it can be computed in $O(d^2)$ time given $\bH^{(t)}$ and $\bu^{(t)}$.
\end{proof}
\section{Missing proofs from Section \ref{sec:fully}}
\label{sec:fully-app}

We first prove the $\omv$-hardness of well-conditioned PSD matrices.
\begin{proof}[Proof of Lemma \ref{lem:omv-real}]
Given a Boolean matrix $\bB \in \{0,1\}^{d\times d}$ in the OMv conjecture (Conjecture~\ref{conj:omv}), we construct a PSD matrix 
\[
\bH = \left[
\begin{matrix}
2\bI_d & \frac{1}{d}\bB\\
\frac{1}{d}\bB^{\top} & 2 \bI_d
\end{matrix}
\right] \in \R^{2d\times 2d}.
\]
We note that $\bH$ is symmetric and $1 \leq \lambda_{d}(\bH) \leq \lambda_1(\bH) \leq 3$, since for any $\bz = (\bz', \bz'')\in \R^{2n}$ with $\|\bz\|_2^2 = 1$, one has
\begin{align*}
\bz^{\top}\bH \bz = &~ 2\|\bz'\|_2^2 + 2\|\bz''\|_2^2 + \frac{1}{d}(\bz')^{\top} \bB^{\top}\bz'' + \frac{1}{d} (\bz'')^{\top} \bB \bz' \\
\leq &~ 2\|\bz'\|_2^2 + 2\|\bz''\|_2^2 + 2\|\bz'\|_2\|\bz''\|_2 \in (1, 3).
\end{align*}
The second step follows from $\bB \in \{0,1\}^{d}$, and therefore, $\|\bB \by\|_2 \leq d\|\by\|_2$ and $\|\bB^{\top}\by\|_2\leq d\|\by\|_2$ hold for any $\by \in \R^{d}$.

Given an online query $\bz \in \{0,1\}^{d}$ for $\bB$, we can assume w.l.o.g.~that $\bz \neq \mathbf{0}_d$. We construct a query vector for $\bH$ as $\ov{\bz} = (\mathbf{0}_d, \frac{\bz}{\|\bz\|_2}) \in \R^{2d}$. Clearly one has $\|\ov{\bz}\|_2 = 1$. 

We prove that one can recover $\bB \bz$ from an $O(1/d^{2})$-approximate answer to $\bH \ov{\bz}$. Let $\hat{\by} = (\by', \by'') \in \R^{2d}$ be such an approximate answer, i.e., $\|\hat{\by} - \bH \ov{\bz}\|_2\leq O(1/d^{2})$. We construct $\by \in \{0, 1\}^{d}$ such that $\forall i \in [d]$, $y_{i} = 1$ if $d \|\bz\|_2 \cdot y'_{i} \geq 0.5$ and $y_{i} = 0$ otherwise. Next we prove that $\by = \bB \bz$.
\[
\big\|(\bB \bz - d\|\bz\|_2 \cdot \by')\big\|_2 = d\|\bz\|_2\cdot\left\|\frac{1}{d}\bB\frac{\bz}{\|\bz\|_2} - \by'\right\|_2\leq d\|\bz\|_2\cdot \|\bH \ov{\bz} - \hat{\by}\|_2 \leq d \cdot \sqrt{d}\cdot O(1/d^2) < 0.5.
\]
The second step follows from $\bH \ov{\bz} = (\frac{1}{d} \bB \frac{\bz}{\|\bz\|_2}, \frac{2 \bz}{\|\bz\|_2}) \in \R^{2d}$, and in the third step we use the fact that $\|\hat{\by} - \bH \ov{\bz}\|_2\leq O(1/d^{2})$, and $\|\bz\|_2\leq \sqrt{d}$ since $\bz \in \{0,1\}^d$.

Since each entry of $\bB \bz$ is an integer, rounding $d\|\bz\|_2 \cdot \by'$ to the closest integer gives the exact solution. Hence, we have $\by = \bB \bz$, and we conclude the proof here.
\end{proof}

We next provide the proof of the technical claim.

\begin{proof}[Proof of Claim \ref{claim:decomposition}]
The first claim follows from
\begin{align*}
200\lambda d^8 \geq &~ L^\ttop(\bx^\ttop) \geq \|(\bU_\perp)^\top\bx^\ttop - \frac{1}{\sqrt{d}}\mathbf{1}_{d-d_1}\|_2^2 
\geq \|\bU_\perp (\bU_\perp)^\top\bx^\ttop - \frac{1}{\sqrt{d}} \cdot \bU_\perp \cdot \mathbf{1}_{d-d_1}\|_2^2\\
= &~ \|\bU_\perp(\bU_\perp)^\top\bx^\ttop - \bx^{*}\|_2^2.
\end{align*}
The first step holds due to Eq.~\eqref{eq:loss2}, the second step follows from the definition of $L^\ttop$ in Eq.~\eqref{eq:lsr-r}, the third step holds since $\bU_\perp$ is orthonormal, the fourth step holds due to the definition $\bx^{*} = \frac{1}{\sqrt{d}} \sum_{j=1}^{d-d_1} \bU_{\perp,j}$ in Eq.~\eqref{eq:def_x*}.

For the second claim, we have that
\begin{align*}
200\lambda d^8 \geq &~ L^\ttop(\bx^\ttop) \geq \frac{1}{100}(\langle \bz^\ttop, \bx^\ttop\rangle - 10)^2 = \frac{1}{100}(\langle \bz_{\bU_\perp}^\ttop, \bx^\ttop\rangle + \langle \bz_{ \bU}^\ttop, \bx^\ttop\rangle - 10)^2 \\
= &~ \frac{1}{100}\Big(\langle \bz_{\bU_\perp}^\ttop, \bU_\perp (\bU_\perp)^{\top}\bx^\ttop\rangle + \langle \bz_{\bU}^\ttop, \bx^\ttop \rangle - 10\Big)^2 \\
= &~ \frac{1}{100} \Big(\langle \bz_{\bU_\perp}^\ttop, \bx^{*}\rangle + \langle \bz_{\bU}^\ttop, \bx^\ttop \rangle - 10 \pm 20d^4\sqrt{\lambda} \Big)^2,
\end{align*}
where the first step again follows from Eq.~\eqref{eq:loss2}, the second step follows from the definition of $L^{(t)}$ in Eq.~\eqref{eq:lsr-r}, the third step follows from $\bz_{r}^\ttop = \bz_{ \bU}^\ttop + \bz_{\bU_{\perp}}^\ttop$, the fourth step follows from $\bz_{\bU_\perp}^\ttop$ lies in the column space of $\bU_\perp$, the fifth step follows from our first claim and $\|\bz_{\bU_\perp}^\ttop\|_2 \leq 1$.

For the last claim, we have
\begin{align*}
(1+\eps)^2 \cdot \lambda \cdot (\Delta_{t}^2 + 1) \geq &~ L^\ttop(\bx^\ttop) \geq \lambda \|\bx^\ttop\|_2^2 
\geq \lambda \|(\mathbf{I} - \bV_{t}\bV_{t}^\top)\bx^\ttop\|_2^2 + \lambda \cdot \left\langle \bx^\ttop, \frac{\bz_{\bU}^\ttop}{\|\bz_{\bU}^\ttop\|_2}\right\rangle^2 \\
\geq &~ \lambda \|(\mathbf{I} - \bV_{t}\bV_{t}^\top)\bx^\ttop\|_2^2 + \lambda \cdot \Big( \frac{10 - \langle \bz_{\bU_\perp}^\ttop,\bx^{*}\rangle - 200d^4 \sqrt{\lambda}}{\|\bz_{\bU}^\ttop\|_2} \Big)^2 \\
\geq &~ \lambda \|(\mathbf{I} - \bV_{t}\bV_{t}^\top)\bx^\ttop\|_2^2 + \lambda \cdot \Big( \frac{(10 - \langle \bz_{\bU_\perp}^\ttop,\bx^{*}\rangle) \cdot( 1 - 200d^4 \sqrt{\lambda})}{\|\bz_{\bU}^\ttop\|_2} \Big)^2 \\
= &~ \lambda \|(\mathbf{I} - \bV_{t}\bV_{t}^\top)\bx^\ttop\|_2^2 + \lambda \cdot \Delta_{t}^2 \cdot (1 - 200d^4\sqrt{\lambda})^2.
\end{align*}
Here the first step follows from $L^\ttop(\bx_{t}^{*}) \leq \lambda (\Delta_{t}^2 + 1)$ in Eq.~\eqref{eq:loss} and $\bx^\ttop$ is $\eps$-approximately optimal, the second step follows from the definition of $L^{(t)}$ in Eq.~\eqref{eq:lsr-r}, the third step follows from decomposing $\bx^\ttop$ into the component orthogonal to $\bV_{t}$ and the component in the same direction as $\bz_{\bU}^\ttop$ and ignoring the component in $\bU_\perp$, the fourth step follows from plugging in our second claim, the fifth step follows from $100 - \langle \bz_{\bU_\perp}^\ttop, \bx^{*}\rangle \geq 1$, and the last step follows from $\Delta_{t} = \frac{10 - \langle \bz_{\bU_\perp}^\ttop, \bx^{*}\rangle}{\|\bz_{\bU}^\ttop\|_2}$ in Eq.~\eqref{eq:def_Delta_rt}.

Since $\eps = 1/100$ is a constant and $\lambda = \frac{1}{d^{40}}$, and also note that $\Delta_t \geq 9$, we conclude with $\|(\mathbf{I} - \bV_{t}\bV_{t}^\top)\bx_{r}^\ttop\|_2 \leq 2 \sqrt{\eps} \cdot \Delta_{t}$ from the above calculation.
\end{proof}

\section{Missing proofs from Section \ref{sec:formulation}}
\label{sec:pre-app}
We prove the basic property of online leverage score. 
\begin{proof}[Proof of Fact \ref{fact:online_leverage_score}]
We provide a short proof.
Let $\bN =  ((\bM)^{(t-1)})^\top \bM^{(t-1)} \in \R^{(d+1)\times (d+1)}$. For any $\bx \in \R^{d+1}$, define $\by = \bN^{1/2} \bx$. We have
\begin{align*}
\bx^{\top} \boldm^\ttop (\boldm^\ttop)^{\top} \bx = &~ (\by^{\top} \bN^{-1/2} \boldm^\ttop)^2 \\
\leq &~ \|\by\|_2^2 \cdot (\boldm^\ttop)^{\top} \bN^{-1} \boldm^\ttop \\
= &~ \bx^{\top} \bN \bx \cdot (\boldm^\ttop)^{\top} \bN^{-1} \boldm^\ttop \\
= &~ \bx^{\top} \bN \bx \cdot \tauo^\ttop,
\end{align*}
where the second step follows from Cauchy-Schwarz inequality. This finishes the proof.
\end{proof}

\section{Missing proofs from Section \ref{sec:upper}}
\label{sec:upper-app}
\paragraph{Missing proofs from Section~\ref{sec:data_structure}.}
We first provide the proof of Lemma \ref{lem:close_form_formula_algorithm}. 
We will use the Woodbury identity to compute the changes of the exact solution.
\begin{fact}[Woodbury identity]\label{fac:woodbury}
Let $\bA \in \R^{n\times n}, \bC\in \R^{k\times k}, \bU\in \R^{n\times k}, \bV\in\R^{k\times n}$, one has 
\begin{align*}
(\bA + \bU\bC\bV)^{-1} = \bA^{-1} - \bA^{-1}\bU (\bC^{-1} + \bV\bA^{-1}\bU)^{-1}\bV\bA^{-1}.
\end{align*}
In particular, when $\bU, \bV$ are vectors, i.e., $\bU = \ba, V = \ba^{\top}$, and $\bc = \mathsf{1}$, one has
\begin{align*}
(\bA + \ba\ba^{\top})^{-1} = \bA^{-1} - \bA^{-1}\ba (1 + \ba^{\top}\bA^{-1}\ba)^{-1}\ba^{\top}\bA^{-1}.
\end{align*}
\end{fact}

Now we are ready to prove Lemma~\ref{lem:close_form_formula_algorithm}.
\begin{proof}[Proof of Lemma \ref{lem:close_form_formula_algorithm}]
First note that all claims of the lemma hold for $t=0$ in \textsc{Preprocess}. 

We assume that the claims hold for $t-1$, and inductively prove them for $t$. We only prove the lemma for $\bH^{(t)}$ (Part 4), $\bB^{(t)}$ (Part 5), $\bG^{(t)}$ (Part 7) and $\bu^{(t)}$ (Part 8). The rest follows directly from the algorithm description.

{ \bf Part 4 ($\bH^{(t)} = \big( (\bN^{(t)})^{\top} \bN^{(t)} \big)^{-1}$) \ \ } If $\nu^{(t)} = 0$, then $\bN^{(t)} = \bN^{(t-1)}$ and $\bH^{(t)} = \bH^{(t-1)} = \big( (\bN^{(t)})^{\top} \bN^{(t)} \big)^{-1}$.
Otherwise when $\nu^{(t)} \neq 0$, one has $\bN^{(t)} = [(\bN^{(t-1)})^{\top}, \boldm^{(t)} / \sqrt{p^{(t)}}]^{\top}$. Using the Woodbury identity, we have
\begin{align*}
    &~\big( (\bN^{(t)})^{\top} \bN^{(t)} \big)^{-1} \\
    = &~ \big( (\bN^{(t-1)})^{\top} \bN^{(t-1)} + (\boldm^{(t)})^{\top} \boldm^{(t)} / p^{(t)} \big)^{-1} \\
    = &~ \big( (\bN^{(t-1)})^{\top} \bN^{(t-1)} \big)^{-1} - \frac{\big( (\bN^{(t-1)})^{\top} \bN^{(t-1)} \big)^{-1} \boldm^{(t)} (\boldm^{(t)})^{\top} \big( (\bN^{(t-1)})^{\top} \bN^{(t-1)} \big)^{-1} / p^{(t)}}{1 + (\boldm^{(t)})^{\top} \big( (\bN^{(t-1)})^{\top} \bN^{(t-1)} \big)^{-1} \boldm^{(t)} / p^{(t)}}\\
    = &~ \bH^{(t-1)} - \frac{\bH^{(t-1)} \boldm^{(t)} (\boldm^{(t)})^{\top} \bH^{(t-1)} / p^{(t)}}{1 + (\boldm^{(t)})^{\top} \bH^{(t-1)} \boldm^{(t)} / p^{(t)}},
\end{align*}
The second term is exactly the $\Delta \bH$ term when calling {\sc UpdateMembers} (Line~\ref{line:Delta_H}). Hence, $\bH^{(t)} = \bH^{(t-1)} + \Delta \bH = \big( (\bN^{(t)})^{\top} \bN^{(t)} \big)^{-1}$.

{ \bf Part 5 ($\bB^{(t)} = \bN^{(t)} \bH^{(t)}$) \ \ } If $\nu^{(t)} = 0$, then $\bN^{(t)} = \bN^{(t-1)}$ and $\bH^{(t)} = \bH^{(t-1)}$, so $\bB^{(t)} = \bB^{(t-1)} = \bN^{(t)} \bH^{(t)}$.
Otherwise when $\nu^{(t)} \neq 0$, one has $\bN^{(t)} = [(\bN^{(t-1)})^{\top}, \boldm^{(t)} / \sqrt{p^{(t)}}]^{\top}$ and $\bH^{(t)} = \bH^{(t-1)} + \Delta \bH$. We have
\begin{align*}
    \bN^{(t)} \bH^{(t)} = &~ 
    \begin{bmatrix}
    \bN^{(t-1)} \cdot \bH^{(t)} \\
    (\boldm^{(t)})^{\top} \cdot \bH^{(t)} / \sqrt{p^{(t)}}
    \end{bmatrix} \\
    = &~ \begin{bmatrix}
    \bB^{(t-1)} + \bN^{(t-1)} \cdot \Delta \bH \\
    (\boldm^{(t)})^{\top} \cdot \bH^{(t)} / \sqrt{p^{(t)}}.
    \end{bmatrix}
\end{align*}

This is exactly what we compute in Line~\ref{line:update-B} of {\sc UpdateMembers}.

{ \bf Part 7 ($\bG^{(t)} = \big( (\bA^{(t)})^{\top} (\bD^{(t)})^2 \bA^{(t)} \big)^{-1}$) \ \ } 
The proof is analogous to that of Part 4 ($\bH^{(t)}$). 
If $\nu^{(t)} = 0$, then $\bG^{(t)} = \bG^{(t-1)} = \big( (\bA^{(t-1)})^{\top} (\bD^{(t-1)})^2 \bA^{(t-1)} \big)^{-1} = \big( (\bA^{(t)})^{\top} (\bD^{(t)})^2 \bA^{(t)} \big)^{-1}$.
On the other hand, when $\nu^{(t)} \neq 0$, using the Woodbury identity, one has
\begin{align*}
    \big( (\bA^{(t)})^{\top} (\bD^{(t)})^2 \bA^{(t)} \big)^{-1}  = &~ \big( (\bA^{(t-1)})^{\top} (\bD^{(t-1)})^2 \bA^{(t-1)} + (\ba^{(t)})^{\top}\ba^{(t)}/p^{(t)} \big)^{-1}  \\
    =&~ \bG^{(t-1)} - \frac{\bG^{(t-1)}\ba^{(t)} (\ba^{(t)})^{\top} \bG^{(t-1)} / p^{(t)}}{1 + (\ba^{(t)})^{\top} \bG^{(t-1)} \ba^{(t)} / p^{(t)}}
    = \bG^{(t)}.
\end{align*}
This is exactly what we compute in Line~\ref{line:update-G} of {\sc UpdateMembers}.

{ \bf Part 8 ($\bu^{(t)} = (\bA^{(t)})^{\top}(\bD^{(t)})^2 \bb^{(t)}$) \ \ } 
We focus on the case $\nu^{(t)} \neq 0$, and we have
\begin{align*}
    \bu^{(t)} = \bu^{(t-1)} + \beta^{(t)} \cdot \ba^{(t)} / p^{(t)} = \bA^{(t-1)}(\bD^{(t-1)})^2 \bb^{(t-1)} + (\nu^{(t)})^2 \cdot \ba^{(t)} \cdot \beta^{(t)} = \bA^{(t)}(\bD^{(t)})^2 \bb^{(t)}
\end{align*}
The first step follows from the updating rule of the data structure (Line~\ref{line:update-u} of {\sc UpdateMembers}), and the second step follows from $\bu^{(t-1)}= \bA^{(t-1)}(\bD^{(t-1)})^2 \bb^{(t-1)}$ and $\nu^{(t)} = 1/\sqrt{p^{(t)}}$, the last step follows from the definition of $\bA^{(t)}, \bb^{(t)}$ and $\bD^{(t)}$.
\end{proof}

\paragraph{Missing proofs from Section~\ref{sec:correct-oblivious}.}
We then prove Lemma \ref{lem:spectral-online-leverage-score}. We make use of the following matrix Chernoff bound for adaptive sequences.
\begin{lemma}[Matrix Chernoff: Adaptive sequence. Theorem 3.1 of \cite{tropp2011user}]
\label{lem:matrix-adaptive}
Consider a finite adapted sequence $\{\bX_k\}$ of positive-semidefinite matrices with dimension $d$, and suppose that
\[
\lambda_{\max}(\bX_k) \leq R \quad \text{almost surely}.
\]
Define the finite series 
\[
\bY := \sum_{k}\bX_k \quad \text{and}  \quad \bW:= \sum_{k}\E_{k-1} \bX_k.
\]
For all $\mu \geq 0$, 
\begin{align*}
\Pr[\lambda_{\min}(\bY) \leq (1-\eps)\mu \quad\text{and}\quad \lambda_{\min}(\bW) \geq \mu] \leq &~ d\cdot \Big[\frac{e^{-\eps}}{(1-\eps)^{1-\eps}}\Big]^{\mu/R} \text{ for } \eps \in [0, 1)\\
\Pr[\lambda_{\max}(\bY) \geq (1+\eps)\mu \quad\text{and}\quad \lambda_{\max}(\bW) \leq \mu] \leq &~ d\cdot \Big[\frac{e^{\eps}}{(1+\eps)^{1+\eps}}\Big]^{\mu/R} \text{ for } \eps \geq 0.
\end{align*}
\end{lemma}

Now we are ready to prove Lemma~\ref{lem:spectral-online-leverage-score}.
\begin{proof}[Proof of Lemma \ref{lem:spectral-online-leverage-score}]
Let 
\[
\bX_{0} = ((\bM^{(T)})^\top \bM^{(T)})^{-1/2}((\bM^{(0)})^\top \bM^{(0)})((\bM^{(T)})^\top \bM^{(T)})^{-1/2}.
\]
and
\[
\bX_{t} = \nu_{t}^2 \cdot ((\bM^{(T)})^\top \bM^{(T)})^{-1/2}\boldm^\ttop (\boldm^\ttop)^\top ((\bM^{(T)})^\top \bM^{(T)})^{-1/2}, \quad \forall t\in [T].
\]
Then $\{\bX_0\}_{t\in [T]}$ is an adaptive sequence that satisfies 
(1) $\sum_{t=1}^{T}\E_{t-1}[\bX_t] = \mathbf{I}$, (2) $\bX_{t} \preceq \frac{\eps^{2}}{3\log(d/\delta)} \mathbf{I}$.
The second property follows from 
\begin{align*}
\bX_{t} = &~ \frac{1}{p_t} ((\bM^{(T)})^\top \bM^{(T)})^{-1/2}\boldm^\ttop (\boldm^\ttop)^\top ((\bM^{(T)})^\top \bM^{(T)})^{-1/2}\\
\preceq &~ \frac{\eps^2}{3\log(d/\delta)} \cdot \frac{1}{\tauo^\ttop}  ((\bM^{(T)})^\top \bM^{(T)})^{-1/2}\boldm^\ttop (\boldm^\ttop)^\top ((\bM^{(T)})^\top \bM^{(T)})^{-1/2}\\
\preceq &~ \frac{\eps^2}{3\log(d/\delta)}  ((\bM^{(T)})^\top \bM^{(T)})^{-1/2} ((\bM^{(t-1)})^\top \bM^{(t-1)}) ((\bM^{(T)})^\top \bM^{(T)})^{-1/2}\\
\preceq &~ \frac{\eps^2}{3\log(d/\delta)} \mathbf{I}.
\end{align*}
where we assumed $p_t = 3\eps^{-2}\log(d/\delta)\cdot\tauo^\ttop < 1$ in the second step. This is wlog because we can split $\bX_t$ into smaller terms if $p_t = 1$. The third step follows from Fact \ref{fact:online_leverage_score}.

Now we can apply the matrix Chernoff bound (Lemma \ref{lem:matrix-adaptive}) with $\mu = 1$, $R = \frac{\eps^2}{3\log(d/\delta)}$, we have that with probability at least $1-\delta$, one has
\[
(1-\eps)\mathbf{I} \preceq \sum_{t=0}^{T}\bX_t \preceq (1+\eps)\mathbf{I},
\]
and this implies
\[
(1-\eps)(\bM^{(T)})^\top \bM^{(T)} \preceq (\bM^{(0)})^\top \bM^{(0)} + \sum_{t=1}^{T}\nu_t^2 \cdot \boldm^\ttop (\boldm^\ttop)^\top  \preceq (1+\eps)(\bM^{(T)})^\top \bM^{(T)}.
\]
We conclude the proof here.
\end{proof}

\paragraph{Missing proofs from Section~\ref{sec:upper_robust}.}
Next, we prove the following claim that is used in the proof of Lemma~\ref{lem:intrinsic_new}.
\begin{claim}
\label{claim:upper-tech1}
Condition on the event of Eq.~\eqref{eq:adaptive-upper3}, the largest singular values of $\bY$ is at most $4 \sigma_{\max}^2$.
\end{claim}
\begin{proof}
We prove this by contradiction. Suppose the largest singular value of $\bY$ is $\sigma^2 > 4 \sigma_{\max}^2$. Let $\bx' = \arg\max_{\bx \in \R^d, \|\bx\|_2 = 1} \bx^{\top} \bY \bx$, and it satisfies that $\bx'^{\top} \bY \bx' = \sigma^2$. There must exist some $\bx'' \in \mathcal{B}$ such that $\|\bx' - \bx''\|_2 \leq \frac{\epsilon}{100 \kappa}$, and $\bx''^{\top} \bY \bx'' \leq (1 + \frac{\epsilon}{2}) \|\bM^{(T)} \bx''\|_2^2 \leq (1 + \epsilon) \sigma_{\max}^2$. We have
\begin{align*}
(\bx' - \bx'')^{\top} \bY (\bx' - \bx'') = &~ \|\bY^{1/2} (\bx' - \bx'')\|_2^2 \\
\geq &~ (\sqrt{\bx'^{\top} \bY \bx'} - \sqrt{\bx''^{\top} \bY \bx''})^2 \\
\geq &~ (\sigma - (1 + \epsilon) \sigma_{\max})^2 \\
\geq &~ (\sigma - (1 + \epsilon) \sigma_{\max})^2 \cdot (100 \kappa / \epsilon)^2 \|\bx' - \bx''\|_2^2
\end{align*}
where the second step follows from triangle inequality of $\ell_2$ norm, the third step follows from $\bx'^{\top} \bY \bx' = \sigma^2$ and $\bx''^{\top} \bY \bx'' \leq (1 + \epsilon) \sigma_{\max}^2$, and the last step follows from $\|\bx' - \bx''\|_2 \leq \frac{\epsilon}{100 \kappa}$. Since $\sigma > 2 \sigma_{\max}$ and $\epsilon < 1/8$, we have $100 (\sigma - (1 + \epsilon) \sigma_{\max}) > \sigma$, and this contradicts with our definition that $\bx'$ is the unit vector that corresponds to the largest singular value of $\bY$.
\end{proof}

\paragraph{Missing proofs from Section~\ref{sec:time}.}
We next prove Lemma \ref{lem:worst_case_query_time}.
\begin{proof}[Proof of Lemma \ref{lem:worst_case_query_time}]
If the $t$-row is not sampled, we only need to invoke the {\sc Sample} procedure. The most time-consuming step of {\sc Sample} is to compute $\wt{\bB}^{(t-1)} \cdot \boldm^{(t)}$ when computing $\tau^{(t)}$ (Line~\ref{line:levarage_score} of Algorithm \ref{algo:sample}). Since $\wt{\bB}^{(t-1)} \in \R^{O(\log(T/\delta)) \times (d+1)}$ and $\boldm^{(t)} = [(\ba^{(t)})^{\top}, \beta^{(t)}]^{\top}$, this takes $O(\log(T/\delta) \cdot \nnz(\ba^{(t)}))$ time.

If the $t$-row is sampled, besides the {\sc Sample} procedure, the data structure also needs to invoke the {\sc UpdateMembers} procedure. 
The most time-consuming step is to compute $\wt{\bB}^{(t)} = \bJ^{(t)} \cdot \bB^{(t)}$ on Line~\ref{line:update-wt_B}. Indeed, it's easy to see that all other computations only involve matrix-vector multiplications and matrix additions, and they can be computed in $O(s^{(t)} d)$ time.
Since $\bJ^{(t)} \in \R^{O(\log(T/\delta)) \times s^{(t)}}$ and $\bB^{(t)} \in \R^{s^{(t)} \times (d+1)}$, computing $\wt{\bB}^{(t)} = \bJ^{(t)} \cdot \bB^{(t)}$ takes $O(s^{(t)} d \log(T/\delta))$ time.
\end{proof}

To prove Lemma~\ref{lem:number-row}, we make use of the following concentration result that is a direct application of Freedman's inequality.
\begin{lemma}
\label{lem:number_sampled_rows-concentration}
Let $p_1, p_2, \cdots, p_T \in [0,1]$ be a sequence of sampling probabilities chosen by adaptive adversary and always satisfies $\sum_{t=1}^T p_t \leq U$. 
Let 
\[
x^{(t)} = 
\begin{cases}
1 & \text{w.p. } p_t \\
0 & \text{w.p. } 1 - p_t
\end{cases}
\]
and let $y^{(0)} = 0$, $y^{(t)} = y^{(t-1)} + x^{(t)}$ for any $t\in [T]$. Then for any $u > 0$,  the final outcome $y^{(T)}$ satisfies
\[
\Pr\left[y^{(T)} \geq u + \sum_{t=1}^T p_t  \right] \leq \exp\left(-\frac{u^2/2}{U + u/3}\right).
\]
\end{lemma}
\begin{proof}
Let $\overline{x}^{(t)} = x^{(t)} - p_t$, and note that $\E_{t-1}[\overline{x}^{(t)}] = 0$ and $|\overline{x}^{(t)}| \leq 1$. 
Let $\overline{y}^{(0)} = 0$, and $\overline{y}^{(t)} = \overline{y}^{(t-1)} + \overline{y}^{(t)}$. Note that the sequence $\overline{y}^{(0)}, \overline{y}^{(1)}, \cdots, \overline{y}^{(T)}$ is a martingale, and $\overline{y}^{(T)} = y^{(T)} - \sum_{t=1}^T p_t$.
We have
\[
\E_{t-1}[(\overline{x}^{(t)})^2] = p_t \cdot (1 - p_t)^2 + (1 - p_t) \cdot p_t^2 = p_t \cdot (1 - p_t).
\]
and the variance satisfies
\[
\Var = \sum_{t=1}^T \E_{t-1}[(\overline{x}^{(t)})^2]
= \sum_{t=1}^T p_t \cdot (1 - p_t) \leq \sum_{t=1}^T p_t \leq U.
\]
Using Freedman's inequality (Lemma \ref{thm:freedman}) with $R = 1$, $\sigma^2 = U$, and any $u > 0$, we have
\[
\Pr[\overline{y}^{(T)} \geq u] \leq \exp\left(-\frac{u^2/2}{U + u/3}\right). \qedhere
\]
\end{proof}

Now we are ready to prove Lemma~\ref{lem:number-row}.
\begin{proof}[Proof of Lemma \ref{lem:number-row}]
For oblivious adversary, conditioning on the event of Lemma \ref{lem:correctness_algorithm}, the expected number of rows are at most 
\begin{align*}
\sum_{t=1}^{T}C_{\text{obl}}\cdot \tau^\ttop \leq 2\sum_{t=1}^{T}C_{\text{obl}} \cdot \tauo^\ttop = O\left(d\eps^{-2}\log(T/\delta)\log(\frac{\sigma_{\max}}{\sigma_{\min}})\right).
\end{align*}
Plugging $U = O(\eps^{-2}d\log(T/\delta)\log(\sigma_{\max} / \sigma_{\min}))$ into Lemma \ref{lem:number_sampled_rows-concentration}, we obtain Eq.~\eqref{eq:number-row-oblivious}.

For adaptive adversary, conditioning on the high probability event of Lemma \ref{lem:spectral-approximation-robust}, the expected number of rows are at most 
\begin{align*}
\sum_{t=1}^{T}C_{\text{adv}}\cdot \tau^\ttop \leq 2\sum_{t=1}^{T}C_{\text{adv}} \cdot \tauo^\ttop = O\left(d^2\eps^{-2}\log(T/\delta)\log^2(\frac{\sigma_{\max}}{\sigma_{\min}})\right).
\end{align*}
Plugging $U = O\left(d^2\eps^{-2}\log(T/\delta)\log^2(\frac{\sigma_{\max}}{\sigma_{\min}})\right)$ into Lemma \ref{lem:number_sampled_rows-concentration}, we obtain Eq.~\eqref{eq:number-row-adaptive}.
\end{proof}

\section{Empirical study}
\label{sec:exp}

As part of our program, we perform empirical evaluations over our algorithms developed for partially dynamic LSR, over both synthetic and real-world datasets. 
Our method is most suitable for input data that are non-uniform. 
Indeed, if the data has low coherence (they are all similar to each other), then the naive uniform sampling is already as good as leverage score sampling.

\vspace{+2mm}
{\noindent \bf Synthetic dataset \ \ } We follow the empirical study of \cite{dl19} and generate data from the {\em elliptical model}. In this model $\ba^{(t)} = w^{(t)}\Sigma \bz^{(t)}$, where $\bz^{(t)} \sim N(0, \bI_d)$ is a random Gaussian vector, $\Sigma \in \R^{d\times d}$ is a PSD matrix, and $w^{(t)}$ is a scalar. The label is generated as $b^{(t)} = \langle \ba^{(t)}, \bx^{\star}\rangle + w^{(t)}\xi^{(t)}$, where $\bx^{\star} \in \R^{d}$ is a hidden vector and $\xi \sim N(0,1)$ is standard Gaussian noise. This model has a long history in multivariate statistics, see e.g.~\cite{martin1979multivariate}. 
In our experiments, we set $\Sigma = \bI_d$ for simplicity. 
In order to make the dataset non-uniform, after the initial phase (the first 10\% of the data), we randomly choose $d/10$ rows among the next 10\% of data to have a large scalar of $\sqrt{T}$. The rest of the data have a scalar of $1$. 
We set $T = 500000$ and $d = 500$.

\vspace{+2mm}
{\noindent \bf Real-world dataset \ \ } We use the VirusShare dataset from the UCI Machine Learning Repository\footnote{\url{https://archive.ics.uci.edu/ml/datasets.php}}. We select this dataset because it has a large number of features and data points, and has low errors when fitted by a linear model. 
The dataset is collected from Nov 2010 to Jul 2014 by VirusShare (an online platform for malware detection).
It has $T = 177856$ data points and $d = 482$ features.

\vspace{+2mm}
{\noindent \bf Baseline algorithms \ \ } We compare with three baseline methods. 
\begin{enumerate}
\item {\em Kalman's approach} makes use of the Woodbury identity and gives an exact solution. 
\item The {\em uniform sampling} approach samples new rows uniformly at random. 
\item The {\em row sampling} approach samples new rows according to the exact online leverage scores \cite{cmp20}.
\end{enumerate}

\vspace{+2mm}
{\noindent \bf Parameters \ \ } In both our method and the row sampling method, we use an error parameter $\epsilon$. We set the sampling probability to be $p = \min\{\tau \epsilon^{-2} / 2, 1\}$ (except when $\epsilon=1$ we set $p = \min\{\tau, 1\}$ to make it non-trivial), where $\tau$ is the approximate online leverage score of different methods. 
We only implement our oblivious algorithm since both datasets do not involve adaptive adversary.
The JL matrix in our algorithm has $k$ number of rows, where we set $k = c_{\epsilon} \cdot \epsilon^{-2}$ for some constants $c_{\epsilon}$ so that $k \approx 20$.  The raw data of our experiments are shown in Table~\ref{tab:experiment}.

\begin{figure}[!htbp]
\centering
  \begin{tabular}{@{}c@{}}
    \includegraphics[width=.48\textwidth]{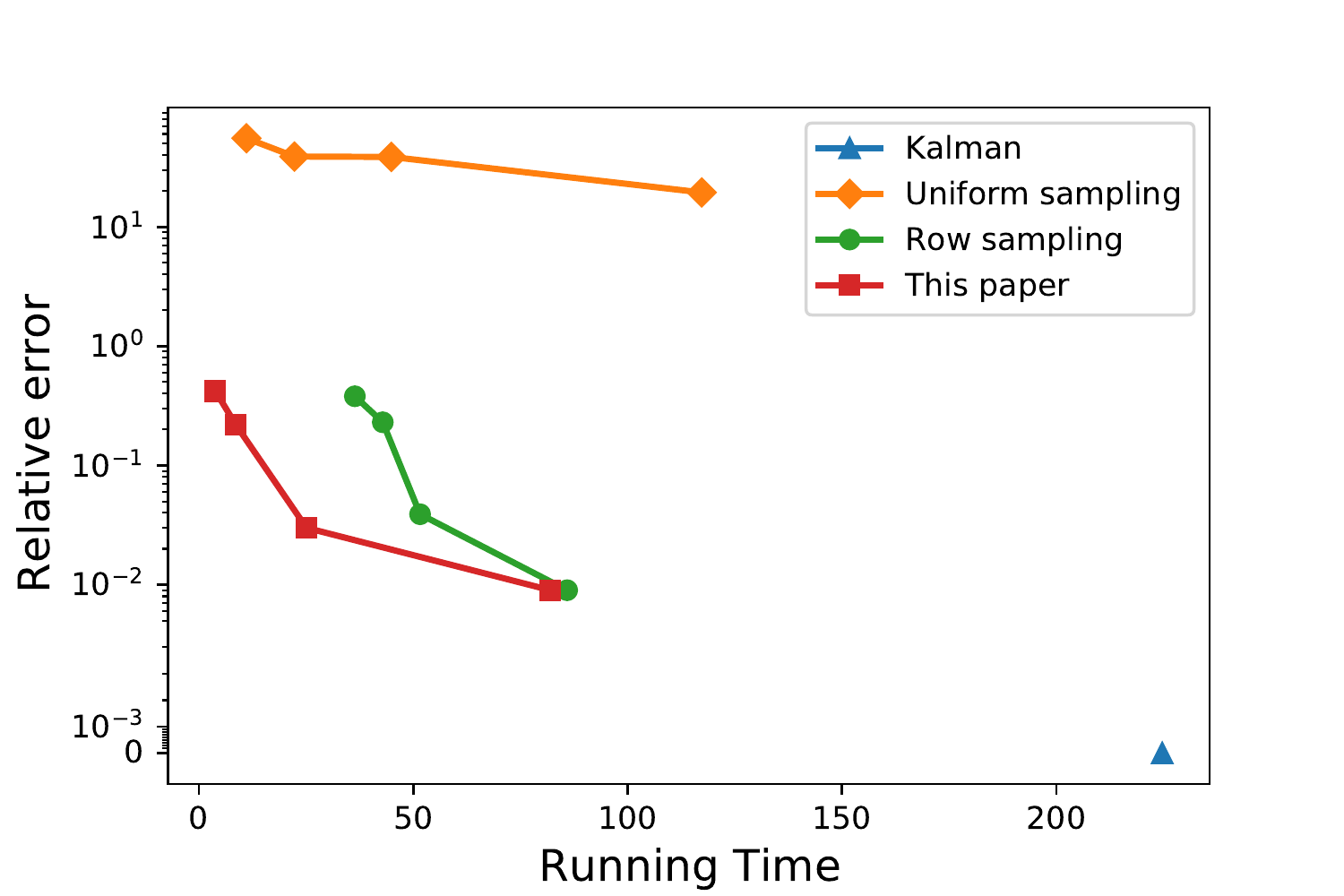}
    \\[\abovecaptionskip] \small (a) Synthetic dataset
      \end{tabular}
  \begin{tabular}{@{}c@{}}
    \includegraphics[width=.48\textwidth]{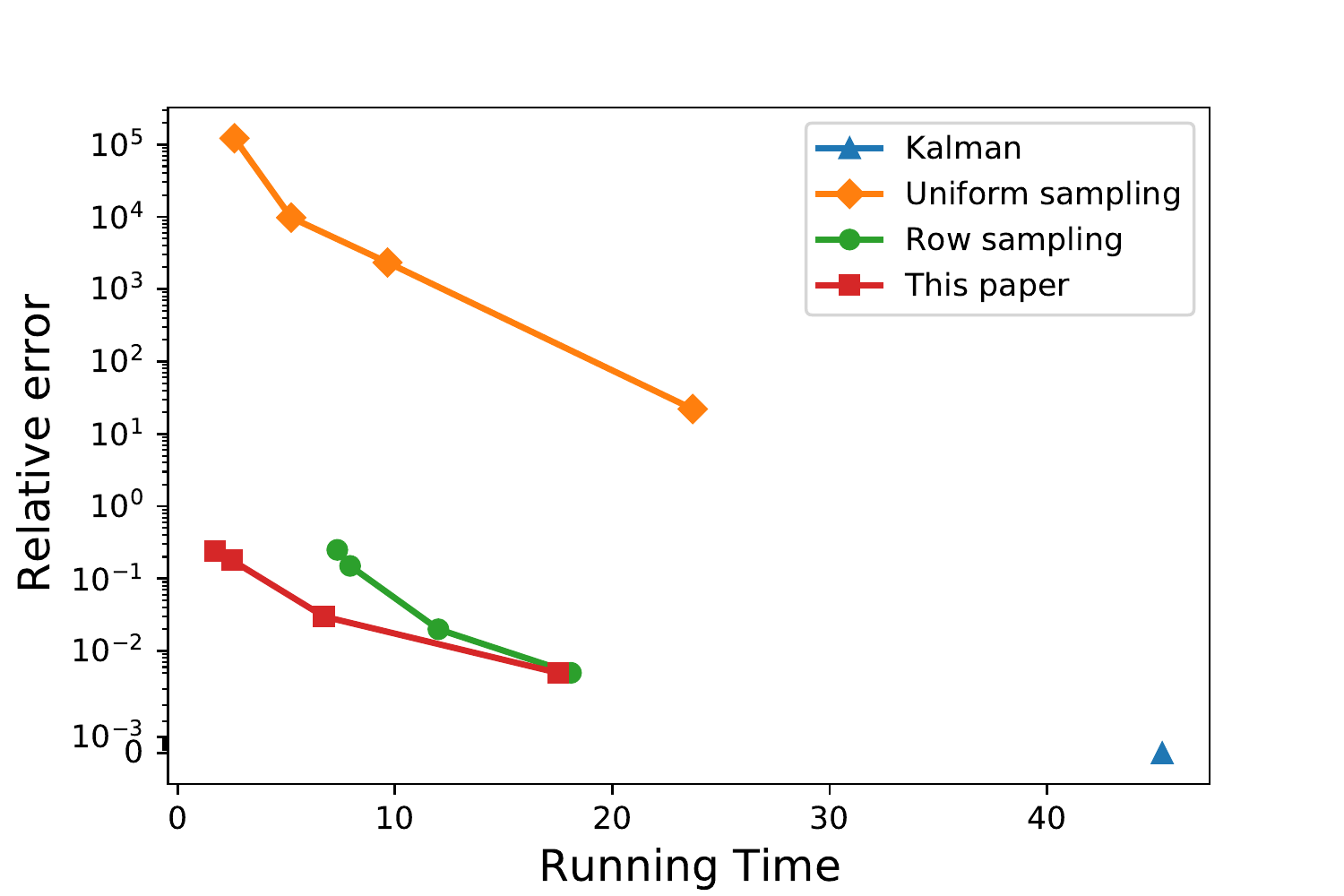}
    \\[\abovecaptionskip] \small (b) VirusShare dataset
    \end{tabular}

\caption{Experiment results. The $x$-axis shows the running time (unit: seconds), and the $y$-axis shows the relative error $(\mathrm{err}/\mathrm{err}_{\mathrm{std}} - 1)$, where $\mathrm{err}$ is the error of the particular approach, and $\mathrm{err}_{\mathrm{std}}$ is the error of the static Normal equation. 
The $y$-axis is on a symlog scale, where for range $\geq 0.005$ we show the base-$10$ log scale, and for range $[0, 0.005)$ we show the linear scale. Kalman's approach has a relative error of $0$, and except this point, all other data points are in the range of the log scale. 
For uniform sampling, we take sampling probability $p=0.05,0.1,0.2,0.5$. For row sampling and our algorithm, we take the error parameter $\eps = 0.1, 0.2, 0.5, 1$. }
\label{fig:exp1}
\end{figure}

\vspace{+2mm}
{\noindent \bf Experiment results \ \ }
Our experiments are executed on an Apple M1 CPU with codes written in MATLAB. We repeat all experiments for at least $5$ times and take the mean. 
On both datasets, we initiate the model based on the first 10\% of the data.
The experiment results are formally presented in Figure \ref{fig:exp1}.
Our algorithm consistently outperforms baseline methods: It runs faster when achieving comparable error rates.

\begin{table}[!ht]
    \centering
    \begin{tabular}{|c|c|c|c|c|c|c|}
    \hline
    Dataset    & Method & Error & Time & Parameters  \\  \hline
    Synthetic & Kalman &  1 & 224.7s &   \\ \hline
    Synthetic & ours &   1.42 & 3.85s & $\eps=1$  \\ \hline
    Synthetic & ours &   1.22 & 8.6s & $\eps=0.5$  \\ \hline
    Synthetic & ours &   1.03 & 25.1s & $\eps=0.2$  \\ \hline
    Synthetic & ours &    1.009 & 82.0s & $\eps=0.1$ \\ \hline
    Synthetic & row sampling &   1.38 & 36.39s & $\eps=1$   \\ \hline
    Synthetic & row sampling &   1.23 & 42.9s & $\eps=0.5$  \\ \hline
    Synthetic & row sampling &   1.039 & 51.6s & $\eps=0.2$  \\ \hline
    Synthetic & row sampling &   1.009& 85.9s & $\eps=0.1$  \\ \hline
    Synthetic & uniform &   56.2 & 11.1s & $p=0.05$  \\ \hline
    Synthetic & uniform &   39.8 & 22.3s & $p=0.1$  \\ \hline
    Synthetic & uniform &   39.5 & 44.9s & $p=0.2$ \\ \hline
    Synthetic & uniform &  20.4 & 117.3s & $p=0.5$ \\ 
    \Xhline{3\arrayrulewidth}
    VirusShare & Kalman & 1 & 45.3s &  \\ \hline
    VirusShare & ours & 1.24 & 1.74s & $\eps = 1$   \\ \hline
    VirusShare & ours & 1.18 & 2.50s & $\eps = 0.5$   \\ \hline
    VirusShare & ours & 1.03 & 6.73s & $\eps = 0.2$    \\ \hline
    VirusShare & ours & 1.005 & 17.5s & $\eps = 0.1$    \\ \hline
    VirusShare & row sampling & 1.25 & 7.35s & $\eps = 1$  \\ \hline
    VirusShare & row sampling & 1.15 & 7.94s & $\eps = 0.5$  \\ \hline
    VirusShare & row sampling & 1.02& 12.0s & $\eps = 0.2$  \\ \hline
    VirusShare & row sampling & 1.005 & 18.1s & $\eps = 0.1$  \\ \hline
    VirusShare & uniform sampling & 1.2153e+05 & 2.62s & $p = 0.05$  \\ \hline
    VirusShare & uniform sampling & 9.7335e+03 & 5.23s & $p = 0.1$  \\ \hline
    VirusShare & uniform sampling & 2.3363e+03 & 9.66s & $p = 0.2$  \\ \hline
    VirusShare & uniform sampling & 23.1 & 23.7s & $p = 0.5$  \\ \hline
    \end{tabular}
    \caption{Experiment results.}
    \label{tab:experiment}
\end{table}

\end{document}